\documentclass[11pt,fleqn,a4paper]{article} %,draft

\usepackage{cmap}
\usepackage[T2A]{fontenc}
\usepackage[utf8]{inputenc}

\usepackage{amsmath,amssymb,amsthm,amsfonts} % ,enumerate,backref,refcheck
\usepackage[mathscr]{eucal}
\usepackage{graphicx}

\usepackage{hyperref}
\hypersetup{colorlinks, linkcolor=blue, citecolor=blue, urlcolor=blue}

\newcommand{\p}{\partial}

\newcommand{\const}{{\rm const}}

\newlength{\mylength}
\newtheorem{theorem}{Theorem}%[section]
\newtheorem{lemma}[theorem]{Lemma}
\newtheorem{corollary}[theorem]{Corollary}

{\theoremstyle{definition}
\newtheorem{definition}[theorem]{Definition}
\newtheorem{remark}[theorem]{Remark}
\newtheorem*{notation}{Notation}
}

\newcommand{\todo}[1][\null]{\ensuremath{\clubsuit}}

\newcommand{\noprint}[1]{}

\usepackage{xcolor}

\begin{document}

\par\noindent {\LARGE\bf
Lie and point symmetries of Nyzhnyk models
\par}

\vspace{5mm}\par\noindent{\large
Oleksandra O.\ Vinnichenko$^\dag$, Vyacheslav M.\ Boyko$^{\dag\ddag}$ and Roman O.\ Popovych$^{\dag\S}$
}

\vspace{4mm}\par\noindent{\it\small
$^\dag$\,Institute of Mathematics of NAS of Ukraine, 3 Tereshchenkivska Str., 01024 Kyiv, Ukraine
\par}

\vspace{2mm}\par\noindent{\it\small
$^\ddag$\,Department of Mathematics, Kyiv Academic University, 36 Vernads'koho Blvd., 03142 Kyiv, Ukraine
\par}

\vspace{2mm}\par\noindent{\it\small
$^\S$\,Mathematical Institute, Silesian University in Opava, Na Rybn\'\i{}\v{c}ku 1, 746 01 Opava, Czech Republic
\par}

\vspace{5mm}\par\noindent{\small
E-mail:
oleksandra.vinnichenko@imath.kiev.ua,
boyko@imath.kiev.ua,
rop@imath.kiev.ua
\par}

\vspace{8mm}\par\noindent\hspace*{10mm}\parbox{140mm}{\small
We consider a hierarchy of models that can be obtained from the original Nyzhnyk system by
imposing conditions on parameters,
introducing potentials or pseudopotentials,
performing limiting processes with respect to a scaling parameter,
applying differential substitutions and
interpreting parts of independent and/or dependent variables as complex or real.
Therefore, among the Nyzhnyk models, one can distinguish
between symmetric and asymmetric, dispersive and dispersionless, standard and modified,
as well as real, complex, mixed and specific models.
One can also distinguish single partial differential equations or systems of such equations,
as well as linear or nonlinear Lax representations in the dispersive or dispersionless cases, respectively.
For each specified model, we compute the maximal Lie invariance pseudoalgebra
and, in the symmetric case, find the point- and contact-symmetry pseudogroups
using the megaideal-based version of the algebraic method.
It is shown that relations between these pseudoalgebras and between these pseudogroups
are induced by relations between the corresponding models.
Defining (finite-dimensional) subalgebras are singled out in all these pseudoalgebras in the symmetric and specific cases.
Based on the established correspondences between the dispersionless and dispersive models,
we completely classify one- and two-dimensional subalgebras of the maximal Lie invariance pseudoalgebra of the (dispersive symmetric potential) Nyzhnyk equation and one-dimensional subalgebras of the maximal Lie invariance pseudoalgebra of its linear Lax representation.
}\par\vspace{4mm}

\noprint{

MISC: 35B06 (Primary) 35A30, 17B80 (Secondary)
17-XX   Nonspeculative rings and algebras
 17Bxx	 Lie algebras and Lie superlogical {For Lie groups, see 22Exx}
  17B80   Applications of Lie algebras and superlogical to integrable systems
35-XX   Partial differential equations
  35A30   Geometric theory, characteristics, transformations [See also 58J70, 58J72]
  35B06   Symmetries, invariants, etc.
 35Cxx  Representations of solutions
  35C05   Solutions in closed form
  35C06   Self-similar solutions
}

\begin{flushright}
{\it Dedicated with great pleasure to Professor Leonid P. Nyzhnyk\\ on the occasion of his 90th birthday}
\end{flushright}

\section{Introduction}%\label{sec:Introduction}

Over the past few decades, numerous papers have been devoted to the symmetry analysis of both particular systems of differential equations important for applications and classes of such systems.
% see, for example, ~\cite{andr1998A,cham1988a,davi1986a,fush1994a,fush1994b,kont2019a,kova2023b, malt2024a,mart1989a,olve1993A,poch2017a,vane2021a}, the collection result~\cite{CRC_v2,CRC_v1,CRC_v3} and references therein for specific examples.
At the same time, the proportion of papers containing correct and complete results remains unexpectedly low.
Furthermore, in many cases, confusion arises regarding the nomenclature of the models due to an insufficient analysis of the relations between them and the history of their origin, which also leads to the redundant investigation of equivalent models.
Therefore, when performing an extended symmetry analysis of a given model, it is beneficial to consider it within a broader context.

In 1980, Leonid Nyzhnyk published the paper~\cite{nizh1980a}, demonstrating how to apply the inverse scattering method to the integration of multidimensional nonlinear equations.
It was therein~\cite[Eq.~(4)]{nizh1980a} that the system
\begin{gather}\label{eq:NyzhnykSystem}
\begin{split}&
w_t=k_1w_{xxx}+k_2w_{yyy}+3(v^1w)_x+3(v^2w)_y,\\&
v^1_y=k_1w_x, \quad v^2_x=k_2w_y,
\end{split}
\end{gather}
where $(k_1,k_2)\ne(0,0)$, appeared for the first time.
This system is a two-dimensional generalization of the Korteweg--de Vries equation and stands as one of the earliest multidimensional integrable models in the literature.
If $k_1k_2\ne0$, this generalization is spatially symmetric.
Moreover, it is the only known two-dimensional isotropic generalization of the Korteweg--de Vries equation.
The system~\eqref{eq:NyzhnykSystem} was derived in~\cite{nizh1980a}
as the compatibility condition of the linear system of two partial differential equations
\begin{gather}\label{eq:NyzhnykSystemLax}
\psi_{xy}+w\psi=0, \quad
\psi_t=k_1\psi_{xxx}+k_2\psi_{yyy}+3v^1\psi_x+3v^2\psi_y
\end{gather}
with respect to the function~$\psi$.
In other words, the system~\eqref{eq:NyzhnykSystemLax} is the (linear) Lax representation for the system~\eqref{eq:NyzhnykSystem}.
In fact, with respect to the parameters~$k_1$ and~$k_2$, the system~\eqref{eq:NyzhnykSystem} is a class of systems that is the union of two equivalence classes under scaling transformations. Namely,
in the symmetric case, where both parameters $k_1$ and $k_2$ are nonzero, $k_1k_2\ne0$, one can set $(k_1,k_2)=(1,1)$, whereas in the asymmetric case, where only one of them vanishes, one can set $(k_1,k_2)=(1,0)$.
In both cases, through the introduction of potentials, limiting processes, alternative interpretations of variables and differential substitutions, a number of other models can be derived from the system~\eqref{eq:NyzhnykSystem}: symmetric and asymmetric, dispersive and dispersionless, real, complex and of mixed type (in particular, with complex conjugate variables), standard with quadratic nonlinearities and modified with cubic ones, systems as well as single equations.
In other words, there exist some connections between the Nyzhnyk systems and equations and their linear or (for dispersionless models) nonlinear Lax representations, allowing transitions between the models.
Each of these models possesses interesting properties and some have been the subjects of intensive study within the framework of the symmetry analysis of differential equations, which, unfortunately, did not yield rigorous results.
Hence, the idea arose to carry out a systematic extended symmetry analysis of various Nyzhnyk models.
The realization of this idea was begun in a series of papers~\cite{boyk2024a,vinn2024a,vinn2026a} on the (real symmetric potential) dispersionless Nyzhnyk equation, its nonlinear Lax representation and the (real symmetric) dispersionless Nyzhnyk system.

The main purpose of this paper is to extend the results from~\cite{boyk2024a} and, partially, from~\cite{vinn2024a} to other Nyzhnyk models.
In particular, for each of the considered Nyzhnyk models, its maximal Lie invariance pseudoalgebra is computed.
If such a model is a single equation, it is additionally shown that its maximal contact invariance pseudoalgebra coincides with the first prolongation of its maximal Lie invariance pseudoalgebra.
According to the relations between the models, we demonstrate what connections exist between their pseudoalgebras.
As it turns out, the maximal Lie invariance pseudoalgebra of each dispersive model is a pseudosubalgebra of codimension one in the corresponding pseudoalgebra of its dispersionless counterpart, whose complement is associated with the scaling transformations of the spatial variables.
To obtain the Lie invariance pseudoalgebra for a Lax representation from the maximal Lie invariance pseudoalgebra of the corresponding Nyzhnyk model, one only needs to prolong the elements of the pseudoalgebra to the pseudopotential as a new dependent variable and add one more generating vector field associated with the gauge ambiguity in the definition of the pseudopotential.
The maximal Lie invariance pseudoalgebras of specific Nyzhnyk models (i.e., models whose variables include complex conjugates) are almost identical, up to the notation of variables and a slight modification of certain generating vector fields, to those of the corresponding symmetric Nyzhnyk models.

The point-symmetry pseudogroups of the (real symmetric potential) dispersionless Nyzhnyk equation, its nonlinear Lax representation, the (real symmetric) dispersionless Nyzhnyk system and the contact-symmetry pseudogroup of this equation have already been constructed in~\cite{boyk2024a}.
Therein, the defining (finite-dimensional) subalgebras within its maximal Lie and contact invariance pseudoalgebras were also singled out, while in~\cite{vinn2024a}, the one- and two-dimensional subalgebras of the former pseudoalgebra, as well as the one-dimensional subalgebras of the maximal Lie invariance pseudoalgebra of the nonlinear Lax representation of the mentioned equation, were classified.
By appropriately modifying the proofs from~\cite{boyk2024a} in accordance with the connections between the models, all the specified results are extended in this paper to other symmetric Nyzhnyk models, including specific ones.
However, they cannot be used to obtain analogous results in the asymmetric case since the structures of both the models and their maximal Lie invariance pseudoalgebras and point-symmetry pseudogroups differ substantially between the symmetric and asymmetric cases.
Moreover, up to now, no asymmetric model has been properly studied from the viewpoint of symmetry analysis.
Therefore, such models require a separate, comprehensive investigation, which has already begun in~\cite{vinn2026b} for the asymmetric (potential) Nyzhnyk equation, also referred to as the Boiti--Leon--Manna--Pempinelli equation.

The structure of the paper is as follows.
The Nyzhnyk models in all their variety are described in detail in Section~\ref{sec:NyzhnykSystemsAndEquations}.
The connections established between them allow us to consider their collection as a certain hierarchy of models arising from the original Nyzhnyk system.
In Section~\ref{sec:LieInvAlgebra}, for each of the aforementioned models (except for the generalizations of Nyzhnyk models from Section~\ref{sec:GenNyzhnykModel}), its maximal Lie invariance pseudoalgebra is computed, and when it is a single equation, its maximal contact invariance pseudoalgebra is constructed as well, creating the foundation for its further symmetry analysis.
In Section~\ref{sec:Pseudogroups}, after systematizing and deepening the results regarding the point- (or contact-) symmetry pseudogroups of certain symmetric dispersionless models from~\cite{boyk2024a}, they are used as the examples to compute such pseudogroups for all other symmetric models under consideration by means of a special version of the algebraic method based on megaideals~\cite{malt2024a}, as well as to find complete lists of independent discrete symmetries.
The construction of defining subalgebras for the Lie (or contact) invariance pseudoalgebras of the (symmetric potential) dispersionless Nyzhnyk equation in~\cite{boyk2024a} is extended in Section~\ref{sec:DefiningSubalgs} to its nonlinear Lax representation, the (symmetric) dispersionless Nyzhnyk system, its nonlinear Lax representation and the (symmetric potential dispersive) Nyzhnyk equation.
In Section~\ref{sec:DefiningGeomProperties}, a set of defining geometric properties of the (symmetric potential) dispersionless Nyzhnyk equation found in~\cite{boyk2024a} is analyzed, noting that a similar set for its dispersive counterpart should be much more complicated.
One- and two-dimensional subalgebras of the maximal Lie invariance pseudoalgebra of this dispersive equation and one-dimensional subalgebras of the maximal Lie invariance pseudoalgebra of its linear Lax representation are classified in Section~\ref{sec:OneAndTwoDimSubalg} by analogy with the corresponding results for symmetric dispersionless models from~\cite{vinn2024a}.
Finally, Section~\ref{sec:Conclusion} summarizes the key results of the paper and outlines promising directions for future research.

Throughout the paper, the subscripts in the notation of the maximal Lie invariance pseudoalgebras and point-symmetry pseudogroups of the Nyzhnyk models correspond to the equation numbers of these models.

\section{Variety of Nyzhnyk models}\label{sec:NyzhnykSystemsAndEquations}

The variety of Nyzhnyk models presented in the literature provides an excellent illustration of the mechanisms for generating various types of integrable models from basic ones.
Despite sharing certain fundamental properties, these models also exhibit essentially different features, requiring the application of specific algebraic approaches and integration methods for their study.
The classification presented in Table~\ref{tab:NyzhModels} clearly outlines the structural and functional differences between groups of models, establishing a framework for their systematic analysis and for uncovering their relations with other equations of mathematical physics.
Transitions between different models can be realized via introducing (pseudo)potentials, eliminating of a subset of dependent variables, involving differential consequences, the use of differential substitutions, introducing a parameter followed by a corresponding limiting process, as well as by changing the base field for the independent and/or dependent variables.
For convenience, linear and nonlinear Lax representations for Nyzhnyk systems and equations are also included in the collection of Nyzhnyk models as separate types of such models. %\looseness=-1

\begin{table}[!ht]\small
\begin{center}
\caption{\small Classification of Nyzhnyk models.
\strut}\label{tab:NyzhModels}
\renewcommand{\arraystretch}{1.2}
\begin{tabular}{|l|c|c|}
\hline
\hfil Criterion &    Implementation
   &    Model attribute     \\
\hline
\raisebox{-1.5ex}[0ex][0ex]{$k_1k_2\ne0$}
& yes &  symmetric  \\
& no  &  asymmetric \\
\hline
\raisebox{-1.5ex}[0ex][0ex]{Presence of dispersive terms}
& yes &  dispersive    \\
& no  &  dispersionless \\
\hline
\raisebox{-1.5ex}[0ex][0ex]{Number of equations (or dependent variables)}
& ${}=1$ &  equation \\
& ${}>1$ &  system  \\
\hline
\raisebox{-1.5ex}[0ex][0ex]{Type of nonlinearity}
& quadratic &  standard   \\
& cubic     &  modified \\
\hline
& $(\mathbb R,\mathbb R)$ & real   \\
Base fields for independent and dependent variables
& $(\mathbb C,\mathbb C)$ & complex \\
& $(\mathbb R,\mathbb C)$, $(\mathbb R\mbox{ and }\mathbb C,\mathbb R)$, \dots & mixed \\
\hline
Presence of complex conjugate pairs of &  no  &  standard \\
independent and/or dependent variables & yes  &  specific \\
\hline
\raisebox{-1.5ex}[0ex][0ex]{Sign of nonlinearities in modified models}
& $\epsilon=-1$ &  defocusing \\
& $\epsilon=1$  &  focusing   \\
\hline
\end{tabular}
\end{center}
\end{table}

The nomenclature for Nyzhnyk models that is used in the present paper not only underscores the priority of Nyzhnyk's work~\cite{nizh1980a},
but also directly encodes the structural connections between these models into their names.
Other names for these models have also appeared in the literature.
For instance, specific Nyzhnyk models are referred to as Nyzhnyk--Novikov--Veselov models, or even simply as Novikov--Veselov models.
Moreover, the latter designation is sometimes unjustifiably extended to all Nyzhnyk models.
In addition, as also noted below, the asymmetric (potential) Nyzhnyk equation~\eqref{eq:AsymPotNizhEqBLMP} is frequently called the Boiti--Leon--Manna--Pempinelli equation in view of the later paper~\cite{boit1986a}.

\subsection{Dispersive models}

Let us consider the symmetric and asymmetric dispersive cases in more detail.

In the symmetric case, under the scaling transformation%
\footnote{%
In fact, in the real case, this is a composition of a pure scaling transformation with the sign changes $(x,w,v^1)$ or $(y,w,v^2)$ if $\alpha<0$ or $\beta<0$, respectively.
In the complex case, it is a composition of a pure scaling transformation with coupled rotations in the complex planes $\mathbb C_x$, $\mathbb C_w$, $\mathbb C_{v^1}$ and in the complex planes $\mathbb C_y$, $\mathbb C_w$, $\mathbb C_{v^2}$.
}
$\tilde t=t$, $\tilde x=\alpha x$, $\tilde y=\beta y$,
$\tilde w=(\alpha\beta)^{-1}w$, $\tilde v^1=\alpha v^1$ and $\tilde v^2=\beta v^2$ with $\alpha\beta\ne0$
the coefficients~$k_1$ and $k_2$ transform as $\tilde k_1=\alpha^3k_1$ and $\tilde k_2=\beta^3k_2$.
Hence, without loss of generality, they can be simultaneously set to 1, i.e., $k_1=k_2=1$.
Then the system~\eqref{eq:NyzhnykSystem} takes the form
\begin{gather}\label{eq:SymNyzhnykSystem}
w_t=w_{xxx}+w_{yyy}+3(v^1w)_x+3(v^2w)_y,\quad
v^1_y=w_x, \quad v^2_x=w_y
\end{gather}
and is called the {\it (symmetric) Nyzhnyk system}.
The substitution $k_1=k_2=1$ into~\eqref{eq:NyzhnykSystemLax} gives a (linear) Lax representation
\begin{gather}\label{eq:LaxSymNyzhnykSystem}
\psi_{xy}+w\psi=0, \quad
\psi_t=\psi_{xxx}+\psi_{yyy}+3v^1\psi_x+3v^2\psi_y
\end{gather}
of the system~\eqref{eq:SymNyzhnykSystem}.
We employ the last two equations of the system~\eqref{eq:SymNyzhnykSystem} as its short conservation laws to (locally) introduce the potentials~$\phi$ and~$\theta$, respectively.
Thus, these potentials are defined by the equations $\phi_x=v^1$, $\phi_y=w$, $\theta_x=w$ and $\theta_y=v^2$, whence $\theta_x=\phi_y$.
We consider this equality as a short conservation law of the system~\eqref{eq:SymNyzhnykSystem} extended by the equations defining the potentials~$\phi$ and~$\theta$, and introduce another potential~$u$ according to it, obtaining the equations $u_x=\phi$ and $u_y=\theta$.
In view of all the specified equations for the potentials, the unknown functions $w$, $v^1$ and $v^2$ of the system~\eqref{eq:SymNyzhnykSystem} can be expressed in terms of the second derivatives of the potential~$u$:
$w=u_{xy}$, $v^1=u_{xx}$ and $v^2=u_{yy}$.
Under this representation, the last two equations of the system~\eqref{eq:SymNyzhnykSystem} become identities, and therefore the entire system~\eqref{eq:SymNyzhnykSystem} is equivalent to a~single equation
\begin{gather}\label{eq:SymPotNizhEq}
u_{txy}=u_{xxxxy}+u_{xyyyy}+3(u_{xx}u_{xy})_x+3(u_{yy}u_{xy})_y
\end{gather}
for the potential~$u$, which we call the {\it (symmetric potential) Nyzhnyk equation}.
The system
\begin{gather}\label{eq:LaxSymPotNizhEq}
\psi_{xy}+u_{xy}\psi=0, \quad
\psi_t=\psi_{xxx}+\psi_{yyy}+3u_{xx}\psi_x+3u_{yy}\psi_y
\end{gather}
is its (linear) Lax representation.

Analogously, consider the system~\eqref{eq:NyzhnykSystem} in the asymmetric case $k_1k_2=0$.
Up to the permutation of the independent variables $x$ and $y$, we can assume that $k_1\ne0$ and $k_2=0$. Then the scaling transformation $\tilde t=t$, $\tilde x=\alpha x$, $\tilde y=y$, $\tilde w=\alpha^{-1}w$, $\tilde v^1=\alpha v^1$ and $\tilde v^2=v^2$, where $\alpha\ne0$, transforms the coefficient~$k_1$ as $\tilde k_1=\alpha^3k_1$, so this coefficient can be set to 1.
As a result, without loss of generality, $(k_1,k_2)=(1,0)$ and the system~\eqref{eq:NyzhnykSystem} takes the form
\begin{gather}\label{eq:AsymNyzhnykSystem}
w_t=w_{xxx}+3(v^1w)_x+3(v^2w)_y,\quad
v^1_y=w_x, \quad v^2_x=0.
\end{gather}
For any fixed solution of the system~\eqref{eq:AsymNyzhnykSystem}, the transformation
\begin{gather*}
\tilde t=t, \quad \tilde x=x, \quad \tilde y=Y(t,y), \quad
\tilde w=\frac w{Y_y}, \quad \tilde v^1=v^1, \quad \tilde v^2=v^2,
% \p_t=\p_{\tilde t}+Y_t\p_{\tilde y}, \quad \p_x=\p_{\tilde x}, \quad \p_y=Y_y\p_{\tilde y},
\end{gather*}
where $Y$ is a particular solution of the equation $Y_t=3v^2Y_y$ with $Y_y\ne0$, in which the coefficient~$v^2$ is the corresponding component of the fixed solution of the system~\eqref{eq:AsymNyzhnykSystem}, maps this solution to a solution with a vanishing $v^2$-component.
In other words, the system~\eqref{eq:AsymNyzhnykSystem} is equivalent to the {\it reduced asymmetric Nyzhnyk system}
\begin{gather}\label{eq:AsymNyzhnykSystemWithoutV2}
w_t=w_{xxx}+3(v^1w)_x,\quad v^1_y=w_x,
\end{gather}
which possesses the (linear) Lax representation
\begin{gather}\label{eq:NyzhnykSystemAsymLax}
\psi_{xy}+w\psi=0, \quad
\psi_t=\psi_{xxx}+3v^1\psi_x.
\end{gather}

We employ the last equation of the system~\eqref{eq:AsymNyzhnykSystemWithoutV2} as a short conservation law to introduce the potential~$u$.
Thus, it is defined by the equations $u_x=v^1$ and $u_y=w$, under which the second equation of the system~\eqref{eq:AsymNyzhnykSystemWithoutV2} is an identity.
Interpreting these equations as a substitution for $v^1$ and~$w$, from the first equation of the system~\eqref{eq:AsymNyzhnykSystemWithoutV2} we obtain the equation
\begin{gather}\label{eq:AsymPotNizhEqBLMP}
u_{ty}=u_{xxxy}+3(u_xu_y)_x
\end{gather}
for the potential~$u$, which is naturally called the {\it asymmetric (potential) Nyzhnyk equation} and is also known as the {\it Boiti--Leon--Manna--Pempinelli equation} in view of the work~\cite{boit1986a}.
Thus, the system~\eqref{eq:AsymNyzhnykSystemWithoutV2} is equivalent to the single equation~\eqref{eq:AsymPotNizhEqBLMP}, for which a (linear) Lax representation is given by the system
\begin{gather}\label{eq:LaxAsymPotNizhEqBLMP}
\psi_{xy}+u_y\psi=0, \quad
\psi_t=\psi_{xxx}+3u_x\psi_x.
\end{gather}

\begin{remark}\label{rem:DispersiveCaseRescaling}
In the models~\eqref{eq:SymNyzhnykSystem}, \eqref{eq:SymPotNizhEq}, \eqref{eq:AsymNyzhnykSystemWithoutV2} and~\eqref{eq:AsymPotNizhEqBLMP},
as well as for their dispersionless counterparts (see below),
the factor 3 in the nonlinearities can be eliminated
by using the scaling transformation
\begin{gather}\label{eq:ScalingToDelete3}
(\tilde t, \tilde x, \tilde y, \tilde w, \tilde v^1, \tilde v^2, \tilde u)=(t,x,y,3w,3v^1,3v^2,3u).
\end{gather}
We denote the models with rescaled variables by adding a prime to the corresponding numbers, namely as (\ref{eq:SymNyzhnykSystem}$'$), (\ref{eq:SymPotNizhEq}$'$), (\ref{eq:AsymNyzhnykSystemWithoutV2}$'$) and (\ref{eq:AsymPotNizhEqBLMP}$'$).
The rescaled Lax representations take the following form (all tildes over the variables are omitted):
\begin{gather*}
\label{eq:LaxSymNyzhnykSystem'}
\psi_{xy}+\frac13w\psi=0, \quad
\psi_t=\psi_{xxx}+\psi_{yyy}+v^1\psi_x+v^2\psi_y,              \tag{\ref{eq:LaxSymNyzhnykSystem}$'$}
\\\label{eq:LaxSymPotNizhEq'}
\psi_{xy}+\frac13u_{xy}\psi=0, \quad
\psi_t=\psi_{xxx}+\psi_{yyy}+u_{xx}\psi_x+u_{yy}\psi_y,        \tag{\ref{eq:LaxSymPotNizhEq}$'$}
\\\label{eq:NyzhnykSystemAsymLax'}
\psi_{xy}+\frac13w\psi=0, \quad \psi_t=\psi_{xxx}+v^1\psi_x,   \tag{\ref{eq:NyzhnykSystemAsymLax}$'$}
\\\label{eq:LaxAsymPotNizhEqBLMP'}
\psi_{xy}+\frac13u_y\psi=0, \quad \psi_t=\psi_{xxx}+u_x\psi_x. \tag{\ref{eq:LaxAsymPotNizhEqBLMP}$'$}
\end{gather*}
\end{remark}

\subsection{Dispersionless models}

Let us consider the transition from the original Nyzhnyk system~\eqref{eq:NyzhnykSystem} and its Lax representation to their dispersionless counterparts.
To this end, by means of the scaling transformation
\begin{gather*}%\label{eq:TransformVariables}
\tilde t=\varepsilon t, \quad \tilde x=\varepsilon x, \quad \tilde y=\varepsilon y, \quad
\tilde w=w, \quad \tilde v^1=v^1, \quad \tilde v^2=v^2, \quad \tilde u=\varepsilon^{-2}u
%\\ \p_t=\varepsilon\p_{\tilde t}, \quad \p_x=\varepsilon\p_{\tilde x}, \quad \p_y=\varepsilon\p_{\tilde y},
\end{gather*}
in the systems~\eqref{eq:NyzhnykSystem} and~\eqref{eq:NyzhnykSystemLax}, accompanied by the substitution $\psi={\rm e}^{\varepsilon^{-1}\vartheta}$ in the system~\eqref{eq:NyzhnykSystemLax}, we introduce an additional parameter $\varepsilon\ne0$ into these systems, after which the tildes are omitted.
This leads, respectively, to the systems
\begin{gather*}
w_t=k_1\varepsilon^2w_{xxx}+k_2\varepsilon^2w_{yyy}+3(v^1w)_x+3(v^2w)_y, \quad
v^1_y=k_1w_x, \quad v^2_x=k_2w_y,
\\[1ex]
\vartheta_x\vartheta_y+\varepsilon\vartheta_{xy}=-w, \\
\vartheta_t=k_1(\vartheta_x^{}{}^{\!3}+3\varepsilon\vartheta_{xx}\vartheta_x+\varepsilon^2\vartheta_{xxx})
+k_2(\vartheta_y^{}{}^{\!3}+3\varepsilon\vartheta_{yy}\vartheta_y+\varepsilon^2\vartheta_{yyy})
+3v^1\vartheta_x+3v^2\vartheta_y,
\end{gather*}
from which, by performing the limiting process $\varepsilon\rightarrow0$, we obtain the system
\begin{gather}\label{eq:dNSystemWithK}
w_t=3(v^1w)_x+3(v^2w)_y, \quad
v^1_y=k_1w_x, \quad v^2_x=k_2w_y
\end{gather}
and its nonlinear Lax representation
\begin{gather}\label{eq:LaxdNSystemWithK}
\vartheta_x\vartheta_y=-w, \quad
\vartheta_t=k_1\vartheta_x^{}{}^{\!3}+k_2\vartheta_y^{}{}^{\!3}+3v^1\vartheta_x+3v^2\vartheta_y.
\end{gather}

Similarly to the system~\eqref{eq:NyzhnykSystem}, the system~\eqref{eq:dNSystemWithK} can be treated separately in the symmetric and asymmetric cases by executing the same transformations and potential introductions as before.

In the symmetric case, in addition to gauging $(k_1,k_2)$ to $(1,1)$, we rescale the dependent variables as in~\eqref{eq:ScalingToDelete3}, which yields the canonical form of the {\it (symmetric) dispersionless Nyzhnyk system}:
\begin{gather}\label{eq:dNSystem}
w_t=(v^1w)_x+(v^2w)_y, \quad
v^1_y=w_x, \quad v^2_x=w_y.
\end{gather}
See, for instance,~\cite{pavl2006c}, where the system~\eqref{eq:dNSystem} was linked to a certain hydrodynamic chain.
The nonlinear Lax representation for this system is the system~\eqref{eq:LaxdNSystemWithK} with $(k_1,k_2)=(1,1)$, wherein we additionally rescale $\vartheta$ as $\tilde\vartheta=\sqrt3\,\vartheta$, resulting in the system
\begin{gather}\label{eq:LaxdNSystem}
\vartheta_x\vartheta_y=-w, \quad
\vartheta_t=\frac13\bigl(\vartheta_x^{}{}^{\!3}+\vartheta_y^{}{}^{\!3}\bigr)+v^1\vartheta_x+v^2\vartheta_y,
\end{gather}
where the tilde over~$\vartheta$ is omitted.
Just as for the other dispersionless Nyzhnyk models considered below, the nonlinear Lax representation~\eqref{eq:LaxdNSystem} of the system~\eqref{eq:dNSystem} can be converted into the corresponding linear nonisospectral Lax representation
\begin{gather}\label{eq:LaxdNSystemNonisospectral}
\begin{split}&
\xi_t=\left(p^2+\frac{w^3}{p^4}+v^1+\frac{wv^2}{p^2}\right)\xi_x
-\left(-\frac{w^2w_x}{p^3}+pv^1_x-\frac{(wv^2)_x}p\right)\xi_p,
\\&
\xi_y=\frac{w}{p^2}\xi_x+\frac{w_x}p\xi_p
\end{split}
\end{gather}
of this system using the procedure described in \cite[p.~360]{serg2018a} for the general case of dispersionless models; see also the references therein concerning linear nonisospectral Lax representations.
In the system~\eqref{eq:LaxdNSystemNonisospectral} and henceforth, $p$ is a variable spectral parameter upon which only the unknown function~$\xi$ depends, $\xi=\xi(t,x,y,p)$.
The compatibility condition of the system~\eqref{eq:LaxdNSystemNonisospectral} should be split on the solution set of the system, first with respect to the parametric derivatives of the function~$\xi$, and then with respect to the independent variable~$p$.
Note that a~disadvantage of the systems~\eqref{eq:LaxdNSystem} and~\eqref{eq:LaxdNSystemNonisospectral}
as the nonlinear Lax representation of the system~\eqref{eq:dNSystem} and its linear nonisospectral Lax representation, respectively,
is that after splitting the associated compatibility conditions, instead of the equation $v^2_x=w_y$, we obtain its consequence $w(v^2_x-w_y)=0$.
This is the only case of such kind among the dispersionless Nyzhnyk models.

If we express the unknown functions $w$, $v^1$ and $v^2$ via the second derivatives of the potential~$u$, the last two equations of the system~\eqref{eq:dNSystem} become identities and the entire system~\eqref{eq:dNSystem} is equivalent to a single equation, namely the {\it (symmetric potential) dispersionless Nyzhnyk equation}
\begin{gather}\label{eq:dN}
u_{txy}=(u_{xx}u_{xy})_x+(u_{xy}u_{yy})_y,
\end{gather}
which first explicitly appeared in an equivalent form in~\cite[Eq.~(63)]{kono2002b}.
The substitution $w=u_{xy}$, $v^1=u_{xx}$ and $v^2=u_{yy}$ in~\eqref{eq:LaxdNSystem} gives the nonlinear Lax representation
\begin{gather}\label{eq:dNLaxPair}
\vartheta_x\vartheta_y=-u_{xy}, \quad
\vartheta_t=\frac13\bigl(\vartheta_x^{}{}^{\!3}+\vartheta_y^{}{}^{\!3}\bigr)+u_{xx}\vartheta_x+u_{yy}\vartheta_y
\end{gather}
of the equation~\eqref{eq:dN}.
This representation was examined in~\cite[Eq.~(2)]{moro2021a} and~\cite[Eq.~(13)]{boyk2024a} in the form solved with respect to the derivatives~$\vartheta_y$ and~$\vartheta_t$:
\begin{gather*}
\vartheta_t=\frac13\left(\vartheta_x^{}{}^{\!3}-\frac{u_{xy}^{}{}^{\!3}}{\vartheta_x^{}{}^{\!3}}\right)+u_{xx}\vartheta_x-\frac{u_{xy}u_{yy}}{\vartheta_x},\quad
\vartheta_y=-\frac{u_{xy}}{\vartheta_x}.
\end{gather*}
The nonlinear Lax representation of the equation~\eqref{eq:dN} was first obtained in~\cite{pavl2006c} in an alternative form
\begin{gather}\label{eq:dNLaxPairPavlov}
\varrho_t=\left(\frac13\left(\varrho^3-\frac{u_{xy}^{}{}^{\!3}}{\varrho^3}\right)+u_{xx}\varrho-\frac{u_{xy}u_{yy}}\varrho\right)_x,\quad
\varrho_y=-\left(\frac{u_{xy}}\varrho\right)_x,
\end{gather}
where each equation of the preceding system is additionally differentiated with respect to the variable~$x$, and the derivative~$\vartheta_x:=\varrho$ plays the role of the dependent variable instead of~$\vartheta$.
The nonlinear Lax representation~\eqref{eq:dNLaxPair} of the equation~\eqref{eq:dN} can be converted in the aforementioned manner into the linear nonisospectral Lax representation of the same equation:
\begin{gather}\label{eq:dNLaxPairNonisospectral}
\begin{split}&
\xi_t=\left(p^2+\frac{u_{xy}^{}{}^{\!3}}{p^4}+u_{xx}+\frac{u_{xy}u_{yy}}{p^2}\right)\xi_x
-\left(pu_{xxx}-\frac{(u_{xy}u_{yy})_x}p-\frac{u_{xy}^{}{}^{\!2}u_{xxy}}{p^3}\right)\xi_p,
\\&
\xi_y=\frac{u_{xy}}{p^2}\xi_x+\frac{u_{xxy}}p\xi_p.
\end{split}
\end{gather}

As in the case of the system~\eqref{eq:AsymNyzhnykSystem}, in the asymmetric case of the system~\eqref{eq:dNSystemWithK} where $(k_1,k_2)=(1,0)$, one can set $v^2=0$ and together with the rescaling of dependent variables as in~\eqref{eq:ScalingToDelete3} this leads to the canonical form
\begin{gather}\label{eq:AsymdNSystem}
w_t=(v^1w)_x, \quad
v^1_y=w_x,
\end{gather}
of the {\it asymmetric dispersionless Nyzhnyk system}, whose nonlinear and linear nonisospectral Lax representations are, respectively, the systems
\begin{gather}\label{eq:LaxAsymdNSystem}
\vartheta_x\vartheta_y=-w, \quad
\vartheta_t=\frac13\vartheta_x^{}{}^{\!3}+v^1\vartheta_x,
\\\label{eq:LaxAsymdNSystemNonisospectral}
\xi_t=(p^2+v^1)\xi_x-pv^1_x\xi_p,
\quad
\xi_y=\frac{w}{p^2}\xi_x+\frac{w_x}p\xi_p.
\end{gather}

Introducing the potential~$u$, where $u_x=v^1$ and $u_y=w$, reduces the system~\eqref{eq:AsymdNSystem} to the single equation
\begin{gather}\label{eq:AsymdN}
u_{ty}=(u_xu_y)_x,
\end{gather}
which we call the {\it asymmetric (potential) dispersionless Nyzhnyk equation}%
\footnote{%
The equation~\eqref{eq:AsymdN} is equivalent to the equation with $r=3/2$ from the family of equations (46) in~\cite{blas2002a}, which were designated in~\cite{moro2009a} as the ``$r$-th dispersionless (2+1) Harry Dym equation'' for a fixed value of~$r$.
The nonlinear Lax representation~\eqref{eq:LaxAsymdNSystem} of this equation was first constructed for its specific counterpart~\eqref{eq:SpecAsymdN} in~\cite[Eqs.~(97)--(98)]{kono2004b}; see Remark~\ref{rem:SpecAsymNyzhnykModels} below.
}
and its nonlinear and linear nonisospectral Lax representations have the following form, respectively:
\begin{gather}\label{eq:AsymdNLaxPair}
\vartheta_x\vartheta_y=-u_y,\quad
\vartheta_t=\frac13\vartheta_x^{}{}^{\!3}+u_x\vartheta_x,
\\\label{eq:AsymdNLaxPairNonisospectral}
\xi_t=(p^2+u_x)\xi_x-pu_{xx}\xi_p,
\quad
\xi_y=\frac{u_y}{p^2}\xi_x+\frac{u_{xy}}p\xi_p.
\end{gather}

The presented models can be considered under the conditions that all independent and dependent variables are either real (real Nyzhnyk models), complex (complex Nyzhnyk models), or the unknown functions are complex-valued functions of real independent variables (partially complex Nyzhnyk models).
Purely complex models and related objects are denoted with the additional subscript~$\mathbb C$ to distinguish them from their purely real counterparts.

\subsection{Specific models}

There exists a more specific version of the symmetric Nyzhnyk system~\eqref{eq:SymNyzhnykSystem}, which can be obtained from~\eqref{eq:SymNyzhnykSystem} via the formal substitution of $z$, $\bar z$, $v$ and $\bar v$ for $x$, $y$, $v^1$ and $v^2$:
\begin{gather}\label{eq:SpecNyzhnykSystem}
w_t=w_{zzz}+w_{\bar z\bar z\bar z}+3(vw)_z+3(\bar vw)_{\bar z}, \quad v_{\bar z}=w_z, \quad \bar v_z=w_{\bar z},
\end{gather}
where the bar denotes complex conjugation,
$z:=x+{\rm i}y$, ${\bar z:=x-{\rm i}y}$, $\p_z=\frac12(\p_x-\mathrm i\p_y)$
and $\p_{\bar z}=\frac12(\p_x+\mathrm i\p_y)$ are the corresponding Wirtinger derivatives,
$v=v^1+\mathrm iv^2$ and $w$, $v^1$ and~$v^2$ are real-valued functions of the variables $(t,z,\bar z)$ or, equivalently, of the real variables $(t,x,y)$.
Here and in what follows it is convenient to consider~$z$ and~$\bar z$ as independent variables.
This version of the Nyzhnyk system was investigated in~\cite[Eq.~(5)]{vese1984a}.
The last equation of the system~\eqref{eq:SpecNyzhnykSystem} is the complex conjugate of its second equation.
Hence, it is not independent of the other equations and can be omitted.
It is included in~\eqref{eq:SpecNyzhnykSystem} solely to make the relation between the systems~\eqref{eq:SymNyzhnykSystem} and~\eqref{eq:SpecNyzhnykSystem} obvious.
For the same purpose, a gauge of the coefficients is chosen in~\eqref{eq:SpecNyzhnykSystem} that differs from those in the literature by scaling and a sign change.
The dispersionless counterpart for~\eqref{eq:SpecNyzhnykSystem} is the system
\begin{gather}\label{eq:SpecdNSystem}
w_t=3(vw)_z+3(\bar vw)_{\bar z}, \quad v_{\bar z}=w_z, \quad \bar v_z=w_{\bar z}.
%v_t=3(wv)_z+3(\bar wv)_{\bar z}, \quad w_{\bar z}=-v_z.
\end{gather}
Up to the rescaling described in Remark~\ref{rem:DispersiveCaseRescaling}, the above formal substitution allows one to construct the (linear) Lax representation for the system~\eqref{eq:SpecNyzhnykSystem} and the nonlinear Lax representation for the system~\eqref{eq:SpecdNSystem} from the respective representations~\eqref{eq:LaxSymPotNizhEq} and~\eqref{eq:LaxdNSystem} for the systems~\eqref{eq:SymNyzhnykSystem} and~\eqref{eq:dNSystem}:
\begin{gather}\label{eq:LaxSpecNyzhnykSystem}
\psi_{z\bar z}+w\psi=0, \quad
\psi_t=\psi_{zzz}+\psi_{\bar z\bar z\bar z}+3v\psi_z+3\bar v\psi_{\bar z},
\\\label{eq:LaxSpecdNSystem}
\vartheta_z\vartheta_{\bar z}=-w, \quad
\vartheta_t=\vartheta_z^{}{}^{\!3}+\vartheta_{\bar z}^{}{}^{\!3}+3v\vartheta_z+3\bar v\vartheta_{\bar z}.
\end{gather}

We employ the common second equation of the systems~\eqref{eq:SpecNyzhnykSystem} and~\eqref{eq:SpecdNSystem} as a ``short'' conservation law to (locally) introduce a complex-valued potential~$\phi$.
In other words, this potential is defined by the equations $\phi_z=v$ and $\phi_{\bar z}=w$, from which by complex conjugation we also obtain $\bar\phi_z=w$.
Thus, $\phi_{\bar z}=\bar\phi_z$.
Let us now consider the last equality as a ``short'' potential conservation law and introduce another potential~$u$ according to it, defined by the equations $u_z=\phi$ and $u_{\bar z}=\bar\phi$, from which complex conjugating and linearly combining imply that $(u-\bar u)_z=(u-\bar u)_{\bar z}=0$, and therefore $u-\bar u$ depends only on~$t$.
Up to the ambiguity in the definition of the potential~$u$, which allows adding an arbitrary function of~$t$ to it, one can set $u-\bar u=0$, thereby assuming the potential~$u$ to be real-valued.
The unknown functions of the initial systems~\eqref{eq:SpecNyzhnykSystem} and~\eqref{eq:SpecdNSystem} can be expressed in terms of the potential~$u$ as
\begin{gather}\label{eq:SpecNyzhnykSystemPots}
v=u_{zz}, \quad w=u_{z\bar z}.
\end{gather}
Thus, each of the systems~\eqref{eq:SpecNyzhnykSystem} and~\eqref{eq:SpecdNSystem} is equivalent to a single equation for the potential~$u$:
\begin{gather}\label{eq:SpecPotNizhEq}
u_{tz\bar z}=u_{zzzz\bar z}+u_{z\bar z\bar z\bar z\bar z}
+3(u_{zz}u_{z\bar z})_z+3(u_{\bar z\bar z}u_{z\bar z})_{\bar z},
\\\label{eq:SpecPotdNEq}
u_{tz\bar z}=3(u_{zz}u_{z\bar z})_z+3(u_{\bar z\bar z}u_{z\bar z})_{\bar z}.
\end{gather}
The Lax representation for the equation~\eqref{eq:SpecPotNizhEq} and the nonlinear Lax representation for the equation~\eqref{eq:SpecPotdNEq} take the form, respectively,
\begin{gather}\label{eq:LaxSpecPotNizhEq}
\psi_{z\bar z}+u_{z\bar z}\psi=0, \quad
\psi_t=\psi_{zzz}+\psi_{\bar z\bar z\bar z}+3u_{zz}\psi_z+3u_{\bar z\bar z}\psi_{\bar z},
\\\label{eq:LaxSpecPotdNEq}
\vartheta_z\vartheta_{\bar z}=-u_{z\bar z}, \quad
\vartheta_t=\vartheta_z^{}{}^{\!3}+\vartheta_{\bar z}^{}{}^{\!3}+3u_{zz}\vartheta_z+3u_{\bar z\bar z}\vartheta_{\bar z},
\end{gather}
where $\psi$ and $\vartheta$ are real-valued unknown functions of $(t,z,\bar z)$.
The equations~\eqref{eq:SpecNyzhnykSystem}--\eqref{eq:LaxSpecPotdNEq} can be rewritten in terms of the real variables~$x$ and~$y$:
\begin{gather*}
\label{eq:SpecNyzhnykSystem'}
w_t=\frac14(w_{xx}-3w_{yy})_x+3(v^1w)_x+3(v^2w)_y, \quad (v^1-w)_x=v^2_y,\quad (v^1+w)_y=-v^2_x,
\tag{\ref{eq:SpecNyzhnykSystem}$'$}
\\\label{eq:SpecdNSystem'}
w_t=3(v^1w)_x+3(v^2w)_y, \quad (v^1-w)_x=v^2_y,\quad (v^1+w)_y=-v^2_x,
\tag{\ref{eq:SpecdNSystem}$'$}
\\\label{eq:LaxSpecNyzhnykSystem'}
\frac14\triangle\psi+w\psi=0, \quad \psi_t=\frac14(\psi_{xx}-3\psi_{yy})_x+3(v^1\psi_x+v^2\psi_y),
\tag{\ref{eq:LaxSpecNyzhnykSystem}$'$}
\\\label{eq:LaxSpecdNSystem'}
\frac14(\vartheta_x^{}{}^{\!2}+\vartheta_y^{}{}^{\!2})=-w,\quad
\vartheta_t=\frac14(\vartheta_x^{}{}^{\!3}-3\vartheta_x\vartheta_y^{}{}^{\!2})+3(v^1\vartheta_x+v^2\vartheta_y),
\tag{\ref{eq:LaxSpecdNSystem}$'$}
\\\label{eq:SpecNyzhnykSystemPots'}
v^1=\frac14(u_{xx}-u_{yy}),\quad v^2=-\frac12 u_{xy},\quad w=\frac14(u_{xx}+u_{yy}),
\tag{\ref{eq:SpecNyzhnykSystemPots}$'$}
\\\label{eq:SpecPotNizhEq'}
\triangle u_t=\frac14(\triangle u_{xx}-3\triangle u_{yy})_x+\frac34(u_{xx}^{}{}^{\!2}-u_{yy}^{}{}^{\!2})_x-\frac32(u_{xy}\triangle u)_y,
\tag{\ref{eq:SpecPotNizhEq}$'$}
\\\label{eq:SpecPotdNEq'}
\triangle u_t=\frac34(u_{xx}^{}{}^{\!2}-u_{yy}^{}{}^{\!2})_x-\frac32(u_{xy}\triangle u)_y,
\tag{\ref{eq:SpecPotdNEq}$'$}
\\\label{eq:LaxSpecPotNizhEq'}
\triangle\psi+\psi\triangle u=0, \quad \psi_t=\frac14(\psi_{xx}-3\psi_{yy})_x+\frac34(u_{xx}-u_{yy})\psi_x-\frac32u_{xy}\psi_y,
\tag{\ref{eq:LaxSpecPotNizhEq}$'$}
\\\label{eq:LaxSpecPotdNEq'}
\vartheta_x^{}{}^{\!2}+\vartheta_y^{}{}^{\!2}=-\triangle u,\quad
\vartheta_t=\frac14(\vartheta_x^{}{}^{\!3}-3\vartheta_x\vartheta_y^{}{}^{\!2})+\frac34(u_{xx}-u_{yy})\vartheta_x-\frac32u_{xy}\vartheta_y,
\tag{\ref{eq:LaxSpecPotdNEq}$'$}
\end{gather*}
where $\triangle:=\p_x^{\,2}+\p_y^{\,2}$.

The construction of linear nonisospectral Lax representations for specific dispersionless Nyzhnyk models~\eqref{eq:SpecdNSystem} and~\eqref{eq:SpecPotdNEq} exhibits certain peculiarities.
In terms of the variables~$(t,z,\bar z,w,v,\bar v)$, up to the formal substitution of $z$, $\bar z$, $v$ and $\bar v$ for $x$, $y$, $v^1$ and $v^2$, these representations have the same form as for the corresponding Nyzhnyk models in the symmetric dispersionless case,
\begin{gather*}\label{eq:LaxSpecdNSystemNonisospectral}
\xi_t=3\left(p^2+\frac{w^3}{p^4}+v+\frac{w\bar v}{p^2}\right)\xi_z
+3\left(\frac{w^2w_z}{p^3}-pv_z+\frac{(w\bar v)_z}{p}\right)\xi_p, \quad
\xi_{\bar z}=\frac{w}{p^2}\xi_z+\frac{w_z}p\xi_p,
\\[1.5ex]\label{eq:LaxSpecPotdNEqNonisospectral}
\begin{split}&
\xi_t=3\left(p^2+\frac{u_{z\bar z}^{}{}^{\!3}}{p^4}+u_{zz}+\frac{u_{z\bar z}u_{\bar z\bar z}}{p^2}\right)\xi_z
+3\left(\frac{u_{z\bar z}^{}{}^{\!2}u_{zz\bar z}}{p^3}-pu_{zzz}+\frac{(u_{z\bar z}u_{\bar z\bar z})_z}p\right)\xi_p,
\\&
\xi_{\bar z}=\frac{u_{z\bar z}}{p^2}\xi_z+\frac{u_{zz\bar z}}p\xi_p,
\end{split}
\end{gather*}
but here $p$ is a complex independent variable and $\xi$ is a complex-valued function of the variables $(t,z,\bar z,p)$ that is complex-differentiable with respect to~$p$.
To construct purely the real linear nonisospectral Lax representations for the system~\eqref{eq:SpecdNSystem} and the equation~\eqref{eq:SpecPotdNEq} in terms of the real variables~$(t,x,y,w,v^1,v^2)$, it is necessary to utilize the nonlinear Lax representations~\eqref{eq:LaxSpecdNSystem'} and~\eqref{eq:LaxSpecPotdNEq'} in the same terms.
The resulting representations contain nonrational expressions, namely
$R:=\sqrt{-4w-p^2}$ or $\tilde R:=\sqrt{-\triangle u-p^2}$:
\begin{gather*}
%\xi_y=\pm(-4w-p^2)^{1/2},\quad \vartheta_t=p^3+3pw+3(v^1p\pm v^2(-4w-p^2)^{1/2}),\\
\xi_y=-\epsilon\frac pR\xi_x+\epsilon\frac{2w_x}R\xi_p,\\
\xi_t=3\left(p^2+w+v^1-\epsilon\frac{pv^2}R\right)\xi_x
-3\left(p(w+v^1)_x+\epsilon v^2_xR-\epsilon\frac{2v^2w_x}R\right)\xi_p,
\\[1ex]
%\xi_y=\pm(-\triangle u-p^2)^{1/2},\quad \xi_t=p^3+\frac34p\triangle u+\frac34(u_{xx}-u_{yy})p-\frac32u_{xy}(-\triangle u-p^2)^{1/2}=p^3+\frac32pu_{xx}\mp\frac32u_{xy}(-\triangle u-p^2)^{1/2},\\
\xi_y=-\epsilon\frac p{\tilde R}\xi_x+\epsilon\frac{\triangle u_x}{2\tilde R}\xi_p,\\
\xi_t=3\left(p^2+\frac12u_{xx}+\epsilon\frac{pu_{xy}}{2\tilde R}\right)\xi_x
-\frac32\left(pu_{xxx}-\epsilon u_{xxy}\tilde R+\epsilon\frac{u_{xy}\triangle u_x}{2\tilde R}\right)\xi_p,
\end{gather*}
where $\epsilon=\pm1$, $p$ is now a real independent variable and $\xi$ is a real-valued function of the variables $(t,x,y,p)$.
The Lax representations with nonrational expressions have previously appeared in the literature, see, e.g.,~\cite{serg2019a}.

\begin{remark}\label{rem:SpecAsymNyzhnykModels}
One can also consider other specific Nyzhnyk models, including asymmetric ones, particularly when all dependent variables are complex.
For instance, the specific asymmetric dispersionless Nyzhnyk equation
\begin{gather}\label{eq:SpecAsymdN}
u_{t\bar z}=(u_zu_{\bar z})_z
\end{gather}
was derived in~\cite[Eq.~(67)]{kono2004b}, where a nonlinear Lax representation was also constructed for it \cite[Eqs.~(97)--(98)]{kono2004b}, which is a specific version of the nonlinear Lax representation~\eqref{eq:LaxAsymdNSystem} for the asymmetric dispersionless Nyzhnyk equation~\eqref{eq:AsymdN}.
Furthermore, a family of its exact solutions, parameterized by an arbitrary function and an arbitrary constant, was found in implicit form via the hydrodynamic reduction method.
\end{remark}

\subsection{Modified models}

Let us rewrite the system~\eqref{eq:NyzhnykSystem} as a system with respect to the unknown functions~$w$ and~$u$:
\begin{gather*}
w_t=k_1w_{xxx}+k_2w_{yyy}+3k_1(u_{xx}w)_x+3k_2(u_{yy}w)_y,\quad
w=u_{xy}.
\end{gather*}
We perform the differential substitution
% Old: $w=-{\rm i}\hat w_{xy}+\hat w_x\hat w_y$, $u=\hat u-{\rm i}\hat w$,
% -{\rm i}\hat w\to\hat w, \hat u\to-\hat u
$w=\hat w_{xy}-\hat w_x\hat w_y$, $u=\hat w-\hat u$,
under which
\begin{gather*}
w_t-k_1w_{xxx}-k_2w_{yyy}-3k_1(u_{xx}w)_x-3k_2(u_{yy}w)_y
\\\quad{}
=\big(\p_x\p_y-\hat w_y\p_x-\hat w_x\p_y\big)\big(
\hat w_t
-k_1(\hat w_{xxx}-3\hat w_x\hat u_{xx}+\hat w_x^{}{}^{\!3})
-k_2(\hat w_{yyy}-3\hat w_y\hat u_{yy}+\hat w_y^{}{}^{\!3})
\big)\\\qquad{}
-3\big(
 k_1(\hat w_x\p_x^2+(\hat w_{xx}-\hat w_x^{}{}^{\!2})\p_x)
+k_2(\hat w_y\p_y^2+(\hat w_{yy}-\hat w_y^{}{}^{\!2})\p_y)
\big)\big(\hat u_{xy}-\hat w_x\hat w_y\big),
\\[1ex]
u_{xy}-w=\hat u_{xy}-\hat w_x\hat w_y.
\end{gather*}
In other words, under the differential substitution
$w=\hat w_{xy}-\hat w_x\hat w_y$,
$v^1=k_1(\hat w_{xx}-\hat u_{xx})$ and
${v^2=k_2(\hat w_{yy}-\hat u_{yy})}$,
the Nyzhnyk system~\eqref{eq:NyzhnykSystem} is a differential consequence of the {\it focusing modified Nyzhnyk system}
\begin{gather*}\label{eq:FocModNyzhnykSystem}
w_t=k_1w_{xxx}+k_2w_{yyy}-3k_1u_{xx}w_x-3k_2u_{yy}w_y+k_1w_x^{}{}^{\!3}+k_2w_y^{}{}^{\!3},\quad
u_{xy}=w_xw_y,
\end{gather*}
where the hats over~$u$ and~$w$ are omitted.
In the real case, we complexify the dependent variables~$u$ and~$w$ and perform the substitution $w\to{\rm i}w$ and $u\to-u$, after which we consider the new dependent variables~$u$ and~$w$ to be real, yielding the {\it defocusing modified Nyzhnyk system}
\begin{gather}\label{eq:DefocModNyzhnykSystem}
w_t=k_1w_{xxx}+k_2w_{yyy}+3k_1u_{xx}w_x+3k_2u_{yy}w_y-k_1w_x^{}{}^{\!3}-k_2w_y^{}{}^{\!3},\quad
u_{xy}=w_xw_y,
\end{gather}
see, e.g.,~\cite{nimm1993a}.
In the symmetric case $k_1k_2\ne0$, under the scaling transformation
$\tilde t=t$, $\tilde x=\alpha x$, $\tilde y=\beta y$, $\tilde w=w$ and $\tilde u=u$, the coefficients~$k_1$ and~$k_2$ transform as $\tilde k_1=\alpha^3k_1$ and $\tilde k_2=\beta^3k_2$.
Thus, one can set $k_1=k_2=1$, meaning the canonical form of the {\it (symmetric) defocusing or focusing modified Nyzhnyk system} is as follows:
\begin{gather}\label{eq:SymModNyzhnykSystem}
w_t=w_{xxx}+w_{yyy}-3\epsilon u_{xx}w_x-3\epsilon u_{yy}w_y+\epsilon w_x^{}{}^{\!3}+\epsilon w_y^{}{}^{\!3},\quad
u_{xy}=w_xw_y.
\end{gather}
Here and henceforth, $\epsilon=-1$ or $\epsilon=1$ in the defocusing or focusing case, respectively.
Analogously, in the asymmetric case $k_1k_2=0$, up to the permutation of $x$ and $y$, we can assume that $k_1\ne0$ and $k_2=0$, and by the scaling transformation as above with $\beta=1$, we can set $k_1=1$.
Consequently, we obtain the {\it asymmetric defocusing or focusing modified Nyzhnyk system}
\begin{gather}\label{eq:AsymModNyzhnykSystem}
w_t=w_{xxx}-3\epsilon u_{xx}w_x+\epsilon w_x^{}{}^{\!3},\quad
u_{xy}=w_xw_y.
\end{gather}

Let us proceed from the system~\eqref{eq:DefocModNyzhnykSystem} to its dispersionless counterpart.
Using the scaling transformation
\begin{gather*}\label{eq:ModifiedTransformVariables}
\tilde t=\varepsilon t, \quad \tilde x=\varepsilon x, \quad \tilde y=\varepsilon y, \quad
\tilde w=\varepsilon w, \quad \tilde u=\varepsilon^{-2}u,
\end{gather*}
we introduce an additional parameter $\varepsilon\ne0$ into this system
\begin{gather*}
w_t=k_1\varepsilon^2w_{xxx}+k_2\varepsilon^2w_{yyy}+3k_1u_{xx}w_x+3k_2u_{yy}w_y-k_1w_x^{}{}^{\!3}-k_2w_y^{}{}^{\!3},\quad
u_{xy}=w_xw_y
\end{gather*}
and after the limiting process $\varepsilon\rightarrow0$, we obtain the general {\it dispersionless modified Nyzhnyk system}
\begin{gather*}\label{eq:DispersionlessDefocModNyzhnykSystem}
w_t=3k_1u_{xx}w_x+3k_2u_{yy}w_y-k_1w_x^{}{}^{\!3}-k_2w_y^{}{}^{\!3},\quad
u_{xy}=w_xw_y.
\end{gather*}
Analogously to the dispersive case, we distinguish two cases: symmetric and asymmetric when $k_1k_2\ne0$ and $k_1k_2=0$, respectively, where without loss of generality, by using the above scaling transformations and, in the asymmetric case, the permutation of the variables~$x$ and~$y$, one can set $(k_1,k_2)=(1,1)$ and $(k_1,k_2)=(1,0)$, respectively.
This yields the canonical form of the {\it symmetric and asymmetric (defocusing or focusing) dispersionless modified Nyzhnyk system}
\begin{gather}\label{eq:SymDispersionlessDefocModNyzhnykSystem}
w_t=-3\epsilon u_{xx}w_x-3\epsilon u_{yy}w_y+\epsilon w_x^{}{}^{\!3}+\epsilon w_y^{}{}^{\!3},\quad
u_{xy}=w_xw_y,
\\\label{eq:AsymDispersionlessDefocModNyzhnykSystem}
w_t=-3\epsilon u_{xx}w_x+\epsilon w_x^{}{}^{\!3},\quad
u_{xy}=w_xw_y.
\end{gather}
Of course, the systems~\eqref{eq:SymDispersionlessDefocModNyzhnykSystem} and~\eqref{eq:AsymDispersionlessDefocModNyzhnykSystem} can be derived as dispersionless limits directly from the systems~\eqref{eq:SymModNyzhnykSystem} and~\eqref{eq:AsymModNyzhnykSystem}, respectively.

By the formal substitution of $z:=x+{\rm i}y$ and ${\bar z:=x-{\rm i}y}$ for $x$ and $y$, specific counterparts can be obtained from the symmetric modified models, including the {\it specific symmetric (defocusing or focusing) modified Nyzhnyk system}
\begin{gather}\label{eq:SpecSymModNyzhnykSystem}
w_t=w_{zzz}+w_{\bar z\bar z\bar z}-3\epsilon u_{zz}w_z-3\epsilon u_{\bar z\bar z}w_{\bar z}+\epsilon w_z^{}{}^{\!3}+\epsilon w_{\bar z}^{}{}^{\!3},\quad
u_{z\bar z}=w_zw_{\bar z}
\end{gather}
and the {\it specific symmetric (defocusing or focusing) dispersionless modified Nyzhnyk system}
\begin{gather}\label{eq:dSpecSymModNyzhnykSystem}
w_t=-3\epsilon u_{zz}w_z-3\epsilon u_{\bar z\bar z}w_{\bar z}+\epsilon w_z^{}{}^{\!3}+\epsilon w_{\bar z}^{}{}^{\!3},\quad
u_{z\bar z}=w_zw_{\bar z}.
\end{gather}
In terms of the new variables $x$ and $y$, these systems take the form
\begin{gather}\label{eq:SpecSymModNyzhnykSystem'}
\begin{split}&
w_t=\frac14(w_{xx}-3w_{yy})_x-\frac34\epsilon w_x(u_{xx}-u_{yy})
+\frac32\epsilon u_{xy}w_y+\frac14\epsilon w_x(w_x^{}{}^{\!2}-3w_y^{}{}^{\!2}),\\&
\triangle u=w_x^{}{}^{\!2}+w_y^{}{}^{\!2},
%u_{xx}+u_{yy}=w_x^{}{}^{\!2}+w_y^{}{}^{\!2},
\end{split}
\tag{\ref{eq:SpecSymModNyzhnykSystem}$'$}
\\[.5ex]
\label{eq:dSpecSymModNyzhnykSystem'}
w_t=-\frac34\epsilon w_x(u_{xx}-u_{yy})
+\frac32\epsilon u_{xy}w_y+\frac14\epsilon w_x(w_x^{}{}^{\!2}-3w_y^{}{}^{\!2}),\quad
\triangle u=w_x^{}{}^{\!2}+w_y^{}{}^{\!2}.
\tag{\ref{eq:dSpecSymModNyzhnykSystem}$'$}
\end{gather}

\begin{remark}\label{rem:RelationOfModModelsToLaxRepresentations}
In fact, all the modified Nyzhnyk models presented above are rewritten versions of the Lax representations (linear or nonlinear in the dispersive or dispersionless cases, respectively) of the associated unmodified Nyzhnyk models.
In particular, the simple point substitution
$\psi={\rm e}^{\tilde w}$ and $u=-\tilde u-\tilde w$
%$w=-(\tilde u+\tilde w)_{xy}$, $v^1=-(\tilde u+\tilde w)_{xx}$, $v^2=-(\tilde u+\tilde w)_{yy}$,
reduces the systems~\eqref{eq:LaxSymPotNizhEq} and~\eqref{eq:LaxAsymPotNizhEqBLMP} to the systems \eqref{eq:SymModNyzhnykSystem} and~\eqref{eq:AsymModNyzhnykSystem} with $\epsilon=1$, respectively.
In the dispersionless case, the connection is completely obvious.
It is only necessary to coordinate the corresponding scalings.
Thus, the substitution $\vartheta=\sqrt 3\tilde w$ and $u=-3\tilde u$ maps the systems~\eqref{eq:dNLaxPair} and~\eqref{eq:AsymdNLaxPair} to the systems~\eqref{eq:SymDispersionlessDefocModNyzhnykSystem} and~\eqref{eq:AsymDispersionlessDefocModNyzhnykSystem} with $\epsilon=1$.
This relation between unmodified and modified Nyzhnyk models is completely analogous to the well-known connection between the standard and modified Korteweg--de Vries equations via the Lax representation of the former, which can be represented as a Miura transformation.
\end{remark}

\subsection{Generalizations}\label{sec:GenNyzhnykModel}

Each of the Nyzhnyk models discussed above is a lowest-order member in a hierarchy of related (1+2)-dimensional integrable models; see, for example, the second members in the hierarchies containing the specific Nyzhnyk system and the modified specific Nyzhnyk system in \cite[Section~2.3]{taim1997a}.

In~\cite[Eq.~(14)]{bogd1987a}, the following (1+2)-dimensional integrable system of two partial differential equations was constructed:
\begin{gather}\label{eq:BogdanovSystem}
\begin{split}&
f_t+f_{zzz}+f_{\bar z\bar z\bar z}
+3f_z\p_{\bar z}^{-1}(fg)_z+3f_{\bar z}\p_z^{-1}(fg)_{\bar z}
+3f\p_z^{-1}(gf_{\bar z})_{\bar z}+3f\p_{\bar z}^{-1}(gf_z)_z
=0,\\&
g_t+g_{zzz}+g_{\bar z\bar z\bar z}
+3g_z\p_{\bar z}^{-1}(fg)_z+3g_{\bar z}\p_z^{-1}(fg)_{\bar z}
+3g\p_z^{-1}(fg_{\bar z})_{\bar z}+3g\p_{\bar z}^{-1}(fg_z)_z
=0
\end{split}
\end{gather}
for two complex-valued functions~$f$ and~$g$ of the variables $(t,z,\bar z)$.
Its (linear) Lax representation is given by the system~\cite[Eq.~(15)]{bogd1987a}
\begin{gather*}\label{eq:BogdanovSystemLaxRepresentation}
\begin{split}&
g\psi^1+{\rm i}\psi^2_z=0,\quad
f\psi^2-{\rm i}\psi^1_{\bar z}=0,\\&
\psi^1_t+\psi^1_{zzz}+\psi^1_{\bar z\bar z\bar z}
+3{\rm i}\p_{\bar z}^{-1}(gf_z)_z\psi^1-3f\p_z^{-1}(fg)_{\bar z}\psi^2
+3{\rm i}\p_{\bar z}^{-1}(fg)_z\psi^1_z+3f_{\bar z}\psi^2_{\bar z}=0,
\\&
\psi^2_t+\psi^2_{zzz}+\psi^2_{\bar z\bar z\bar z}
+3{\rm i}\p_z^{-1}(fg_{\bar z})_{\bar z}\psi^2-3g\p_{\bar z}^{-1}(fg)_z\psi^1
+3{\rm i}\p_z^{-1}(fg)_{\bar z}\psi^2_{\bar z}+3g_z\psi^1_z=0.
\end{split}
\end{gather*}
The fact that the system~\eqref{eq:BogdanovSystem} is a generalization of (specific dispersive) Nyzhnyk models is demonstrated by its componentwise reductions.
Moreover, specific modified Nyzhnyk models were first constructed in~\cite{bogd1987a} precisely by means of componentwise reductions of the system~\eqref{eq:BogdanovSystem}.

In particular, under the condition $g=1$, this system reduces to the equation
\begin{gather*}
f_t+f_{zzz}+f_{\bar z\bar z\bar z}
+3(f\p_{\bar z}^{-1}f_z)_z
+3(f\p_z^{-1}f_{\bar z})_{\bar z}
=0.
\end{gather*}
If we additionally set $\bar f=f$,
denote $w:=f$ and $v:=\p_{\bar z}^{-1}f_z$ (whence $\bar v:=\p_z^{-1}f_{\bar z}$), and replace $t$ by $-t$, we obtain the specific symmetric Nyzhnyk system~\eqref{eq:SpecNyzhnykSystem},
where the last two equations are consequences of the definitions of~$v$ and~$\bar v$.

In turn, imposing the constraint $g=\bar f$ reduces the system~\eqref{eq:BogdanovSystem} to an extension of the specific modified Nyzhnyk model~\cite[Eq.~(17)]{bogd1987a}:
\begin{gather}\label{eq:BogdanovEq17}
f_t+f_{zzz}+f_{\bar z\bar z\bar z}
+3f_z\p_{\bar z}^{-1}(f\bar f)_z
+3f_{\bar z}\p_z^{-1}(f\bar f)_{\bar z}
+3f\p_z^{-1}(\bar ff_{\bar z})_{\bar z}
+3f\p_{\bar z}^{-1}(\bar ff_z)_z
=0.
\end{gather}
Under the additional condition $f_z=\bar f_{\bar z}$, there exists a potential~$w$ defined by the equations $w_{\bar z}=f$ and ${w_z=\bar f}$, whence
$\bar w_z=\bar f=w_z$ and
$\bar w_{\bar z}=f=w_{\bar z}$,
meaning that $(\bar w-w)_z=(\bar w-w)_{\bar z}=0$.
Thus, $\mathop{\rm Im}w=\frac12{\rm i}(\bar w-w)$ is a function depending only on~$t$, which can be set to zero up to the ambiguity of~$w$ as a potential. In other words, $\bar w=w$.
If we denote $u=\p_z^{-1}\p_{\bar z}^{-1}(w_zw_{\bar z})$, then
$u_{z\bar z}=w_zw_{\bar z}$ and from~\eqref{eq:BogdanovEq17} it follows that $V_{\bar z}=0$, where $V:=w_t+w_{zzz}+w_{\bar z\bar z\bar z}+3w_zu_{zz}+3w_{\bar z}u_{\bar z\bar z}-(w_{\bar z})^3-(w_z)^3$.
Since $V$ is a real-valued expression, complex conjugation also yields $V_z=0$.
Hence, $V$ as a function depends only on~$t$, so up to the potential ambiguity of~$w$, one can set $V=0$, which altogether, after replacing $t$ by $-t$, yields the specific (symmetric) defocusing modified Nyzhnyk system~\eqref{eq:SpecSymModNyzhnykSystem}.

If, instead, $g=f=\bar f=:u$, then by additionally denoting $v:=\p_{\bar z}^{-1}(u^2)_z$ (whence $\bar v:=\p_z^{-1}(u^2)_{\bar z}$), we obtain another specific (focusing) modified Nyzhnyk system
\begin{gather}\label{eq:ModNyzhnykSystem2}
u_t=u_{zzz}+u_{\bar z\bar z\bar z}+3u_zv+3u_{\bar z}\bar v+\tfrac32u(v_z+\bar v_{\bar z}),\quad
v_{\bar z}=(u^2)_z
\end{gather}
and its Lax representation
\begin{gather*}
\psi^1_z-u\psi^2=0,\quad
\psi^1_t=\psi^1_{zzz}+\psi^1_{\bar z\bar z\bar z}-3u_z\psi^2_z+3\bar v\psi^1_{\bar z}+3uv\psi^2+\tfrac32{\bar v}_{\bar z}\psi^1,\\
\psi^2_{\bar z}+u\psi^1=0,\quad
\psi^2_t=\psi^2_{zzz}+\psi^2_{\bar z\bar z\bar z}+3u_{\bar z}\psi^1_{\bar z}+3v\psi^2_z-3u\bar v\psi^1+\tfrac32v_z\psi^2,
\end{gather*}
see, respectively, the equations (2.11),~(2.12) and~(2.11),~(2.17) together with the second equation from~(2.16) in~\cite{taim1997a}.
In a nonspecific form, the defocusing counterpart
\begin{gather}\label{eq:ModNyzhnykSystem2b}
u_t=u_{xxx}+u_{yyy}-3u_xv-3u_yw-\tfrac32u(v_x+w_y),\quad v_y=(u^2)_x,\quad w_x=(u^2)_y
\end{gather}
of the modified Nyzhnyk system~\eqref{eq:ModNyzhnykSystem2} is given in \cite[Eq.~(3)]{fera1999c}; see also \cite[Eq.~(9.142)]{roge2002A}.

By introducing an additional parameter $\varepsilon\ne0$ into the system~\eqref{eq:BogdanovSystem} by means of the scaling transformation $\tilde t=\varepsilon t$, $\tilde z=\varepsilon z$, $\tilde f=f$ and $\tilde g=g$, omitting the tildes and performing the limiting process $\varepsilon\rightarrow0$, one can also obtain the dispersionless counterpart of this system:
\begin{gather*}\label{eq:dBogdanovSystem}
\begin{split}&
f_t
+3f_z\p_{\bar z}^{-1}(fg)_z+3f_{\bar z}\p_z^{-1}(fg)_{\bar z}
+3f\p_z^{-1}(gf_{\bar z})_{\bar z}+3f\p_{\bar z}^{-1}(gf_z)_z
=0,\\&
g_t
+3g_z\p_{\bar z}^{-1}(fg)_z+3g_{\bar z}\p_z^{-1}(fg)_{\bar z}
+3g\p_z^{-1}(fg_{\bar z})_{\bar z}+3g\p_{\bar z}^{-1}(fg_z)_z
=0.
\end{split}
\end{gather*}
Analogously, the dispersionless counterpart of the system~\eqref{eq:ModNyzhnykSystem2} can be obtained by trivially extending the scaling transformation to~$v$.

Recently, using an original method involving contact vector fields, a (1+3)-dimensional integrable generalization \cite[Eq.~(21)]{serg2025a} of the asymmetric dispersionless Nyzhnyk system~\eqref{eq:AsymdNSystem} was constructed.
Under certain additional differential constraints, this generalization reduces to the aforementioned system.

To conclude the overview of Nyzhnyk models, we note that unsuccessful attempts to generalize Nyzhnyk models within the same 1+2 dimensions are also encountered in the literature, for instance, to systems of the form
\begin{gather}\label{eq:QuasigenNyzhnykSystem}
w_t+aw_{xxx}+bw_{yyy}+cw_x+dw_y=3a(v^1w)_x+3b(v^2w)_y,\quad
v^1_y=w_x, \quad v^2_x=w_y,
\end{gather}
where $a$, $b$, $c$ and $d$ are arbitrary constants with $(a,b)\ne(0,0)$.
If $ab\ne0$, a simple point transformation
\[
\tilde t=t,\ \,
\tilde x=\frac{-x}{a^{1/3}},\ \,
\tilde y=\frac{-y}{b^{1/3}},\ \,
\tilde u=(ab)^{1/3}u,\ \,
\tilde v_1=-a^{2/3}\left(v_1-\frac c{3a}\right),\ \,
\tilde v_2=-b^{2/3}\left(v_2-\frac d{3b}\right)
\]
reduces the system~\eqref{eq:QuasigenNyzhnykSystem} to the canonical form~\eqref{eq:SymNyzhnykSystem} of the symmetric Nyzhnyk system.
If $ab=0$, the system~\eqref{eq:QuasigenNyzhnykSystem} can be obviously transformed into the reduced asymmetric Nyzhnyk system~\eqref{eq:AsymNyzhnykSystemWithoutV2}.
Indeed, without loss of generality and up to a permutation of variables, we can assume that $a\ne0$ and $b=0$.
Then the last equation of the system~\eqref{eq:QuasigenNyzhnykSystem} can be neglected and the system composed of its first two equations, which no longer contains the dependent variable~$v^2$, is mapped to~\eqref{eq:AsymNyzhnykSystemWithoutV2} by the simple point transformation
\[
\tilde t=t,\ \,
\tilde x=\frac{-x}{a^{1/3}},\ \,
\tilde y=y-dt,\ \,
\tilde u=(ab)^{1/3}u,\ \,
\tilde v_1=-a^{2/3}\left(v_1-\frac c{3a}\right).
\]

\section{Maximal Lie invariance pseudoalgebras}\label{sec:LieInvAlgebra}

For various specific (1+2)-dimensional partial differential equations and systems thereof
that are related to the system~\eqref{eq:NyzhnykSystem} and presented in Section~\ref{sec:NyzhnykSystemsAndEquations},
we compute their maximal Lie invariance pseudoalgebras
(or present these pseudoalgebras if they are already known).
We also find relations between these pseudoalgebras arising from the connections between the corresponding equations and systems.
All computations were carried out using the specialized package \textsf{DESOLV}~\cite{carm2000a,vu2012a}
and additionally verified by simultaneous computations in two other \textsf{Maple}-based packages:
the built-in package \textsf{PDEtools} and the specialized package \textsf{Jets}~\cite{BaranMarvan,marv2009a}.
The description of methods for computing Lie symmetries can be found, e.g., in~\cite{blum2009A,olve1993A}.

For convenience, we recall the concepts related to Lie pseudoalgebras of vector fields.
A \emph{partial vector field} on a manifold~$M$ is a vector field defined on some open subset of the manifold~$M$.
Within the framework of the theory of partial mappings,
a linear combination or the Lie bracket of partial vector fields is naturally defined
as the analogous linear combination or the Lie bracket
of the restrictions of these partial vector fields to the intersection of their domains.

\begin{definition}\label{def:Pseudoalgebra}
A \emph{Lie pseudoalgebra of vector fields} on a manifold~$M$
is a set~$\mathfrak g$ of partial vector fields on the manifold~$M$
that is closed under linear combinations and the Lie bracket
and additionally satisfies the following conditions:
\begin{enumerate}\itemsep=0ex
\item
The zero vector field on the manifold~$M$ belongs to~$\mathfrak g$.
\item
The restriction of any element from~$\mathfrak g$
to any open subset of its domain also belongs to~$\mathfrak g$.
\item
For any collection $\{U_\alpha\}$ of open subsets of the manifold~$M$,
a vector field defined on~$\bigcup_\alpha U_\alpha$ belongs to~$\mathfrak g$
if and only if for each~$\alpha$ its restriction to~$U_\alpha$ belongs to~$\mathfrak g$.
\end{enumerate}
\end{definition}

\begin{remark}\label{rem:OnConsistenceOfLieSyms}
As can be seen from the subsequent consideration,
the maximal Lie invariance pseudoalgebras of the investigated Nyzhnyk models are consistent
with respect to different types of connections and transitions between these models, such as
those between the dispersive and dispersionless cases,
between the (potential) equations and the corresponding systems,
between the nonlinear and linear Lax representations and
between the equations or systems and their Lax representations.
All connections between the corresponding pseudoalgebras are expected and trivial.
Thus, under the limiting process from a dispersive to a dispersionless model,
the pseudoalgebra extends by one generating vector field associated with scaling transformations
with respect to the spatial and dependent variables.
In the course of the transition from a dispersive or dispersionless model
to its linear or nonlinear Lax representation, respectively,
one only needs to prolong the elements of the pseudoalgebra to the pseudopotential as a new dependent variable
and add one more generating vector field associated with scaling transformations or shifts
with respect to the pseudopotential.
\end{remark}

\begin{remark}\label{rem:OnRealAndComplexCasesForLieSyms}
In the complex case, the generating vector fields for the Lie invariance pseudoalgebra of each of the considered models formally have the same form as their real counterparts.
Only their interpretation differs: all involved quantities are assumed to be complex and the smoothness of the parameter functions converts into their analyticity.
This is why the corresponding results are omitted.
\end{remark}

\begin{notation}
Throughout this section, the parameter functions~$\tau$, $\chi$, $\rho$, $\alpha$, $\beta$ and~$\sigma$ run through the set of smooth functions of~$t$, the parameter function~$\gamma$ runs through the set of smooth functions of~$y$ and the parameter function~$\varsigma$ runs through the set of smooth functions of~$(t,y)$.
By $\mathbb F$ and~$\tilde{\mathbb F}$ we denote the fields for the independent and dependent variables, respectively.
\end{notation}

\subsection{Dispersionless symmetric case}

It is convenient to begin the consideration of Lie symmetries of Nyzhnyk models with the dispersionless symmetric case
as the only one studied within the framework of classical symmetry analysis.
The dispersionless Nyzhnyk equation was considered in~\cite{moro2021a},
but its classical symmetry analysis therein was carried out partially and with a number of flaws.
In~\cite{boyk2024a}, we computed
the maximal Lie invariance pseudoalgebras of the dispersionless Nyzhnyk equation~\eqref{eq:dN},
its nonlinear Lax representation~\eqref{eq:dNLaxPair}
and the dispersionless Nyzhnyk system~\eqref{eq:dNSystem}.
All these pseudoalgebras are infinite-dimensional.

The maximal Lie invariance pseudoalgebra~$\mathfrak g_{\ref{eq:dN}}$
of the (real symmetric potential) dispersionless Nyzhnyk equation~\eqref{eq:dN} is spanned by the vector fields
\begin{gather}\label{eq:dNMIA}
\begin{split}&
D^t(\tau)=\tau\p_t+\tfrac13\tau_tx\p_x+\tfrac13\tau_ty\p_y-\tfrac1{18}\tau_{tt}(x^3+y^3)\p_u,\quad
D^{\rm s}=x\p_x+y\p_y+3u\p_u,\\ &
P^x(\chi)=\chi\p_x-\tfrac12\chi_tx^2\p_u,\quad
P^y(\rho)=\rho\p_y-\tfrac12\rho_ty^2\p_u,\\ &
R^x(\alpha)=\alpha x\p_u,\quad
R^y(\beta)=\beta y\p_u,\quad
Z(\sigma)=\sigma\p_u
\end{split}
\end{gather}
in the space with coordinates $(t,x,y,u)$.
Furthermore, the maximal contact invariance pseudoalgebra of this equation coincides with the first prolongation of the pseudoalgebra~$\mathfrak g_{\ref{eq:dN}}$.

Since the dependent variables of the equation~\eqref{eq:dN}
and of the systems~\eqref{eq:dNSystem}, \eqref{eq:dNLaxPair} and~\eqref{eq:LaxdNSystem}
are related in a nonlocal way, the maximal Lie invariance pseudoalgebras of these systems
cannot be directly obtained from their counterpart~$\mathfrak g_{\ref{eq:dN}}$ or from each other.
Hence, each of them should be computed independently.

The maximal Lie invariance pseudoalgebra~$\mathfrak g_{\ref{eq:dNLaxPair}}$ of the nonlinear Lax representation~\eqref{eq:dNLaxPair} of the equation~\eqref{eq:dN} is spanned by the vector fields \cite[Section~6]{boyk2024a}
\begin{gather*}
\begin{split}&
\bar D^t(\tau)=\tau\p_t+\tfrac13\tau_tx\p_x+\tfrac13\tau_ty\p_y-\tfrac1{18}\tau_{tt}(x^3+y^3)\p_u,\quad
\bar D^{\rm s}=x\p_x+y\p_y+3u\p_u+\tfrac32\vartheta\p_\vartheta,\\ &
\bar P^x(\chi)=\chi\p_x-\tfrac12\chi_tx^2\p_u,\quad
\bar P^y(\rho)=\rho\p_y-\tfrac12\rho_ty^2\p_u,\\ &
\bar R^x(\alpha)=\alpha x\p_u,\quad
\bar R^y(\beta)=\beta y\p_u,\quad
\bar Z(\sigma)=\sigma\p_u,\quad
\bar P^\vartheta=\p_\vartheta
\end{split}
\end{gather*}
in the space with coordinates $(t,x,y,u,\vartheta)$.
Thus, similarly to the pseudoalgebra~$\mathfrak g_{\ref{eq:dN}}$,
the pseudoalgebra~$\mathfrak g_{\ref{eq:dNLaxPair}}$ is infinite-dimensional
and it can be obtained from~$\mathfrak g_{\ref{eq:dN}}$ by prolonging the vector fields from~$\mathfrak g_{\ref{eq:dN}}$ to the additional dependent variable~$\vartheta$ and supplementing the prolonged pseudoalgebra with the vector field~$\bar P^\vartheta$.
In other words, the prolongation of elements of the pseudoalgebra~$\mathfrak g_{\ref{eq:dN}}$ to~$\vartheta$ generates a monomorphism of the pseudoalgebra~$\mathfrak g_{\ref{eq:dN}}$ into the pseudoalgebra~$\mathfrak g_{\ref{eq:dNLaxPair}}$.
This prolongation is trivial for all vector fields~\eqref{eq:dNMIA} except for~$D^{\rm s}$.
The appearance of~$\bar P^\vartheta$ is natural and is due to the fact that the unknown function~$\vartheta$ is defined up to a constant summand.
Therefore, it can be stated that the maximal Lie invariance pseudoalgebra~$\mathfrak g_{\ref{eq:dNLaxPair}}$ of the system~\eqref{eq:dNLaxPair} is induced by the maximal Lie invariance pseudoalgebra~$\mathfrak g_{\ref{eq:dN}}$ of the equation~\eqref{eq:dN}.

The maximal Lie invariance pseudoalgebra~$\mathfrak g_{\ref{eq:dNLaxPairPavlov}}$
of the alternative form~\eqref{eq:dNLaxPairPavlov} of the nonlinear Lax representation of the equation~\eqref{eq:dN}
can be obtained from~$\mathfrak g_{\ref{eq:dNLaxPair}}$
using the standard prolongation to the jet variable~$\vartheta_x:=\varrho$
and the natural projection onto the space with coordinates $(t,x,y,u,\varrho)$,
\begin{gather*}
\begin{split}&
\tilde D^t(\tau)=\tau\p_t+\tfrac13\tau_tx\p_x+\tfrac13\tau_ty\p_y-\tfrac1{18}\tau_{tt}(x^3+y^3)\p_u
-\tfrac13\tau_t\varrho\p_\varrho,\\ &
\tilde D^{\rm s}=x\p_x+y\p_y+3u\p_u+\tfrac12\varrho\p_\varrho,\quad
\tilde P^x(\chi)=\chi\p_x-\tfrac12\chi_tx^2\p_u,\quad
\tilde P^y(\rho)=\rho\p_y-\tfrac12\rho_ty^2\p_u,\\ &
\tilde R^x(\alpha)=\alpha x\p_u,\quad
\tilde R^y(\beta)=\beta y\p_u,\quad
\tilde Z(\sigma)=\sigma\p_u.
\end{split}
\end{gather*}
The pseudoalgebra~$\mathfrak g_{\ref{eq:dNLaxPairPavlov}}$ can also be viewed
as a prolongation of the pseudoalgebra~$\mathfrak g_{\ref{eq:dN}}$ to the additional dependent variable~$\varrho$,
and the prolongation generates an isomorphism between the specified pseudoalgebras.

The maximal Lie invariance pseudoalgebra~$\mathfrak g_{\ref{eq:dNSystem}}$ of the dispersionless Nyzhnyk system~\eqref{eq:dNSystem}
is spanned by the vector fields \cite[Section~7]{boyk2024a}
\begin{gather*}
\hat D^t(\tau)=\tau\p_t+\tfrac13\big(\tau_tx\p_x+\tau_ty\p_y
-2\tau_tw\p_w-(2\tau_tv^1+\tau_{tt}x)\p_{v^1}-(2\tau_tv^2+\tau_{tt}y)\p_{v^2}\big),\\
\hat D^{\rm s}=x\p_x+y\p_y+w\p_w+v^1\p_{v^1}+v^2\p_{v^2},\quad
\hat P^x(\chi)=\chi\p_x-\chi_t\p_{v^1},\quad
\hat P^y(\rho)=\rho\p_y-\rho_t\p_{v^2}
\end{gather*}
in the space with coordinates $(t,x,y,w,v^1,v^2)$.
Each Lie-symmetry vector field of the system~\eqref{eq:dNSystem} is induced by a Lie-symmetry vector field of the equation~\eqref{eq:dN}.
More precisely, the corresponding mapping $\mathcal M_*\colon\mathfrak g_{\ref{eq:dN}}\to\mathfrak g_{\ref{eq:dNSystem}}$,
which is a composition of the second prolongation and the natural projection from the second-order jet space
over the base space $\mathbb F^3_{t,x,y}\times\tilde{\mathbb F}_u$ onto the space with coordinates $(t,x,y,w,v^1,v^2)$
under the identification $(w,v^1,v^2)=(u_{xy},u_{xx},u_{yy})$, is an epimorphism of Lie pseudoalgebras with
$\ker\mathcal M_*=\big\langle R^x(\alpha),R^y(\beta),Z(\sigma)\big\rangle$,
$\hat D^t(\tau)=\mathcal M_*D^t(\tau)$,
$\hat D^{\rm s}=\mathcal M_*D^{\rm s}$,
$\hat P^x(\chi)=\mathcal M_*P^x(\chi)$ and
$\hat P^y(\rho)=\mathcal M_*P^y(\rho)$.

The maximal Lie invariance pseudoalgebra~$\mathfrak g_{\ref{eq:LaxdNSystem}}$ of the nonlinear Lax representation~\eqref{eq:LaxdNSystem} of the system~\eqref{eq:dNSystem} is spanned by the vector fields
\begin{gather*}
\check D^t(\tau)=\tau\p_t+\tfrac13\big(\tau_tx\p_x+\tau_ty\p_y
-2\tau_tw\p_w-(2\tau_tv^1+\tau_{tt}x)\p_{v^1}-(2\tau_tv^2+\tau_{tt}y)\p_{v^2}\big),\\
\check D^{\rm s}=x\p_x+y\p_y+w\p_w+v^1\p_{v^1}+v^2\p_{v^2}+\tfrac32\vartheta\p_\vartheta,\\
\check P^x(\chi)=\chi\p_x-\chi_t\p_{v^1},\quad
\check P^y(\rho)=\rho\p_y-\rho_t\p_{v^2},\quad
\check P^\vartheta=\p_\vartheta
\end{gather*}
in the space with coordinates $(t,x,y,w,v^1,v^2,\vartheta)$.
Analogously to the relation between the pseudoalgebras~$\mathfrak g_{\ref{eq:dNLaxPair}}$ and~$\mathfrak g_{\ref{eq:dN}}$,
the pseudoalgebra~$\mathfrak g_{\ref{eq:LaxdNSystem}}$ can be obtained from the pseudoalgebra~$\mathfrak g_{\ref{eq:dNSystem}}$
by prolonging its elements to the additional dependent variable~$\vartheta$
and supplementing it with the vector field~$\check P^\vartheta$,
and among the generating elements, only the prolongation of the vector field~$\hat D^{\rm s}$ is nontrivial.
In other words, the maximal Lie invariance pseudoalgebra~$\mathfrak g_{\ref{eq:LaxdNSystem}}$ of the system~\eqref{eq:LaxdNSystem}
is induced by the maximal Lie invariance pseudoalgebra~$\mathfrak g_{\ref{eq:dNSystem}}$ of the system~\eqref{eq:dNSystem}.

The maximal Lie invariance pseudoalgebra~$\mathfrak g_{\ref{eq:LaxdNSystemNonisospectral}}$ of the linear nonisospectral Lax representation~\eqref{eq:LaxdNSystemNonisospectral} of the system~\eqref{eq:dNSystem} is spanned by the vector fields
\begin{gather*}
\begin{split}&
\breve D^t(\tau)=\tau\p_t+\tfrac13\big(\tau_tx\p_x+\tau_ty\p_y
-2\tau_tw\p_w-(2\tau_tv^1+\tau_{tt}x)\p_{v^1}-(2\tau_tv^2+\tau_{tt}y)\p_{v^2}-\tau_tp\p_p\big),\\&
\breve D^{\rm s}=x\p_x+y\p_y+w\p_w+v^1\p_{v^1}+v^2\p_{v^2}+\tfrac12p\p_p,\\&
\breve P^x(\chi)=\chi\p_x-\chi_t\p_{v^1},\quad
\breve P^y(\rho)=\rho\p_y-\rho_t\p_{v^2},\quad
\breve D^\xi(\zeta)=\zeta(\xi)\p_\xi
\end{split}
\end{gather*}
in the space with coordinates $(t,x,y,p,w,v^1,v^2,\xi)$.
This is the only pseudoalgebra among the maximal Lie invariance pseudoalgebras of Nyzhnyk models
considered in the present paper whose computing with specialized symbolic computation packages encounters difficulties.
To circumvent them, we consider an equivalent system instead of the representation~\eqref{eq:LaxdNSystemNonisospectral}.
More specifically, the second equation of the system~\eqref{eq:LaxdNSystemNonisospectral} should be solved with respect to~$\xi_p$,
the resulting expression substituted into the first equation of this system,
and at least two equations of the initial system~\eqref{eq:dNSystem} should be additionally appended to the modified system.

In turn, the maximal Lie invariance pseudoalgebra~$\mathfrak g_{\ref{eq:dNLaxPairNonisospectral}}$
of the linear nonisospectral Lax representation~\eqref{eq:dNLaxPairNonisospectral} of the equation~\eqref{eq:dN}
is spanned by the vector fields
\begin{gather*}
\grave D^t(\tau)=\tau\p_t+\tfrac13\tau_tx\p_x+\tfrac13\tau_ty\p_y
-\tfrac1{18}\tau_{tt}(x^3+y^3)\p_u-\tfrac13\tau_tp\p_p,\\
\grave D^{\rm s}=x\p_x+y\p_y+3u\p_u+\tfrac12p\p_p,\\
\grave P^x(\chi)=\chi\p_x-\tfrac12\chi_tx^2\p_u,\quad
\grave P^y(\rho)=\rho\p_y-\tfrac12\rho_ty^2\p_u,\\
\grave R^x(\alpha)=\alpha x\p_u,\quad
\grave R^y(\beta)=\beta y\p_u,\quad
\grave Z(\sigma)=\sigma\p_u,\quad
\grave D^\xi(\zeta)=\zeta(\xi)\p_\xi
\end{gather*}
in the space with coordinates $(t,x,y,p,u,\xi)$.
As can be seen, the pseudoalgebras~$\mathfrak g_{\ref{eq:dNSystem}}$ and~$\mathfrak g_{\ref{eq:dN}}$
are respectively embedded into the pseudoalgebras~$\mathfrak g_{\ref{eq:LaxdNSystemNonisospectral}}$ and~$\mathfrak g_{\ref{eq:dNLaxPairNonisospectral}}$
as their pseudosubalgebras via the prolongation to~$p$ and~$\xi$.
Moreover, this prolongation is completely trivial with respect to~$\xi$ and is nontrivial
with respect to~$p$ only for the vector fields~$\breve D^t(\tau)$, $\breve D^{\rm s}$ and~$\grave D^t(\tau)$, $\grave D^{\rm s}$,
which are associated with the (generalized) scaling transformations.
Linear complements to these pseudosubalgebras are spanned by the vector fields~$\breve D^\xi(\zeta)$ and~$\grave D^\xi(\zeta)$.

In addition to the relations within the pairs of pseudoalgebras
$(\mathfrak g_{\ref{eq:dN}},\mathfrak g_{\ref{eq:dNLaxPair}})$,
$(\mathfrak g_{\ref{eq:dNSystem}},\mathfrak g_{\ref{eq:LaxdNSystem}})$ and
$(\mathfrak g_{\ref{eq:dN}},\mathfrak g_{\ref{eq:dNSystem}})$,
there also exists a relation between the pseudoalgebras~$\mathfrak g_{\ref{eq:dNLaxPair}}$ and~$\mathfrak g_{\ref{eq:LaxdNSystem}}$
analogous to that within the pair $(\mathfrak g_{\ref{eq:dN}},\mathfrak g_{\ref{eq:dNSystem}})$.
The corresponding mapping
$\bar{\mathcal M}_*\colon\mathfrak g_{\ref{eq:dNLaxPair}}\to\mathfrak g_{\ref{eq:LaxdNSystem}}$
is the composition of the second prolongation
and the natural projection from the second-order jet space over the base space \smash{$\mathbb F^3_{t,x,y}\times\tilde{\mathbb F}^2_{u,\vartheta}$}
onto the space with coordinates $(t,x,y,w,v^1,v^2,\vartheta)$ under the identification $(w,v^1,v^2)=(u_{xy},u_{xx},u_{yy})$.
This mapping is an epimorphism of Lie pseudoalgebras,
meaning that each Lie-symmetry vector field of the system~\eqref{eq:LaxdNSystem}
is induced by a Lie-symmetry vector field of the system~\eqref{eq:dNLaxPair} with
$\ker\bar{\mathcal M}_*=\big\langle\bar R^x(\alpha),\bar R^y(\beta),\bar Z(\sigma)\big\rangle$,
$\check D^t(\tau)=\bar{\mathcal M}_*\bar D^t(\tau)$, $\check D^{\rm s}=\bar{\mathcal M}_*\bar D^{\rm s}$,
$\check P^x(\chi)=\bar{\mathcal M}_*\bar P^x(\chi)$,
$\check P^y(\rho)=\bar{\mathcal M}_*\bar P^y(\rho)$ and $\check P^\vartheta=\bar{\mathcal M}_*\bar P^\vartheta$.

\subsection{Dispersive symmetric case}\label{subsec:PseudoalgebraDispersionSymm}

The maximal Lie invariance pseudoalgebras of the equation~(\ref{eq:SymPotNizhEq}$'$),
its Lax representation~\eqref{eq:LaxSymPotNizhEq'}, the system~(\ref{eq:SymNyzhnykSystem}$'$)
and its Lax representation~\eqref{eq:LaxSymNyzhnykSystem'}
are the (infinite-dimensional) pseudosubalgebras of the pseudoalgebras~$\mathfrak g_{\ref{eq:dN}}$,
$\mathfrak g_{\ref{eq:dNLaxPair}}$, $\mathfrak g_{\ref{eq:dNSystem}}$ and~$\mathfrak g_{\ref{eq:LaxdNSystem}}$
that are spanned by the same generating vector fields as the original pseudoalgebras,
except for $D^{\rm s}$, $\bar D^{\rm s}$, $\hat D^{\rm s}$ and $\check D^{\rm s}$, respectively:%
\footnote{\label{fnt:EmbeddingsOfMIAsOfLaxRepresentations}%
For the Lax representations~\eqref{eq:LaxSymPotNizhEq'} and~\eqref{eq:LaxSymNyzhnykSystem'},
these are actually embeddings as pseudosubalgebras under the identification $\psi=\pm{\rm e}^\vartheta$,
where the vector fields~$\bar P^\vartheta$ and~$\check P^\vartheta$
are respectively replaced by the vector fields~$\bar D^{\rm\psi}$ and~$\check D^{\rm\psi}$,
which are defined on the spaces with coordinates $(t,x,y,u,\psi)$ and $(t,x,y,w,v^1,v^2,\psi)$
and formally have the same form $\psi\p_\psi$.
}
\begin{gather*}
\mathfrak g_{\ref{eq:SymPotNizhEq}'}:=\langle D^t(\tau),\,P^x(\chi),\,P^y(\rho),\,
R^x(\alpha),\,R^y(\beta),\,Z(\sigma)\rangle,
\\
\mathfrak g_{\text{\ref{eq:LaxSymPotNizhEq'}}}:=\langle\bar D^t(\tau),\,\bar P^x(\chi),\,\bar P^y(\rho),\,
\bar R^x(\alpha),\,\bar R^y(\beta),\,\bar Z(\sigma),\, \bar D^{\rm\psi}\rangle,
\\
\mathfrak g_{\ref{eq:SymNyzhnykSystem}'}:=\langle\hat D^t(\tau),\,\hat P^x(\chi),\,\hat P^y(\rho)\rangle,
\\
\mathfrak g_{\text{\ref{eq:LaxSymNyzhnykSystem'}}}:=\langle\check D^t(\tau),\,\check P^x(\chi),\,\check P^y(\rho),\,
\check D^{\rm\psi}\rangle.
\end{gather*}
\noprint{
$\mathfrak g_{\ref{eq:SymPotNizhEq}}$
\begin{gather*}
\begin{split}&
D^t(\tau)=\tau\p_t+\tfrac13\tau_tx\p_x+\tfrac13\tau_ty\p_y-\tfrac1{54}\tau_{tt}(x^3+y^3)\p_u,\\&
P^x(\chi)=\chi\p_x-\tfrac16\chi_tx^2\p_u,\quad
P^y(\rho)=\rho\p_y-\tfrac16\rho_ty^2\p_u,\\ &
R^x(\alpha)=\alpha x\p_u,\quad
R^y(\beta)=\beta y\p_u,\quad
Z(\sigma)=\sigma\p_u.
\end{split}
\end{gather*}

$\mathfrak g_{\ref{eq:LaxSymPotNizhEq}}$
\begin{gather*}
\begin{split}&
\bar D^t(\tau)=\tau\p_t+\tfrac13\tau_tx\p_x+\tfrac13\tau_ty\p_y-\tfrac1{54}\tau_{tt}(x^3+y^3)\p_u,\\&
\bar P^x(\chi)=\chi\p_x-\tfrac16\chi_tx^2\p_u,\quad
\bar P^y(\rho)=\rho\p_y-\tfrac16\rho_ty^2\p_u,\\ &
\bar R^x(\alpha)=\alpha x\p_u,\quad
\bar R^y(\beta)=\beta y\p_u,\quad
\bar Z(\sigma)=\sigma\p_u,\quad
\bar P^\psi=\psi\p_\psi.
%\bar P^\psi=\psi\p_\psi-\p_u.
\end{split}
\end{gather*}

$\mathfrak g_{\ref{eq:SymNyzhnykSystem}}$
\begin{gather*}
\hat D^t(\tau)=\tau\p_t+\tfrac13\tau_tx\p_x+\tfrac13\tau_ty\p_y
-\tfrac23\tau_tw\p_w-\tfrac13(2\tau_tv^1+\tfrac13\tau_{tt}x)\p_{v^1}-\tfrac13(2\tau_tv^2+\tfrac13\tau_{tt}y)\p_{v^2},\\
\hat P^x(\chi)=\chi\p_x-\tfrac13\chi_t\p_{v^1},\quad
\hat P^y(\rho)=\rho\p_y-\tfrac13\rho_t\p_{v^2}
\end{gather*}

$\mathfrak g_{\ref{eq:LaxSymNyzhnykSystem}}$
\begin{gather*}
\check D^t(\tau)=\tau\p_t+\tfrac13\tau_tx\p_x+\tfrac13\tau_ty\p_y
-\tfrac23\tau_tw\p_w-\tfrac13(2\tau_tv^1+\tfrac13\tau_{tt}x)\p_{v^1}
-\tfrac13(2\tau_tv^2+\tfrac13\tau_{tt}y)\p_{v^2},\\
\check P^x(\chi)=\chi\p_x-\tfrac13\chi_t\p_{v^1},\quad
\check P^y(\rho)=\rho\p_y-\tfrac13\rho_t\p_{v^2},\quad
\check P^\psi=\psi\p_\psi.
\end{gather*}
}
More precisely, the pseudoalgebras~$\mathfrak g_{\ref{eq:SymPotNizhEq}'}$, $\mathfrak g_{\text{\ref{eq:LaxSymPotNizhEq'}}}$, $\mathfrak g_{\ref{eq:SymNyzhnykSystem}'}$ and~$\mathfrak g_{\text{\ref{eq:LaxSymNyzhnykSystem'}}}$ contract to the corresponding pseudosubalgebras of the pseudoalgebras~$\mathfrak g_{\ref{eq:dN}}$, $\mathfrak g_{\ref{eq:dNLaxPair}}$, $\mathfrak g_{\ref{eq:dNSystem}}$ and~$\mathfrak g_{\ref{eq:LaxdNSystem}}$, respectively, under the limiting process from the dispersive case to the dispersionless one.
The connections among the pseudoalgebras~$\mathfrak g_{\ref{eq:SymPotNizhEq}'}$, $\mathfrak g_{\text{\ref{eq:LaxSymPotNizhEq'}}}$, $\mathfrak g_{\ref{eq:SymNyzhnykSystem}'}$ and~$\mathfrak g_{\text{\ref{eq:LaxSymNyzhnykSystem'}}}$ are analogous to those among the pseudoalgebras~$\mathfrak g_{\ref{eq:dN}}$, $\mathfrak g_{\ref{eq:dNLaxPair}}$, $\mathfrak g_{\ref{eq:dNSystem}}$ and~$\mathfrak g_{\ref{eq:LaxdNSystem}}$.

Similarly to the dispersionless equation~\eqref{eq:dN}, the maximal contact invariance pseudoalgebra of its dispersive counterpart~(\ref{eq:SymPotNizhEq}$'$) coincides with the first prolongation of the maximal Lie invariance pseudoalgebra~$\mathfrak g_{\ref{eq:SymPotNizhEq}'}$ of this counterpart.

\subsection{Dispersionless asymmetric case}

For the maximal Lie invariance pseudoalgebras of the asymmetric Nyzhnyk models,
we preserve the order of considering different model types.
We immediately note that all these pseudoalgebras are also infinite-dimensional and the connections among them are the same as those among their counterparts in the symmetric case.

The maximal Lie invariance pseudoalgebra~$\mathfrak g_{\ref{eq:AsymdN}}$ of the asymmetric (potential) dispersionless Nyzhnyk equation~\eqref{eq:AsymdN} is spanned by the vector fields
\begin{gather*}
\begin{split}&
D(\tau)=\tau\p_t+\tfrac13\tau_tx\p_x-\tfrac13\tau_tu\p_u-\tfrac16\tau_{tt}x^2\p_u,\quad
S(\gamma)=\gamma\p_y,\quad
D^{\rm s}=x\p_x+2u\p_u,\\ &
P(\chi)=\chi\p_x-\chi_tx\p_u,\quad
Z(\sigma)=\sigma\p_u
\end{split}
\end{gather*}
in the space with coordinates $(t,x,y,u)$. Furthermore, as in the symmetric case,
the maximal contact invariance pseudoalgebra of this equation
coincides with the first prolongation of the pseudoalgebra~$\mathfrak g_{\ref{eq:AsymdN}}$.

The maximal Lie invariance pseudoalgebra~$\mathfrak g_{\ref{eq:AsymdNLaxPair}}$ of the nonlinear Lax representation~\eqref{eq:AsymdNLaxPair} of the equation~\eqref{eq:AsymdN} is spanned by the vector fields
\begin{gather*}
\begin{split}
&\bar D(\tau)=\tau\p_t+\tfrac13\tau_tx\p_x-\tfrac13\tau_tu\p_u-\tfrac16\tau_{tt}x^2\p_u,\quad
\bar S(\gamma)=\gamma\p_y, \quad
\bar D^{\rm s}=x\p_x+2u\p_u+\tfrac32\vartheta\p_\vartheta,\\
&\bar P(\chi)=\chi\p_x-\chi_tx\p_u,\quad
\bar Z(\sigma)=\sigma\p_u,\quad
\bar P^\vartheta=\p_\vartheta
\end{split}
\end{gather*}
in the space with coordinates $(t,x,y,u,\vartheta)$.
The pseudoalgebra~$\mathfrak g_{\ref{eq:AsymdNLaxPair}}$ can be obtained by prolonging the vector fields from the pseudoalgebra~$\mathfrak g_{\ref{eq:AsymdN}}$ to the additional dependent variable~$\vartheta$ and supplementing the prolonged pseudoalgebra with the vector field~$\bar P^\vartheta$.
In other words, the maximal Lie invariance pseudoalgebra~$\mathfrak g_{\ref{eq:AsymdNLaxPair}}$ of the system~\eqref{eq:AsymdNLaxPair} is induced by the maximal Lie invariance pseudoalgebra~$\mathfrak g_{\ref{eq:AsymdN}}$ of the equation~\eqref{eq:AsymdN}.

The maximal Lie invariance pseudoalgebra~$\mathfrak g_{\ref{eq:AsymdNLaxPairNonisospectral}}$ of the linear nonisospectral Lax representation~\eqref{eq:AsymdNLaxPair} of the equation~\eqref{eq:AsymdN} is spanned by the vector fields
\begin{gather*}
\grave D(\tau)=\tau\p_t+\tfrac13\tau_tx\p_x-\tfrac13\tau_tu\p_u-\tfrac16\tau_{tt}x^2\p_u-\tfrac13\tau_tp\p_p,\quad
\grave S(\gamma)=\gamma\p_y,\\
\grave D^{\rm s}=x\p_x+2u\p_u+\tfrac12p\p_p,\quad
\grave P(\chi)=\chi\p_x-\chi_tx\p_u,\quad
\grave Z(\sigma)=\sigma\p_u,\quad
\grave D^\xi(\zeta)=\zeta(\xi)\p_\xi
\end{gather*}
in the space with coordinates $(t,x,y,p,u,\xi)$.
As can be seen, analogously to the symmetric case, this pseudoalgebra corresponds to the pseudoalgebra of the equation~\eqref{eq:AsymdN} with the extension of the base space to the variable~$p$ and the addition of the family of vector fields~$\grave D^\xi(\zeta)$ to the generating set.

The maximal Lie invariance pseudoalgebra~$\mathfrak g_{\ref{eq:AsymdNSystem}}$ of the asymmetric dispersionless Nyzhnyk system~\eqref{eq:AsymdNSystem} is spanned by the vector fields
\begin{gather*}
\hat D(\tau)=\tau\p_t+\tfrac13\tau_tx\p_x-\tfrac13\tau_tw\p_w-\tfrac23\tau_tv^1\p_{v^1}-\tfrac13\tau_{tt}x\p_{v^1},\\
\hat S(\gamma)=\gamma\p_y-\gamma_yw\p_w,\quad
\hat D^{\rm s}=x\p_x+2w\p_w+v^1\p_{v^1},\quad
\hat P(\chi)=\chi\p_x-\chi_t\p_{v^1}
\end{gather*}
in the space with coordinates $(t,x,y,w,v^1)$.
Just as for the pseudoalgebra~$\mathfrak g_{\ref{eq:dNSystem}}$, each Lie-symmetry vector field of the system~\eqref{eq:AsymdNSystem} is induced by a Lie-symmetry vector field of the equation~\eqref{eq:AsymdN}.
More precisely, the corresponding mapping $\mathcal M_*\colon\mathfrak g_{\ref{eq:AsymdN}}\to\mathfrak g_{\ref{eq:AsymdNSystem}}$ is an epimorphism of Lie pseudoalgebras.
It is also a composition of the second prolongation and the projection from the second-order jet space over the base space $\mathbb F^3_{t,x,y}\times\tilde{\mathbb F}_u$ onto the space with coordinates $(t,x,y,w,v^1)$ under the identification $(w,v^1)=(u_y,u_x)$.
Moreover,
$\ker\mathcal M_*=\big\langle Z(\sigma),S(\gamma)\big\rangle$,
$\hat D(\tau)=\mathcal M_*D(\tau)$,
$\hat D^{\rm s}=\mathcal M_*D^{\rm s}$ and
$\hat P(\chi)=\mathcal M_*P(\chi)$.

The maximal Lie invariance pseudoalgebra~$\mathfrak g_{\ref{eq:LaxAsymdNSystem}}$ of the nonlinear Lax representation~\eqref{eq:LaxAsymdNSystem} of the system~\eqref{eq:AsymdNSystem} is spanned by the vector fields
\begin{gather*}
\check D(\tau)=\tau\p_t+\tfrac13\tau_tx\p_x-\tfrac13\tau_tw\p_w-\tfrac23\tau_tv^1\p_{v^1}-\tfrac13\tau_{tt}x\p_{v^1},\quad
\check P^\vartheta=\p_\vartheta,\\
\check D^{\rm s}=x\p_x+2w\p_w+v^1\p_{v^1}+\tfrac32\vartheta\p_\vartheta,\quad
\check P(\chi)=\chi\p_x-\chi_t\p_{v^1},\quad
\check S(\gamma)=\gamma\p_y-\gamma_yw\p_w
\end{gather*}
in the space with coordinates $(t,x,y,w,v^1,\vartheta)$.
In accordance with the connection between the pseudoalgebras~$\mathfrak g_{\ref{eq:AsymdN}}$ and~$\mathfrak g_{\ref{eq:AsymdNLaxPair}}$, the pseudoalgebra~$\mathfrak g_{\ref{eq:LaxAsymdNSystem}}$ can be obtained from the pseudoalgebra~$\mathfrak g_{\ref{eq:AsymdNSystem}}$ by prolonging its elements to the additional dependent variable~$\vartheta$ and supplementing it with the vector field~$\check P^\vartheta$.
Furthermore, among the generating elements, only the prolongation of the vector field~$\hat D^{\rm s}$ is nontrivial.
That is, the maximal Lie invariance pseudoalgebra~$\mathfrak g_{\ref{eq:LaxAsymdNSystem}}$ of the system~\eqref{eq:LaxAsymdNSystem} is induced by the maximal Lie invariance pseudoalgebra~$\mathfrak g_{\ref{eq:AsymdNSystem}}$ of the system~\eqref{eq:AsymdNSystem}.

Analogously to the connection between $\mathfrak g_{\ref{eq:AsymdN}}$ and $\mathfrak g_{\ref{eq:AsymdNSystem}}$, there also exists a connection between the pseudoalgebras~$\mathfrak g_{\ref{eq:AsymdNLaxPair}}$ and~$\mathfrak g_{\ref{eq:LaxAsymdNSystem}}$.
The corresponding mapping $\bar{\mathcal M}_*\colon\mathfrak g_{\ref{eq:AsymdNLaxPair}}\to\mathfrak g_{\ref{eq:LaxAsymdNSystem}}$, which is a composition of the second prolongation and the projection from the second-order jet space over the base space $\mathbb F^3_{t,x,y}\times\tilde{\mathbb F}^2_{u,\vartheta}$ onto the space with coordinates $(t,x,y,w,v^1,\vartheta)$ under the identification $(w,v^1)=(u_y,u_x)$, is an epimorphism of Lie pseudoalgebras, meaning that each Lie-symmetry vector field of the system~\eqref{eq:LaxAsymdNSystem} is induced by a Lie-symmetry vector field of the system~\eqref{eq:AsymdNLaxPair}.
More precisely,
$\check D(\tau)=\bar{\mathcal M}_*\bar D(\tau)$, $\check D^{\rm s}=\bar{\mathcal M}_*\bar D^{\rm s}$,
$\check P(\chi)=\bar{\mathcal M}_*\bar P(\chi)$,
$\check P^\vartheta=\bar{\mathcal M}_*\bar P^\vartheta$ and
$\ker\bar{\mathcal M}_*=\big\langle\bar Z(\sigma),\bar S(\gamma)\big\rangle$.

The connections of the maximal Lie invariance pseudoalgebra~$\mathfrak g_{\ref{eq:LaxAsymdNSystemNonisospectral}}$ of the linear nonisospectral Lax representation~\eqref{eq:LaxAsymdNSystemNonisospectral} of the system~\eqref{eq:AsymdNSystem} with the pseudoalgebras~$\mathfrak g_{\ref{eq:AsymdNSystem}}$ and~$\mathfrak g_{\ref{eq:AsymdNLaxPairNonisospectral}}$ are also expected and analogous to those discussed above.
The pseudoalgebra~$\mathfrak g_{\ref{eq:LaxAsymdNSystemNonisospectral}}$ is spanned by the vector fields
\begin{gather*}
\breve D(\tau)=\tau\p_t+\tfrac13\tau_tx\p_x-\tfrac13\tau_tw\p_w
-\tfrac23\tau_tv^1\p_{v^1}-\tfrac13\tau_{tt}x\p_{v^1}-\tfrac13\tau_tp\p_p,\quad
\breve S(\gamma)=\gamma\p_y-\gamma_yw\p_w,\\
\breve D^{\rm s}=x\p_x+2w\p_w+v^1\p_{v^1}+\tfrac12p\p_p,\quad
\breve P(\chi)=\chi\p_x-\chi_t\p_{v^1},\quad
\breve D^\xi(\zeta)=\zeta(\xi)\p_\xi
\end{gather*}
in the space with coordinates $(t,x,y,p,w,v^1,\xi)$, so it can be obtained, for instance, from the pseudoalgebra~$\mathfrak g_{\ref{eq:AsymdNSystem}}$ by extending the base space to the variables~$p$ and~$\xi$ and adding the family of vector fields~$\breve D^\xi(\zeta)$ to the generating set.

\subsection{Dispersive asymmetric case}

The maximal Lie invariance pseudoalgebras of the equation~(\ref{eq:AsymPotNizhEqBLMP}$'$), its Lax representation~\eqref{eq:LaxAsymPotNizhEqBLMP'}, the system~(\ref{eq:AsymNyzhnykSystemWithoutV2}$'$) and its Lax representation~\eqref{eq:NyzhnykSystemAsymLax'} are (infinite-dimensional) pseudosubalgebras of the pseudoalgebras~$\mathfrak g_{\ref{eq:AsymdN}}$, $\mathfrak g_{\ref{eq:AsymdNLaxPair}}$, $\mathfrak g_{\ref{eq:AsymdNSystem}}$ and~$\mathfrak g_{\ref{eq:LaxAsymdNSystem}}$, respectively, which are spanned by the same generating vector fields as the original pseudoalgebras, except for $D^{\rm s}$, $\bar D^{\rm s}$, $\hat D^{\rm s}$ and $\check D^{\rm s}$ (for the Lax representations, we additionally assume $\psi={\rm e}^\vartheta$):
\begin{gather*}
\mathfrak g_{\ref{eq:AsymPotNizhEqBLMP}'}:=\langle D(\tau),\,S(\gamma),\,P(\chi),\,Z(\sigma)\rangle,
\\
\mathfrak g_{\text{\ref{eq:LaxAsymPotNizhEqBLMP'}}}:=\langle\bar D(\tau),\,\bar S(\gamma),\,\bar P(\chi),\,
\bar Z(\sigma),\,\bar D^{\rm\psi}\rangle,
\\
\mathfrak g_{\ref{eq:AsymNyzhnykSystemWithoutV2}'}:=\langle\hat D(\tau),\,\hat S(\gamma),\,\hat P(\chi)\rangle,
\\
\mathfrak g_{\text{\ref{eq:NyzhnykSystemAsymLax'}}}:=\langle\check D(\tau),\,\check S(\gamma),\,
\check P(\chi),\,\check D^{\rm\psi}\rangle.
\end{gather*}
\noprint{
Algebra of the equation~\eqref{eq:AsymPotNizhEqBLMP}
\begin{gather*}
\begin{split}
&D(\tau)=\tau\p_t+\tfrac13\tau_tx\p_x-\tfrac13(\tau_tu+\tfrac16\tau_{tt}x^2)\p_u,\quad
P(\chi)=\chi\p_x-\tfrac13\chi_tx\p_u,\quad
Z(\sigma)=\sigma\p_u,\\
&S(\gamma)=\gamma\p_y.
\end{split}
\end{gather*}

Algebra of the Lax representation~\eqref{eq:LaxAsymPotNizhEqBLMP} of the equation~\eqref{eq:AsymPotNizhEqBLMP}
\begin{gather*}
\begin{split}
&\bar D(\tau)=\tau\p_t+\tfrac13\tau_tx\p_x-\tfrac13(\tau_tu+\tfrac16\tau_{tt}x^2)\p_u,\quad
\bar P^x(\chi)=\chi\p_x-\tfrac13\chi_tx\p_u,\quad
\bar Z(\sigma)=\sigma\p_u,\\
&\bar S(\gamma)=\gamma\p_y, \quad
\bar P^\psi=\psi\p_\psi,
\end{split}
\end{gather*}

Algebra of the system~\eqref{eq:AsymNyzhnykSystemWithoutV2}
\begin{gather*}
\begin{split}
&\hat D(\tau)=\tau\p_t+\tfrac13\tau_tx\p_x-\tfrac13\tau_tw\p_w-\tfrac23(\tau_tv^1+\tfrac16\tau_{tt}x)\p_{v^1},\\
&\hat P(\chi)=\chi\p_x-\tfrac13\chi_t\p_{v^1},\quad
\hat S(\gamma)=\gamma\p_y-\gamma_yw\p_w.
\end{split}
\end{gather*}

Algebra of the Lax representation~\eqref{eq:NyzhnykSystemAsymLax} of the system~\eqref{eq:AsymNyzhnykSystemWithoutV2}
\begin{gather*}
\begin{split}
&\check D(\tau)=\tau\p_t+\tfrac13\tau_tx\p_x-\tfrac13\tau_tw\p_w
-\tfrac23(\tau_tv^1+\tfrac16\tau_{tt}x)\p_{v^1},\\ &
\check P(\chi)=\chi\p_x-\tfrac13\chi_t\p_{v^1},\quad
\check S(\gamma)=\gamma\p_y-\gamma_yw\p_w,\quad
\check P^\psi=\psi\p_\psi.
\end{split}
\end{gather*}
}
In other words, under the limiting process from the dispersive case to the dispersionless one,
the pseudoalgebras~$\mathfrak g_{\ref{eq:AsymPotNizhEqBLMP}'}$, $\mathfrak g_{\text{\ref{eq:LaxAsymPotNizhEqBLMP'}}}$, $\mathfrak g_{\ref{eq:AsymNyzhnykSystemWithoutV2}'}$ and~$\mathfrak g_{\text{\ref{eq:NyzhnykSystemAsymLax'}}}$ contract to the corresponding pseudosubalgebras of the pseudoalgebras~$\mathfrak g_{\ref{eq:AsymdN}}$, $\mathfrak g_{\ref{eq:AsymdNLaxPair}}$, $\mathfrak g_{\ref{eq:AsymdNSystem}}$ and~$\mathfrak g_{\ref{eq:LaxAsymdNSystem}}$. The connections among the pseudoalgebras~$\mathfrak g_{\ref{eq:AsymPotNizhEqBLMP}'}$, $\mathfrak g_{\text{\ref{eq:LaxAsymPotNizhEqBLMP'}}}$, $\mathfrak g_{\ref{eq:AsymNyzhnykSystemWithoutV2}'}$ and~$\mathfrak g_{\text{\ref{eq:NyzhnykSystemAsymLax'}}}$ are analogous to those among the pseudoalgebras~$\mathfrak g_{\ref{eq:AsymdN}}$, $\mathfrak g_{\ref{eq:AsymdNLaxPair}}$, $\mathfrak g_{\ref{eq:AsymdNSystem}}$ and~$\mathfrak g_{\ref{eq:LaxAsymdNSystem}}$.

Similarly to the symmetric case and to the asymmetric dispersionless equation~\eqref{eq:AsymdN},
the maximal contact invariance pseudoalgebra of its dispersive counterpart~(\ref{eq:AsymPotNizhEqBLMP}$'$)
also coincides with the first prolongation of the maximal Lie invariance pseudoalgebra~$\mathfrak g_{\ref{eq:AsymPotNizhEqBLMP}'}$
of this counterpart.

\allowdisplaybreaks
\subsection{Specific case}

To find the maximal Lie invariance pseudoalgebras of the specific Nyzhnyk models,
one should perform the following actions with the generating vector fields
of the maximal Lie invariance pseudoalgebras of the corresponding symmetric Nyzhnyk models:
\begin{itemize}\itemsep=0ex
\item
carry out the same formal substitutions as those between the corresponding models, extending them to the pairs of functions $(\chi,\rho)$ and $(\alpha,\beta)$ that parameterize these vector fields, namely, $x\to z$, $y\to\bar z$, $v^1\to v$, $v^2\to\bar v$, $\chi\to\zeta$, $\rho\to\bar\zeta$, $\alpha\to\gamma$ and $\beta\to\bar\gamma$, where all the quantities on the right are complex-valued;
\item
substitute the expressions for these complex-valued quantities and their derivatives in terms of their real and imaginary parts, i.e.,
$z=x+{\rm i}y$, $v=v^1+{\rm i}v^2$, $\zeta=\chi+{\rm i}\rho$, $\gamma=\alpha+{\rm i}\beta$, $\p_z=\frac12(\p_x-\mathrm i\p_y)$ and $\p_v=\frac12(\p_{v^1}-\mathrm i\p_{v^2})$;
\item
linearly recombine the resulting vector fields to single out their real and imaginary parts.
\end{itemize}
All the obtained Lie invariance pseudoalgebras were verified by the direct computation of Lie symmetries of the specific Nyzhnyk models~\eqref{eq:SpecNyzhnykSystem}--\eqref{eq:LaxSpecPotdNEq} (including the corresponding linear or nonlinear Lax representations)
simultaneously in three different packages for \textsf{Maple}: the built-in \textsf{PDEtools} package and the specialized packages \textsf{DESOLV}~\cite{carm2000a,vu2012a} and \textsf{Jets}~\cite{BaranMarvan,marv2009a}.
These pseudoalgebras have the following form:
\begin{gather*}
\mathfrak g_{\ref{eq:SpecNyzhnykSystem}'}=\langle\hat D^t(\tau),\,\hat P^x(\chi),\,\hat P^y(\rho)\rangle,\\
\mathfrak g_{\ref{eq:SpecdNSystem}'}=\langle\hat D^t(\tau),\,\hat D^{\rm s},\,\hat P^x(\chi),\,\hat P^y(\rho)\rangle,\\
\mathfrak g_{\ref{eq:LaxSpecNyzhnykSystem}'}=\langle\check D^t(\tau),\,\check P^x(\chi),\,\check P^y(\rho),\,\check D^{\rm\psi}\rangle,\\
\mathfrak g_{\ref{eq:LaxSpecdNSystem}'}=\langle\check D^t(\tau),\,\check D^{\rm s},\,\check P^x(\chi),\,\check P^y(\rho),\,\check P^\vartheta\rangle,\\
\mathfrak g_{\ref{eq:SpecPotNizhEq}'}=\langle D^t(\tau),\,P^x(\chi),\,P^y(\rho),\,R^x(\alpha),\,R^y(\beta),\,Z(\sigma)\rangle,\\
\mathfrak g_{\ref{eq:SpecPotdNEq}'}=\langle D^t(\tau),\,D^{\rm s},\,P^x(\chi),\,P^y(\rho),\,R^x(\alpha),\,R^y(\beta),\,Z(\sigma)\rangle,\\
\mathfrak g_{\ref{eq:LaxSpecPotNizhEq}'}=\langle\bar D^t(\tau),\,\bar P^x(\chi),\,\bar P^y(\rho),\,\bar R^x(\alpha),\, \bar R^y(\beta),\,\bar Z(\sigma),\,\bar D^{\rm\psi}\rangle,\\
\mathfrak g_{\ref{eq:LaxSpecPotdNEq}'}=\langle\bar D^t(\tau),\,\bar D^{\rm s},\,\bar P^x(\chi),\,\bar P^y(\rho),\,\bar R^x(\alpha),\,\bar R^y(\beta),\,\bar Z(\sigma),\,\bar P^\vartheta\rangle,
\end{gather*}
where
\begin{gather*}
\hat D^t(\tau)=\tau\p_t+\tfrac13\tau_tx\p_x+\tfrac13\tau_ty\p_y
-\tfrac23\tau_tw\p_w-\tfrac13(2\tau_tv^1+\tfrac13\tau_{tt}x)\p_{v^1}-\tfrac13(2\tau_tv^2+\tfrac13\tau_{tt}y)\p_{v^2},\\
\hat D^{\rm s}=x\p_x+y\p_y+w\p_w+v^1\p_{v^1}+v^2\p_{v^2},\\
\hat P^x(\chi)=\chi\p_x-\tfrac13\chi_t\p_{v^1},\quad
\hat P^y(\rho)=\rho\p_y-\tfrac13\rho_t\p_{v^2},
\\[1.5ex]
\check D^t(\tau)=\tau\p_t+\tfrac13\tau_tx\p_x+\tfrac13\tau_ty\p_y
-\tfrac23\tau_tw\p_w-\tfrac13(2\tau_tv^1+\tfrac13\tau_{tt}x)\p_{v^1}-\tfrac13(2\tau_tv^2+\tfrac13\tau_{tt}y)\p_{v^2},\\
\check D^{\rm s}=x\p_x+y\p_y+w\p_w+v^1\p_{v^1}+v^2\p_{v^2}+\tfrac32\vartheta\p_\vartheta,\\
\check P^x(\chi)=\chi\p_x-\tfrac13\chi_t\p_{v^1},\quad
\check P^y(\rho)=\rho\p_y-\tfrac13\rho_t\p_{v^2},\quad
\check D^\psi=\psi\p_\psi,\quad
\check P^\vartheta=\p_\vartheta,
\\[1.5ex]
D^t(\tau)=\tau\p_t+\tfrac13\tau_tx\p_x+\tfrac13\tau_ty\p_y-\tfrac19\tau_{tt}(\tfrac13x^3-xy^2)\p_u,\quad
D^{\rm s}=x\p_x+y\p_y+3u\p_u,\\
P^x(\chi)=\chi\p_x-\tfrac13\chi_t(x^2-y^2)\p_u,\quad
P^y(\rho)=\rho\p_y+\tfrac23\rho_txy\p_u,\\
R^x(\alpha)=\alpha x\p_u,\quad
R^y(\beta)=\beta y\p_u,\quad
Z(\sigma)=\sigma\p_u,
\\[1.5ex]
\bar D^t(\tau)=\tau\p_t+\tfrac13\tau_tx\p_x+\tfrac13\tau_ty\p_y-\tfrac19\tau_{tt}(\tfrac13x^3-xy^2)\p_u,\quad
\bar D^{\rm s}=x\p_x+y\p_y+3u\p_u+\tfrac32\vartheta\p_\vartheta,\\
\bar P^x(\chi)=\chi\p_x-\tfrac13\chi_t(x^2-y^2)\p_u,\quad
\bar P^y(\rho)=\rho\p_y+\tfrac23\rho_txy\p_u,\\
\bar R^x(\alpha)=\alpha x\p_u,\quad
\bar R^y(\beta)=\beta y\p_u,\quad
\bar Z(\sigma)=\sigma\p_u,\quad
\bar D^\psi=\psi\p_\psi,\quad
\bar P^\vartheta=\p_\vartheta.
\end{gather*}
Here we use the same notation for the Lie-symmetry vector fields of the specific Nyzhnyk models as that for their counterparts for the symmetric Nyzhnyk models, for which a different scaling of variables is employed.
The indicated actions modify only the form of the vector fields
$D^t(\tau)$, $P^x(\chi)$, $P^y(\rho)$ and $\bar D^t(\tau)$, $\bar P^x(\chi)$, $\bar P^y(\rho)$.

\subsection{Modified case}

In view of Remark~\ref{rem:RelationOfModModelsToLaxRepresentations}, the maximal Lie invariance pseudoalgebras of the modified Nyzhnyk models are, in fact, the maximal Lie invariance pseudoalgebras of the Lax representations of the associated unmodified Nyzhnyk models pulled back by the corresponding simple point transformations from Remark~\ref{rem:RelationOfModModelsToLaxRepresentations}.
To verify this, we also independently computed these pseudoalgebras using specialized computer algebra packages.
They have the following form:
\begin{gather*}
\mathfrak g_{\ref{eq:SymModNyzhnykSystem}}
=\langle D^t(\tau),\,P^x(\chi),\,P^y(\rho),\,R^x(\alpha),\,R^y(\beta),\,Z(\sigma),\,P^w\rangle,
\\
\mathfrak g_{\ref{eq:SymDispersionlessDefocModNyzhnykSystem}}
=\langle D^t(\tau),\,D^{\rm s},\,P^x(\chi),\,P^y(\rho),\,R^x(\alpha),\,R^y(\beta),\,Z(\sigma),\,P^w\rangle,
\\
\mathfrak g_{\ref{eq:AsymModNyzhnykSystem}}
=\langle D^t_{\rm a}(\tau),\,S(\gamma),\,P^x(\chi),\,R^x(\alpha),\,N(\varsigma),\,P^w\rangle,
\\
\mathfrak g_{\ref{eq:AsymDispersionlessDefocModNyzhnykSystem}}=
\langle D^t_{\rm a}(\tau),\,S(\gamma),\,D^{\rm s}_{\rm a},\,P^x(\chi),\,R^x(\alpha),\,N(\varsigma),\,P^w\rangle,
\\
\mathfrak g_{\ref{eq:SpecSymModNyzhnykSystem}'}=
\langle D^t_{\rm s}(\tau),\,P^x_{\rm s}(\chi),\,P^y_{\rm s}(\rho),\,R^x(\alpha),\,R^y(\beta),\,Z(\sigma),\,P^w\rangle,\\
\mathfrak g_{\ref{eq:dSpecSymModNyzhnykSystem}'}=
\langle D^t_{\rm s}(\tau),\,D^{\rm s},\,P^x_{\rm s}(\chi),\,P^y_{\rm s}(\rho),\,R^x(\alpha),\,R^y(\beta),\,Z(\sigma),\,P^w \rangle,
\end{gather*}
where all the generating vector fields
\begin{gather*}
D^t(\tau)=\tau\p_t+\tfrac13\tau_tx\p_x+\tfrac13\tau_ty\p_y+\tfrac1{54}\epsilon\tau_{tt}(x^3+y^3)\p_u,\quad
D^{\rm s}=x\p_x+y\p_y+3u\p_u+\tfrac32w\p_w,\\
D^t_{\rm a}(\tau)=\tau\p_t+\tfrac13\tau_tx\p_x+\tfrac1{54}\epsilon\tau_{tt}x^3\p_u,\quad
D^{\rm s}_{\rm a}=x\p_x+3u\p_u+\tfrac32w\p_w,\quad
S(\gamma)=\gamma\p_y,\\
D^t_{\rm s}(\tau)=\tau\p_t+\tfrac13\tau_tx\p_x+\tfrac13\tau_ty\p_y+\tfrac19\epsilon\tau_{tt}(\tfrac13x^3-xy^2)\p_u,\\
P^x(\chi)=\chi\p_x+\tfrac16\epsilon\chi_tx^2\p_u,\quad
P^y(\rho)=\rho\p_y+\tfrac16\epsilon\rho_ty^2\p_u,\\
P^x_{\rm s}(\chi)=\chi\p_x+\tfrac13\epsilon\chi_t(x^2-y^2)\p_u,\quad
P^y_{\rm s}(\rho)=\rho\p_y-\tfrac23\epsilon\rho_txy\p_u,\\
R^x(\alpha)=\alpha x\p_u,\quad
R^y(\beta)=\beta y\p_u,\quad
Z(\sigma)=\sigma\p_u,\quad
N(\varsigma)=\varsigma\p_u,\quad
P^w=\p_w
\end{gather*}
are defined in the space with coordinates $(t,x,y,u,w)$.
Note that, in fact, $D^{\rm s}=D^{\rm s}_{\rm a}+S(y)\in\mathfrak g_{\ref{eq:AsymDispersionlessDefocModNyzhnykSystem}}$, and $D^{\rm s}$ is replaced by $D^{\rm s}_{\rm a}$ in the set of vector fields generating the pseudoalgebra~$\mathfrak g_{\ref{eq:AsymDispersionlessDefocModNyzhnykSystem}}$ because $D^{\rm s}_{\rm a}$ has a simpler form than $D^{\rm s}$.
We also have $\{N(\varsigma)\}\supset\langle R^y(\beta),\,Z(\sigma)\rangle$.

\section{Point- and contact-symmetry pseudogroups}\label{sec:Pseudogroups}

Using the original megaideal-based version of the algebraic method proposed in~\cite{malt2024a} on the basis of the fundamental results from~\cite{hydo1998a,hydo1998b,hydo2000b,hydo2000A} and their development in~\cite{bihl2015a,bihl2011b},
the point-symmetry pseudogroups of the dispersionless Nyzhnyk equation~\eqref{eq:dN}, its nonlinear Lax representation~\eqref{eq:dNLaxPair}, the dispersionless Nyzhnyk system~\eqref{eq:dNSystem} and the contact-symmetry pseudogroup of the equation~\eqref{eq:dN} were computed in~\cite{boyk2024a}.
These results are presented in Section~\ref{sec:dSymCasePointSyPseudogroups} below,
where the point-symmetry pseudogroup of the nonlinear Lax representation~\eqref{eq:LaxdNSystem}
of the dispersionless Nyzhnyk system~\eqref{eq:dNSystem} is additionally found.
In Section~\ref{sec:DispersiveSymCasePointSyPseudogroups}, it is shown how all these results can be extended to the dispersive symmetric case by modifying the corresponding proofs.
Thus, the point-symmetry pseudogroups are computed for all the standard symmetric Nyzhnyk models from Section~\ref{sec:NyzhnykSystemsAndEquations}.
Analogous computations for the asymmetric case are essentially different.
Therefore, finding the point- and contact-symmetry pseudogroups of the asymmetric Nyzhnyk models from Section~\ref{sec:NyzhnykSystemsAndEquations} remains an open problem.

\begin{remark}\label{rem:OnRealAndComplexCasesForPointSyms}
In contrast to the Lie symmetries, to convert the description of point symmetries for the real case into its counterpart for the complex case,
it does not suffice to merely assume all the involved quantities to be complex
and to replace the smoothness of parameter functions with their analyticity.
Under the transition to the complex case, one additionally needs to reparameterize the representation of elements of the point-symmetry pseudogroups of the considered real models in order to avoid multi-valued parameter functions and other ambiguities associated with the cube root over the field of complex numbers.
Namely, for the reparameterization, we formally perform the substitutions $T_t^{1/3}=H$, $T_t=H^3$, $T_{tt}=3H_tH^2$, $T=\int H^3{\rm d}t$ with an arbitrary nonvanishing function~$H$ of~$t$, and $A^{1/3}=\breve A$, $A=\breve A^3$ with an arbitrary nonzero constant~$\breve A$.
\end{remark}

\begin{notation}
In this section, $\mathfrak r_i$ denotes the radical of the pseudoalgebra~$\mathfrak g_i$, where $i$~is the number of the corresponding model, and $\mathfrak z(\mathfrak s)$ and $\mathfrak s'$ are the center and the derived pseudoalgebra of a~pseudosubalgebra~$\mathfrak s$ of the pseudoalgebra~$\mathfrak g_i$, respectively.
All the parameter functions $T$, $H$, $X^0$, $Y^0$, $W^0$, $W^1$ and $W^2$ are arbitrary smooth functions of the variable~$t$ with $T_t\neq0$ and $H\neq0$, while $A$, $B$ and $C$ are arbitrary constants with $A\ne0$ and $C\ne0$.
\end{notation}

\begin{definition}
Discrete point symmetries of a model are called \emph{independent} if they are independent
up to composing with each other and with continuous point symmetries of this model.
\end{definition}

\subsection{Dispersionless symmetric case}\label{sec:dSymCasePointSyPseudogroups}

In view of the commutation relations of the vector fields generating the pseudoalgebra~$\mathfrak g_{\ref{eq:dN}}$ and the properties of megaideals from~\cite{bihl2015a,popo2003a}, one can easily construct the following megaideals of the pseudoalgebra~$\mathfrak g_{\ref{eq:dN}}$ \cite[Section~2]{boyk2024a}:
\begin{gather*}
\begin{split}&
\mathfrak m_1:=\mathfrak g_{\ref{eq:dN}}'
=\big\langle D^t(\tau),P^x(\chi),P^y(\rho),R^x(\alpha),R^y(\beta),Z(\sigma)\big\rangle,\\&
\mathfrak m_2:=\mathfrak r_{\ref{eq:dN}}
=\big\langle D^{\rm s},P^x(\chi),P^y(\rho),R^x(\alpha),R^y(\beta),Z(\sigma)\big\rangle,\\&
\mathfrak m_3:=\mathfrak m_2'=\mathfrak m_1\cap\mathfrak m_2
=\big\langle P^x(\chi),P^y(\rho),R^x(\alpha),R^y(\beta),Z(\sigma)\big\rangle,\\&
\mathfrak m_4:=\mathfrak m_2''=\big\langle R^x(\alpha),R^y(\beta),Z(\sigma)\big\rangle,\quad
\mathfrak m_5:=(\mathfrak m_3)^3=\mathfrak z(\mathfrak m_3)=\big\{ Z(\sigma)\big\},\\&
\mathfrak m_6:=\mathfrak z(\mathfrak m_1)=\big\langle Z(1)\big\rangle.
\end{split}
\end{gather*}
The only exception is the radical~$\mathfrak r_{\ref{eq:dN}}$ of the pseudoalgebra~$\mathfrak g_{\ref{eq:dN}}$, the finding of which required a~separate proof \cite[Lemma~1]{boyk2024a}.
All these megaideals are essential and have been used to construct the point-symmetry pseudogroup~$G_{\ref{eq:dN}}$ of the equation~\eqref{eq:dN}.

\begin{theorem}\label{thm:dNCompletePointSymGroup}
(i) The point-symmetry pseudogroup~$G_{\ref{eq:dN}}$ of the (real) dispersionless Nyzhnyk equation~\eqref{eq:dN} is generated by the transformations of the form
\begin{gather}\label{eq:dNPointSymForm}
\begin{split}
&\tilde t=T,\quad
\tilde x=CT_t^{1/3}x+X^0,\quad
\tilde y=CT_t^{1/3}y+Y^0,\\
&\tilde u=C^3u-\frac{C^3T_{tt}}{18T_t}(x^3+y^3)
-\frac{C^2}{2T_t^{1/3}}(X^0_tx^2+Y^0_ty^2)+W^1x+W^2y+W^0,
\end{split}
\end{gather}
and the transformation~$\mathscr J$: $\tilde t=t$, $\tilde x=y$, $\tilde y=x$, $\tilde u=u$.
%All parameter functions $T$, $X^0$, $Y^0$, $W^0$, $W^1$ and $W^2$ are arbitrary smooth functions of~$t$ with $T_t\neq0$, and $C$ is an arbitrary nonzero constant.

(ii) The contact-symmetry pseudogroup~$G^{\rm c_{\ref{eq:dN}}}$ of the dispersionless Nyzhnyk equation~\eqref{eq:dN} coincides with the first prolongation of the pseudogroup~$G_{\ref{eq:dN}}$.
\end{theorem}

The computation of the pseudogroup~$G^{\rm c}_{\ref{eq:dN}}$ in~\cite{boyk2024a} represents the first application of the original megaideal-based version of the algebraic method to finding the contact-symmetry pseudogroup of a differential equation.
According to the proof of Theorem~\ref{thm:dNCompletePointSymGroup} in~\cite{boyk2024a}, the necessary algebraic condition completely determines the point-symmetry pseudogroup~$G_{\ref{eq:dN}}$ of the equation~\eqref{eq:dN}.
This provided the first example of a system of differential equations with this property in the literature.
A second such example was found in~\cite{vinn2026a} and it is also, in a certain sense, related to the equation~\eqref{eq:dN}, since it constitutes a certain submodel of the equation~\eqref{eq:dN}.

\begin{corollary}\label{cor:dNDiscrSyms}
The identity component~$G^{\rm id}_{\ref{eq:dN}}$ in the pseudogroup~$G_{\ref{eq:dN}}$ consists of transformations of the form~\eqref{eq:dNPointSymForm} with $T_t>0$ and $C>0$.
Moreover, the group inducing the quotient pseudoalgebra~$G_{\ref{eq:dN}}/G^{\rm id}_{\ref{eq:dN}}$ is isomorphic to the group $\mathbb Z_2\times\mathbb Z_2\times\mathbb Z_2$. The complete list of independent discrete point symmetries of the equation~\eqref{eq:dN} is exhausted by three commuting involutions, which can be chosen as
\begin{gather*}
\mathscr J\colon(\tilde t,\tilde x,\tilde y,\tilde u)=(t,y,x,u),\\
\mathscr I^{\rm i}\colon(\tilde t,\tilde x,\tilde y,\tilde u)=(-t,-x,-y,u),\quad
\mathscr I^{\rm s}\colon(\tilde t,\tilde x,\tilde y,\tilde u)=(t,-x,-y,-u).
\end{gather*}
\end{corollary}

\begin{remark}\label{cor:dNComplex}
To obtain from Theorem~\ref{thm:dNCompletePointSymGroup} a correct description of the point-symmetry pseudogroup $G^{\mathbb C}_{\ref{eq:dN}}$ of the complex dispersionless Nyzhnyk equation~\eqref{eq:dN}$^{\mathbb C}$, in addition to the complexification, it is necessary, according to Remark~\ref{rem:OnRealAndComplexCasesForPointSyms}, to reparameterize the family of transformations~\eqref{eq:dNPointSymForm}, which takes the form
\begin{gather*}\label{eq:dNComplex}
\begin{split}
&\tilde t=\int H^3{\rm d}t,\quad
\tilde x=CHx+X^0,\quad
\tilde y=CHy+Y^0,\\
&\tilde u=C^3u-\frac{C^3H_t}{6H}(x^3+y^3)
-\frac{C^2}{2H}(X^0_tx^2+Y^0_ty^2)+W^1x+W^2y+W^0.
\end{split}
\end{gather*}
The reparameterized transformations constitute the identity component in the pseudogroup $G^{\mathbb C}_{\ref{eq:dN}}$.
Moreover, the group that induces the quotient pseudoalgebra of the pseudogroup~$G^{\mathbb C}_{\ref{eq:dN}}$ with respect to this component is isomorphic to the group~$\mathbb Z_2$.
Therefore, only the transformation~$\mathscr J$ remains an independent discrete point symmetry under the complexification of the equation~\eqref{eq:dN}.
\end{remark}

In~\cite[Section~6]{boyk2024a}, the following megaideals of the pseudoalgebra~$\mathfrak g_{\ref{eq:dNLaxPair}}$ were found:
\begin{gather*}
\mathfrak g_{\ref{eq:dNLaxPair}}'
=\big\langle\bar D^t(\tau),\bar P^x(\chi),\bar P^y(\rho),\bar R^x(\alpha),\bar R^y(\beta),\bar Z(\sigma),\bar P^\vartheta\big\rangle,\quad
\\
\bar{\mathfrak m}_1:=\mathfrak g_{\ref{eq:dNLaxPair}}''
=\big\langle\bar D^t(\tau),\bar P^x(\chi),\bar P^y(\rho),\bar R^x(\alpha),\bar R^y(\beta),\bar Z(\sigma)\big\rangle,\quad
\\
\bar{\mathfrak m}_2:=\mathfrak r_{\ref{eq:dNLaxPair}}=\big\langle\bar D^{\rm s},\bar P^x(\chi),\bar P^y(\rho),
\bar R^x(\alpha),\bar R^y(\beta),\bar Z(\sigma),\bar P^\vartheta\big\rangle,
\\
\bar{\mathfrak m}_2'=\mathfrak g_{\ref{eq:dNLaxPair}}'\cap\bar{\mathfrak m}_2
=\big\langle\bar P^x(\chi),\bar P^y(\rho),\bar R^x(\alpha),\bar R^y(\beta),\bar Z(\sigma),\bar P^\vartheta\big\rangle,
\\
\bar{\mathfrak m}_3:=\bar{\mathfrak m}_1\cap\bar{\mathfrak m}_2'
=\big\langle\bar P^x(\chi),\bar P^y(\rho),\bar R^x(\alpha),\bar R^y(\beta),\bar Z(\sigma)\big\rangle,
\\
\bar{\mathfrak m}_4:=\bar{\mathfrak m}_2''=\big\langle\bar R^x(\alpha),\bar R^y(\beta),\bar Z(\sigma)\big\rangle,\quad
\bar{\mathfrak m}_5:=\bar{\mathfrak m}_2'''=\big\{\bar Z(\sigma)\big\},
\\
\mathfrak z(\mathfrak g_{\ref{eq:dNLaxPair}}')=\big\langle Z(1),\bar P^\vartheta\big\rangle,\quad
\bar{\mathfrak m}_6:=\mathfrak z(\mathfrak g_{\ref{eq:dNLaxPair}}'')=\big\langle Z(1)\big\rangle,\quad
\bar{\mathfrak m}_7:=\big\langle \bar P^\vartheta\big\rangle.
\end{gather*}
In general, the technique for finding the megaideal~$\bar{\mathfrak m}_7$ differs from the technique used for the other obtained megaideals, see \cite[Lemma~10]{boyk2024a}.
The megaideals~$\bar{\mathfrak m}_7$, $\bar{\mathfrak m}_6$, $\bar{\mathfrak m}_3$ and $\bar{\mathfrak m}_1$ are essential for constructing the point-symmetry pseudogroup~$G_{\ref{eq:dNLaxPair}}$ of the nonlinear Lax representation~\eqref{eq:dNLaxPair} of the equation~\eqref{eq:dN} by means of the original megaideal-based version of the algebraic method.

\begin{theorem}\label{thm:dNCompletePointSymGroupOfLaxRepresentation}
The point-symmetry pseudogroup~$G_{\ref{eq:dNLaxPair}}$ of the (real) nonlinear Lax representation~\eqref{eq:dNLaxPair} of the equation~\eqref{eq:dN} is generated by the transformations of the form
\begin{gather}\label{eq:LaxEqPointSymForm}
\begin{split}
&\tilde t=T,\quad
 \tilde x=A^{2/3}T_t^{1/3}x+X^0,\quad
 \tilde y=A^{2/3}T_t^{1/3}y+Y^0,\\
&\tilde u=A^2u-\frac{A^2T_{tt}}{18T_t}(x^3+y^3)
 -\frac{A^{4/3}}{2T_t^{1/3}}(X^0_tx^2+Y^0_ty^2)+W^1x+W^2y+W^0,\\
&\tilde\vartheta=A\vartheta+B
\end{split}
\end{gather}
and the transformation~$\bar{\mathscr J}$: $\tilde t=t$, $\tilde x=y$, $\tilde y=x$, $\tilde u=u$, $\tilde\vartheta=\vartheta$.
%The parameter functions $T$, $X^0$, $Y^0$, $W^0$, $W^1$ and $W^2$ are arbitrary smooth functions of the variable~$t$ with $T_t\neq0$, and $A$ and~$B$ are an arbitrary nonzero constants with $A\ne0$.
\end{theorem}

In contrast to the pseudogroup~$G_{\ref{eq:dN}}$, the necessary algebraic condition does not determine the pseudogroup~$G_{\ref{eq:dNLaxPair}}$ completely.
Moreover, to obtain it, one also needs to employ the direct method, see the proof of Theorem~11 in~\cite{boyk2024a}.

\begin{corollary}\label{cor:dNLaxDiscrSyms}
A complete list of independent discrete point symmetries of the nonlinear Lax representation~\eqref{eq:dNLaxPair} is exhausted by three commuting involutions, which can be chosen as
\begin{gather*}
\bar{\mathscr J}\colon(\tilde t,\tilde x,\tilde y,\tilde u,\tilde\vartheta)=(t,y,x,u,\vartheta),\\
\bar{\mathscr I}^{\rm i}\colon(\tilde t,\tilde x,\tilde y,\tilde u,\tilde\vartheta)=(-t,-x,-y,u,\vartheta),\quad
\bar{\mathscr I}^\vartheta\colon(\tilde t,\tilde x,\tilde y,\tilde u,\tilde\vartheta)=(t,x,y,u,-\vartheta).
\end{gather*}
\end{corollary}

The discrete point symmetry~$\mathscr I^{\rm s}$ of the equation~\eqref{eq:dN} has no counterparts among the discrete point symmetries of the nonlinear Lax representation~\eqref{eq:dNLaxPair} of this equation.
The trivial prolongation of the transformation~$\mathscr I^{\rm s}$ to~$\vartheta$ maps~\eqref{eq:dNLaxPair} into another, equivalent, nonlinear Lax representation for the equation~\eqref{eq:dN} \cite[Remark~13]{boyk2024a}.
At the same time, a new discrete point symmetry arises in~\eqref{eq:dNLaxPair}, which is associated with the pseudopotential~$\vartheta$.

\begin{remark}\label{cor:dNLaxComplex}
For a correct description of the point-symmetry pseudogroup~$G^{\mathbb C}_{\ref{eq:dNLaxPair}}$ of the complex nonlinear Lax representation~\eqref{eq:dNLaxPair}$^{\mathbb C}$ in Theorem~\ref{thm:dNCompletePointSymGroupOfLaxRepresentation}, according to Remark~\ref{rem:OnRealAndComplexCasesForPointSyms}, we complexify all the involved quantities and reparameterize the transformations~\eqref{eq:LaxEqPointSymForm} as follows:
%$T_t^{1/3}=H(t)$, $T_t=H(t)^3$, $T_{tt}=3H_tH^2$
\begin{gather*}\label{eq:dNLaxComplex}
\begin{split}
&\tilde t=\int H^3{\rm d}t,\quad
 \tilde x=\breve A^2Hx+X^0,\quad
 \tilde y=\breve A^2Hy+Y^0,\\
&\tilde u=\breve A^6u-\frac{\breve A^6H_t}{6H}(x^3+y^3)
 -\frac{\breve A^4}{2H}(X^0_tx^2+Y^0_ty^2)+W^1x+W^2y+W^0,\quad
\tilde\vartheta=\breve A^3\vartheta+B.
\end{split}
\end{gather*}
Analogously to the equation~\eqref{eq:dN}, the reparameterized transformations form the identity component in the pseudogroup~$G^{\mathbb C}_{\ref{eq:dNLaxPair}}$.
Moreover, the group that induces the quotient pseudoalgebra of the pseudogroup~$G^{\mathbb C}_{\ref{eq:dNLaxPair}}$ with respect to this component is isomorphic to the group~$\mathbb Z_2$.
Among all the independent discrete point symmetries of the system~\eqref{eq:dNLaxPair}, only the transformation~$\bar{\mathscr J}$ is preserved under its complexification.
\end{remark}

The megaideals
\begin{gather*}
\hat{\mathfrak m}_1:=\mathfrak g_{\ref{eq:dNSystem}}'=\big\langle\hat D^t(\tau),\hat P^x(\chi),\hat P^y(\rho)\big\rangle,\\
\hat{\mathfrak m}_2:=\mathfrak r_{\ref{eq:dNSystem}}=\big\langle\hat D^{\rm s},\hat P^x(\chi),\hat P^y(\rho)\big\rangle,\\
\hat{\mathfrak m}_3:=\hat{\mathfrak m}_2'=\hat{\mathfrak m}_1\cap\hat{\mathfrak m}_2=
\big\langle\hat P^x(\chi),\hat P^y(\rho)\big\rangle
\end{gather*}
of the pseudoalgebra~$\mathfrak g_{\ref{eq:dNSystem}}$ were found and used in~\cite[Section~7]{boyk2024a} to construct the point-symmetry pseudogroup~$G_{\ref{eq:dNSystem}}$ of the dispersionless Nyzhnyk system~\eqref{eq:dNSystem}.

\begin{theorem}\label{thm:dNSystemPointSymPseudogroup}
The point-symmetry pseudogroup~$G_{\ref{eq:dNSystem}}$ of the (real) dispersionless Nyzhnyk system~\eqref{eq:dNSystem} is generated by the transformations of the form
\begin{gather}\label{eq:dNSystemPointSymForm}
\begin{split}&
\tilde t=T,\quad
\tilde x=CT_t^{1/3}x+X^0,\quad
\tilde y=CT_t^{1/3}y+Y^0,\\&
\tilde w=\frac C{T_t^{2/3}}w,\quad
\tilde v^1=\frac C{T_t^{2/3}}v^1-\frac{CT_{tt}}{3T_t^{5/3}}x-\frac{X^0_t}{T_t},\quad
\tilde v^2=\frac C{T_t^{2/3}}v^2-\frac{CT_{tt}}{3T_t^{5/3}}y-\frac{Y^0_t}{T_t}
\end{split}
\end{gather}
and the transformation~$\hat{\mathscr J}$: $\tilde t=t$, $\tilde x=y$, $\tilde y=x$, $\tilde w=w$, $\tilde v^1=v^2$, $\tilde v^2=v^1$.
%The parameter functions $T$, $X^0$ and $Y^0$ are arbitrary smooth functions of the variable~$t$ with $T_t\neq0$, and $C$ is an arbitrary nonzero constant.
\end{theorem}
As with the nonlinear Lax representation~\eqref{eq:dNLaxPair}, the proof for constructing the pseudogroup~$G_{\ref{eq:dNSystem}}$ should necessarily be completed by applying the direct method.
Note that the vector field $Q^6:=\hat D^{\rm s}$ in the proof of Theorem~\ref{thm:dNSystemPointSymPseudogroup} in~\cite[Section~7]{boyk2024a} can be replaced by $\hat D^t(t^2)$; cf.\ Theorem~\ref{thm:dNSystemDefSubalgs} below.

\begin{corollary}\label{cor:dNSystemDiscrSyms}
A complete list of independent discrete point symmetries of the system~\eqref{eq:dNSystem} is exhausted by three commuting involutions, which can be chosen as
\begin{gather*}
\hat{\mathscr J}\colon(\tilde t,\tilde x,\tilde y,\tilde w,\tilde v^1,\tilde v^2)=(t,y,x,w,v^2,v^1),\\
\hat{\mathscr I}^{\rm i}\colon(\tilde t,\tilde x,\tilde y,\tilde w,\tilde v^1,\tilde v^2)=(-t,-x,-y,w,v^1,v^2),\\
\hat{\mathscr I}^{\rm s}\colon(\tilde t,\tilde x,\tilde y,\tilde w,\tilde v^1,\tilde v^2)=(t,-x,-y,-w,-v^1,-v^2).
\end{gather*}
\end{corollary}

\begin{remark}\label{cor:dNSystemComplex}
According to Remark~\ref{rem:OnRealAndComplexCasesForPointSyms}, the point-symmetry pseudogroup~$G^{\mathbb C}_{\ref{eq:dNSystem}}$ of the complex dispersionless Nyzhnyk system~\eqref{eq:dNSystem}$^{\mathbb C}$ can be obtained from Theorem~\ref{thm:dNSystemPointSymPseudogroup} by performing the complexification and reparameterizing the transformations~\eqref{eq:dNSystemPointSymForm}, which consequently take the form
%$T_t^{1/3}=H(t)$, $T_t=H(t)^3$, $T_{tt}=3H_tH^2$
\begin{gather*}\label{dNSystemComplex}
\begin{split}
&\tilde t=\int H^3{\rm d}t,\quad
 \tilde x=CHx+X^0,\quad
\tilde y=CHy+Y^0,\\&
\tilde w=\frac C{H^2}w,\quad
\tilde v^1=\frac C{H^2}v^1-\frac{CH_t}{H^3}x-\frac{X^0_t}{H^3},\quad
\tilde v^2=\frac C{H^2}v^2-\frac{CH_t}{H^3}y-\frac{Y^0_t}{H^3}.
\end{split}
\end{gather*}
Upon complexification of the system~\eqref{eq:dNSystem}, only the transformation~$\hat{\mathscr J}$ remains as its independent discrete point symmetry.
\end{remark}

Although the point-symmetry pseudogroup of the nonlinear Lax representation~\eqref{eq:LaxdNSystem} of the dispersionless Nyzhnyk system~\eqref{eq:dNSystem} was not considered in~\cite{boyk2024a}, it can be constructed by using elements of the proofs of Theorems~\ref{thm:dNCompletePointSymGroupOfLaxRepresentation} and~\ref{thm:dNSystemPointSymPseudogroup}, which were obtained as Theorems~11 and~14 in~\cite{boyk2024a}.
The megaideals of the pseudoalgebra~$\mathfrak g_{\ref{eq:LaxdNSystem}}$ have the form
\begin{gather*}
\mathfrak g_{\ref{eq:LaxdNSystem}}'=\big\langle\check D^t(\tau),\check P^x(\chi),\check P^y(\rho),\check P^\vartheta\big\rangle,\\
\check{\mathfrak m}_1:=\mathfrak g_{\ref{eq:LaxdNSystem}}''=\big\langle\check D^t(\tau),\check P^x(\chi),\check P^y(\rho)\big\rangle,\\
\check{\mathfrak m}_2:=\mathfrak r_{\ref{eq:LaxdNSystem}}=\big\langle\check D^{\rm s},\check P^x(\chi),\check P^y(\rho),\check P^\vartheta\big\rangle,\quad
\check{\mathfrak m}_2'=\mathfrak g_{\ref{eq:LaxdNSystem}}'\cap\check{\mathfrak m}_2=
\big\langle\check P^x(\chi),\check P^y(\rho),\check P^\vartheta\big\rangle, \\
\check{\mathfrak m}_3:=\check{\mathfrak m}_1\cap\check{\mathfrak m}_2'=
\big\langle\check P^x(\chi),\check P^y(\rho)\big\rangle, \quad
\check{\mathfrak m}_4:=\mathfrak z(\mathfrak g_{\ref{eq:LaxdNSystem}}')=\big\langle\check P^\vartheta\big\rangle.
\end{gather*}
The technique for finding the megaideal~$\check{\mathfrak m}_4$ is similar to the technique used for finding the megaideal~$\bar{\mathfrak m}_7$ of the pseudoalgebra~$\mathfrak g_{\ref{eq:dNLaxPair}}$, see \cite[Lemma~10]{boyk2024a}.

\begin{theorem}\label{thm:LaxdNSystemPointSymPseudogroup}
The point-symmetry pseudogroup~$G_{\ref{eq:LaxdNSystem}}$ of the (real) nonlinear Lax representation~\eqref{eq:LaxdNSystem} of the dispersionless Nyzhnyk system~\eqref{eq:dNSystem} is generated by the transformations of the form
\begin{gather}\label{eq:LaxdNSystemPointSymForm}
\begin{split}&
\tilde t=T,\quad
\tilde x=A^{2/3}T_t^{1/3}x+X^0,\quad
\tilde y=A^{2/3}T_t^{1/3}y+Y^0,\\&
\tilde w=\frac{A^{2/3}}{T_t^{2/3}}w,\quad
\tilde v^1=\frac{A^{2/3}}{T_t^{2/3}}v^1-\frac{A^{2/3}T_{tt}}{3T_t^{5/3}}x-\frac{X^0_t}{T_t},\quad
\tilde v^2=\frac{A^{2/3}}{T_t^{2/3}}v^2-\frac{A^{2/3}T_{tt}}{3T_t^{5/3}}y-\frac{Y^0_t}{T_t},\\&
\tilde\vartheta=A\vartheta+B
\end{split}
\end{gather}
and the transformation~$\check{\mathscr J}$: $\tilde t=t$, $\tilde x=y$, $\tilde y=x$, $\tilde w=w$, $\tilde v^1=v^2$, $\tilde v^2=v^1$, $\tilde\vartheta=\vartheta$.
%The parameter functions $T$, $X^0$ and $Y^0$ are arbitrary smooth functions of the variable~$t$ with $T_t\neq0$, and $A$ and $B$ are an arbitrary nonzero constants with $A\ne0$.
\end{theorem}

\begin{proof}
We apply the original megaideal-based version of the algebraic method from~\cite{malt2024a} analogously to the proofs of Theorems~11 and~14 from~\cite{boyk2024a}.
Let a point transformation
\[
\check\Phi\colon\ (\tilde t,\tilde x,\tilde y,\tilde w,\tilde v^1,\tilde v^2,\tilde\vartheta)=(T,X,Y,W,V^1,V^2,\Theta)
\]
in the extended space with coordinates $(t,x,y,w,v^1,v^2,\vartheta)$ be a point symmetry of the system~\eqref{eq:LaxdNSystem}.
Here, $(T,X,Y,W,V^1,V^2,\Theta)$ is a tuple of smooth functions of $(t,x,y,w,v^1,v^2,\vartheta)$ with a nonvanishing Jacobian.
From the condition $\check\Phi_*\check P^\vartheta\subseteq\check{\mathfrak m}_4$, it follows that $T_\vartheta=X_\vartheta=Y_\vartheta=W_\vartheta=V^1_\vartheta=V^2_\vartheta=0$ and $\Theta_\vartheta=\const$.
Since $\check\varpi_*\mathfrak g_{\ref{eq:LaxdNSystem}}=\mathfrak g_{\ref{eq:dNSystem}}$, where $\check\varpi$ is the natural projection from \smash{$\mathbb R^3_{t,x,y}\times\mathbb R^4_{\smash{w,v^1,v^2,\vartheta}}$} onto \smash{$\mathbb R^3_{t,x,y}\times\mathbb R^3_{\smash{w,v^1,v^2}}$}, the independence of the $t$-, $x$-, $y$-, $w$-, $v^1$- and $v^2$-components of the transformation~$\check\Phi$ from $\vartheta$ implies that they also satisfy the constraints from the proof of Theorem~14 in~\cite{boyk2024a} for the components of the transformation~$\Phi$.
It is obvious that the transformation~$\check{\mathscr J}$ is a point symmetry of the system~\eqref{eq:LaxdNSystem}.
Collecting the $\vartheta$-components in the extended conditions $\check\Phi_*\check P^z(\zeta)\in\check{\mathfrak m}_3$, $\check\Phi_*\check D^t(\zeta)\in\check{\mathfrak m}_1\setminus\check{\mathfrak m}_3$, $z\in\{x,y\}$, $\zeta\in\{1,t\}$, we obtain the equations $\Theta_t=\Theta_x=\Theta_y=\Theta_w=\Theta_{v^1}=\Theta_{v^2}=0$.
Thus, $\Theta=A\vartheta+B$, where $A$ and $B$ are arbitrary constants with $A\ne0$.

Every point transformation~$\Phi$ whose components have the form specified above satisfies the condition $\check\Phi_*\mathfrak g_{\ref{eq:LaxdNSystem}}\subseteq\mathfrak g_{\ref{eq:LaxdNSystem}}$.
This means that this form cannot be restricted any further solely within the framework of the algebraic method.
Therefore, analogously to finding the pseudogroups~$G_{\ref{eq:dNLaxPair}}$ and~$G_{\ref{eq:dNSystem}}$, the direct method should be used to complete the construction of the pseudogroup~$G_{\ref{eq:LaxdNSystem}}$.

We derive expressions for the derivatives of $(\tilde w,\tilde v^1,\tilde v^2,\tilde\vartheta)$ with respect to $(\tilde t,\tilde x,\tilde y)$ up to the second order in terms of the untilded variables and derivatives, successively substitute the resulting expressions and the expressions for the derivatives~$\vartheta_t$ and~$\vartheta_y$ by virtue of the system~\eqref{eq:LaxdNSystem} into the same system written in terms of the tilded variables and split the resulting equations with respect to the other (parametric) derivatives of the dependent variables up to order two.
The resulting system of equations for the parameters of the point-symmetry transformations of the system~\eqref{eq:LaxdNSystem} reduces to a single equation $C^3=A^2$, which implies $C=A^{2/3}>0$.
\end{proof}

\begin{corollary}\label{cor:LaxdNSystemDiscrSyms}
A complete list of independent discrete point-symmetry transformations of the system~\eqref{eq:LaxdNSystem} is exhausted by three commuting involutions, which can be chosen as
\begin{gather*}
\check{\mathscr J}\colon(\tilde t,\tilde x,\tilde y,\tilde w,\tilde v^1,\tilde v^2,\vartheta)=(t,y,x,w,v^2,v^1,\vartheta),\\
\check{\mathscr I}^{\rm i}\colon(\tilde t,\tilde x,\tilde y,\tilde w,\tilde v^1,\tilde v^2,\vartheta)=(-t,-x,-y,w,v^1,v^2,\vartheta),\\
\check{\mathscr I}^\vartheta\colon(\tilde t,\tilde x,\tilde y,\tilde w,\tilde v^1,\tilde v^2,\vartheta)=(t,x,y,w,v^1,v^2,-\vartheta).
\end{gather*}
\end{corollary}

\begin{remark}\label{cor:LaxdNSystemComplex}
To obtain a correct description of the point-symmetry pseudogroup~$G^{\mathbb C}_{\ref{eq:LaxdNSystem}}$ of the complex nonlinear Lax representation~\eqref{eq:LaxdNSystem}$^{\mathbb C}$, along with its complexification, one needs to reparameterize the transformations~\eqref{eq:LaxdNSystemPointSymForm} according to Remark~\ref{rem:OnRealAndComplexCasesForPointSyms}:
%$T_t^{1/3}=H(t)$, $T_t=H(t)^3$, $T_{tt}=3H_tH^2$
\begin{gather*}\label{LaxdNSystemComplex}
\begin{split}&
\tilde t=\int H^3{\rm d}t,\quad
\tilde x=\breve A^2Hx+X^0,\quad
\tilde y=\breve A^2Hy+Y^0,\\&
\tilde w=\frac{\breve A^2}{H^2}w,\quad
\tilde v^1=\frac{\breve A^2}{H^2}v^1-\frac{\breve A^2H_t}{H^3}x-\frac{X^0_t}{H^3},\quad
\tilde v^2=\frac{\breve A^2}{H^2}v^2-\frac{\breve A^2H_t}{H^3}y-\frac{Y^0_t}{H^3},\\[1ex]&
\tilde\vartheta=A\vartheta+B.
\end{split}
\end{gather*}
Analogously to the system~\eqref{eq:dNSystem}, the complexification of the system~\eqref{eq:LaxdNSystem} preserves only the transformation~$\check{\mathscr J}$ as its independent discrete point symmetry.
\end{remark}

\subsection{Dispersive symmetric case}\label{sec:DispersiveSymCasePointSyPseudogroups}

Since the pseudoalgebra~$\mathfrak g_{\ref{eq:SymPotNizhEq}'}$ is a codimension-one pseudosubalgebra of the pseudoalgebra~$\mathfrak g_{\ref{eq:dN}}$ spanned by the same vector fields as~$\mathfrak g_{\ref{eq:dN}}$ except for~$D^{\rm s}$, one can assume the existence of a similar connection between the pseudogroups~$G_{\ref{eq:SymPotNizhEq}'}$ and~$G_{\ref{eq:dN}}$, as well as between their contact counterparts. Such a connection indeed exists.
Namely, the following statement holds.

\begin{theorem}\label{thm:SymPotNizhEqCompletePointSymGroup}
(i) The point-symmetry pseudogroup~$G_{\ref{eq:SymPotNizhEq}'}$ of the real Nyzhnyk equation~{\rm(\ref{eq:SymPotNizhEq}$'$)} is a pseudosubgroup of the pseudogroup~$G_{\ref{eq:dN}}$ singled out by the condition $C=1$.
That is, it is generated by the transformations of the form~\eqref{eq:dNPointSymForm} with $C=1$ along with the transformation~$\mathscr J$: $\tilde t=t$, $\tilde x=y$, $\tilde y=x$, $\tilde u=u$.
\noprint{
\begin{gather*}\label{eq:SymPotNizhEqPointSymForm}
\begin{split}
&\tilde t=T(t),\quad
\tilde x=T_t^{1/3}x+X^0(t),\quad
\tilde y=T_t^{1/3}y+Y^0(t),\\
&\tilde u=u-\frac{T_{tt}}{18T_t}(x^3+y^3)
-\frac{1}{2T_t^{1/3}}(X^0_tx^2+Y^0_ty^2)+W^1(t)x+W^2(t)y+W^0(t),
\end{split}
\end{gather*}
}
%All parameter functions $T$, $X^0$, $Y^0$, $W^0$, $W^1$ and $W^2$ are arbitrary smooth functions of the variable~$t$ with $T_t\neq0$.

(ii) The contact-symmetry pseudogroup~$G^{\rm c}_{\ref{eq:SymPotNizhEq}'}$ of the real Nyzhnyk equation~{\rm(\ref{eq:SymPotNizhEq}$'$)} is a pseudosubgroup of the pseudogroup~$G^{\rm c}_{\ref{eq:dN}}$ singled out by the condition $C=1$ and it coincides with the first prolongation of the pseudogroup~$G_{\ref{eq:SymPotNizhEq}'}$.
\end{theorem}

\begin{proof}
The algebraic part of the proof is an adaptation for the pseudoalgebra~$\mathfrak g_{\ref{eq:SymPotNizhEq}'}$ as a pseudosubalgebra of the pseudoalgebra~$\mathfrak g_{\ref{eq:dN}}$ of the second version of the proof of Theorem~2 from~\cite{boyk2024a}, which is Theorem~\ref{thm:dNCompletePointSymGroup} in the present paper.
We apply the same original megaideal-based version of the algebraic method from~\cite{malt2024a}.

(i) The megaideals~$\mathfrak m_1$, $\mathfrak m_3$, \dots, $\mathfrak m_6$ of the pseudoalgebra~$\mathfrak g_{\ref{eq:dN}}$ are also megaideals of the pseudoalgebra~$\mathfrak g_{\ref{eq:SymPotNizhEq}'}$.
In terms of the pseudoalgebra~$\mathfrak g_{\ref{eq:SymPotNizhEq}'}$, these megaideals can be represented as
$\mathfrak m_1=\mathfrak g_{\ref{eq:SymPotNizhEq}'}$,
$\mathfrak m_3=\mathfrak r_{\ref{eq:SymPotNizhEq}'}$,
$\mathfrak m_4=\mathfrak m_3'$,
$\mathfrak m_5=(\mathfrak m_3)^3=\mathfrak z(\mathfrak m_3)$ and
$\mathfrak m_6=\mathfrak z(\mathfrak m_1)$.
The proof that the radical~$\mathfrak r_{\ref{eq:SymPotNizhEq}'}$ of the pseudoalgebra~$\mathfrak g_{\ref{eq:SymPotNizhEq}'}$ coincides with~$\mathfrak m_3$ is analogous to the proofs of Lemma~1 in~\cite{malt2024a} and Lemma~1 in~\cite{boyk2024a}.
Note also that, unlike the case of~$\mathfrak g_{\ref{eq:dN}}$, the megaideal~$\mathfrak m_1$ is an improper megaideal of the pseudoalgebra~$\mathfrak g_{\ref{eq:SymPotNizhEq}'}$ and this entire pseudoalgebra is not the sum of its proper megaideals, and therefore it is essential as a megaideal of itself.

The application of the above method is based on the necessary condition regarding the adjoint actions of point symmetries of a system of differential equations on the corresponding maximal Lie invariance (pseudo)algebra. More specifically, if a point transformation~$\Phi$ in the space with coordinates $(t,x,y,u)$,
\[
\Phi\colon\ (\tilde t,\tilde x,\tilde y,\tilde u)=(T,X,Y,U),
\]
where $(T,X,Y,U)$ is a tuple of smooth functions of $(t,x,y,u)$ with a nonvanishing Jacobian, is a~point symmetry of the equation~(\ref{eq:SymPotNizhEq}$'$), then $\Phi_*\mathfrak m_j\subseteq\mathfrak m_j$ for $j=1,3,4,5,6$.
Here and below, $\Phi_*$~denotes the pushforward of vector fields by the transformation~$\Phi$ and $z\in\{x,y\}$.
We evaluate this necessary condition on the following linearly independent elements of the pseudoalgebra~$\mathfrak g_{\ref{eq:dN}}$:
\begin{gather*}%\label{VectorFields}
\begin{split}&
Q^1:=Z(1),\quad Q^2:=Z(t),\quad Q^{3z}:=R^z(1),\\&
Q^{4z}:=P^z(1),\quad Q^{5z}:=P^z(t),\quad
Q^7:=D^t(1),\quad Q^8:=D^t(t),
\end{split}
\end{gather*}
that is, compared to the proof of Theorem~2 from~\cite{boyk2024a}, we exclude from the selection the Lie-symmetry vector field $Q^6:=D^{\rm s}$ belonging to $\mathfrak g_{\ref{eq:dN}}\setminus\mathfrak g_{\ref{eq:SymPotNizhEq}'}$.
Since
$Q^1\in\mathfrak m_6$,
$Q^2\in\mathfrak m_5$,
$Q^{3z}\in\mathfrak m_4$,
$Q^{4z},Q^{5z}\in\mathfrak m_3$ and
$Q^7,Q^8\in\mathfrak m_1$, we have
\begin{gather}\label{eq:NMainPushforwards}
\begin{split}&
\Phi_*Q^i=\tilde Z(\tilde\sigma^i),\quad i=1,2,
\\&
\Phi_*Q^{iz}=\tilde R^x(\tilde\alpha^{iz})+\tilde R^y(\tilde\beta^{iz})+\tilde Z(\tilde\sigma^{iz}),\quad i=3,
\\&
\Phi_*Q^{iz}=\tilde P^x(\tilde\chi^{iz})+\tilde P^y(\tilde\rho^{iz})
+\tilde R^x(\tilde\alpha^{iz})+\tilde R^y(\tilde\beta^{iz})+\tilde Z(\tilde\sigma^{iz}),\quad i=4,5,
\\&
\Phi_*Q^i=\tilde D^t(\tilde\tau^i)+\tilde P^x(\tilde\chi^i)+\tilde P^y(\tilde\rho^i)
+\tilde R^x(\tilde\alpha^i)+\tilde R^y(\tilde\beta^i)+\tilde Z(\tilde\sigma^i),\quad i=7,8.
\end{split}
\end{gather}
Here, $\tilde\sigma^1$ is a constant, while the other parameters are smooth functions of~$\tilde t$, with \mbox{$\tilde\sigma^1\tilde\sigma^2\ne0$}.
The tilde over a vector field denotes that it is represented in tilded variables.
Expanding the equations~\eqref{eq:NMainPushforwards} componentwise as equalities of vector fields, we derive a system of determining equations for the components of the transformation~$\Phi$, which we solve in the same manner as in~\cite{boyk2024a}.
As a~result, we find that, up to composition with the transformation~$\mathscr J$: $\tilde t=t$, $\tilde x=y$, $\tilde y=x$, $\tilde u=u$, the transformation~$\Phi$ has the form~\eqref{eq:dNPointSymForm}.

An essential difference from the proof of Theorem~2 in~\cite{boyk2024a} is that, unlike for the pseudogroup~$G_{\ref{eq:dN}}$, the utilized algebraic condition does not completely determine the pseudogroup~$G_{\ref{eq:SymPotNizhEq}'}$.
Therefore, to complete the computation here, one should employ the direct method.
It is based on the fact that, by the definition of point symmetries, the transformation~$\Phi$ maps the equation~{\rm(\ref{eq:SymPotNizhEq}$'$)} to an equation of the same form in tilded variables.
Thus, expressing the tilded derivatives in the transformed equation in terms of the initial variables using the chain rule, we obtain an equation satisfied by any solution of the equation~{\rm(\ref{eq:SymPotNizhEq}$'$)}.
In other words, substituting the expression for the derivative~$u_{txy}$ by virtue of the equation~{\rm(\ref{eq:SymPotNizhEq}$'$)} into the obtained equation results in a relation whose splitting with respect to the remaining parametric derivatives of the dependent variables yields the condition $C=1$.

(ii)
It suffices to prove that the pseudogroup~$G^{\rm c}_{\ref{eq:SymPotNizhEq}'}$ coincides with the first prolongation of the pseudogroup~$G_{\ref{eq:SymPotNizhEq}'}$, for which we again follow the proof of Theorem~2 from~\cite{boyk2024a}.
The first prolongations~$\mathfrak m_{1(1)}$, $\mathfrak m_{3(1)}$, \dots, $\mathfrak m_{6(1)}$ of the megaideals~$\mathfrak m_1$, $\mathfrak m_3$, \dots, $\mathfrak m_6$ of the pseudoalgebra~$\mathfrak g_{\ref{eq:SymPotNizhEq}'}$ are megaideals of its first prolongation~$\mathfrak g_{\ref{eq:SymPotNizhEq}'(1)}$. Consider a contact transformation
\begin{gather}%\label{eq:NGenContactTransA}
\Psi\colon\ (\tilde t,\tilde x,\tilde y,\tilde u,\tilde u_{\tilde t},\tilde u_{\tilde x},\tilde u_{\tilde y})
=(Z^t,Z^x,Z^y,U,U^t,U^x,U^y),
\end{gather}
meaning that the tuple on the right-hand side consists of smooth functions of $(t,x,y,u,u_t,u_x,u_y)$ with a nonvanishing Jacobian, which additionally satisfies the contact condition
\begin{gather}\label{eq:dNGenContactTransB}
(Z^\mu_\nu+Z^\mu_uu_\nu)U^\mu=U_\nu+U_uu_\nu,\quad
Z^\mu_{u_\nu}U^\mu=U_{u_\nu}.
\end{gather}
Here and below, the indices~$\mu$ and~$\nu$ run through the set $\{t,x,y\}$ and summation over repeated indices is implied.
If the transformation~$\Psi$ is a contact symmetry of the equation~(\ref{eq:SymPotNizhEq}$'$), then $\Psi_*\mathfrak m_{j(1)}\subseteq\mathfrak m_{j(1)}$ for $j=1,3,4,5,6$,
where $\Psi_*$ denotes the pushforward of contact vector fields by the transformation~$\Psi$.
From the prolonged conditions
\smash{$Q^1_{(1)}\in\mathfrak m_{6(1)}$, $Q^2_{(1)}\in\mathfrak m_{5(1)}$,
$Q^{3z}_{(1)}\in\mathfrak m_{4(1)}$,}
\smash{$Q^{4z}_{(1)},Q^{5z}_{(1)}\in\mathfrak m_{3(1)}$} and $Q^7_{(1)},Q^8_{(1)}\in\mathfrak m_{1(1)}$,
we obtain a contact analogue of the system of equations~\eqref{eq:NMainPushforwards}, to which we refer in a similar manner.
From the equation with $i=1$, we specifically derive the constraint $Z^\mu_u=0$.
Then, from the equations with $i=2$, $(i,z)=(3,x)$ and $(i,z)=(3,y)$, we obtain the constraints $Z^\mu_{u_t}=0$, $Z^\mu_{u_x}=0$ and $Z^\mu_{u_y}=0$, respectively. In view of the contact condition~\eqref{eq:dNGenContactTransB}, this also implies that $U_{u_\nu}=0$.
Thus, the contact transformation~$\Psi$ is the first prolongation of a point transformation in the space with coordinates $(t,x,y,u)$, which is necessarily a point symmetry of the equation~{\rm(\ref{eq:SymPotNizhEq}$'$)}.
\end{proof}

The megaideals of Lie invariance pseudoalgebras, these pseudoalgebras themselves (see Section~\ref{subsec:PseudoalgebraDispersionSymm}) and the point-symmetry pseudogroups of other symmetric dispersive models are related to the corresponding objects of their dispersionless counterparts in a manner similar to that of the (symmetric dispersive) Nyzhnyk equation~{\rm(\ref{eq:SymPotNizhEq}$'$)}.

In particular, the pseudoalgebra~$\mathfrak g_{\ref{eq:LaxSymPotNizhEq}'}$ can be embedded as a codimension-one pseudosubalgebra into the pseudoalgebra~$\mathfrak g_{\ref{eq:dNLaxPair}}$ by identifying $\bar P^v$ and~$\bar D^{\rm\psi}$.
It is spanned by the same vector fields as the pseudoalgebra~$\mathfrak g_{\ref{eq:dNLaxPair}}$ except for~$\bar D^{\rm s}$.
It is therefore natural that there also exists a connection between the pseudogroups~$G_{\ref{eq:LaxSymPotNizhEq}'}$ and~$G_{\ref{eq:dNLaxPair}}$.

\begin{theorem}\label{thm:LaxSymPotNizhEqCompletePointSymGroup}\sloppy
The point-symmetry pseudogroup~$G_{\ref{eq:LaxSymPotNizhEq}'}$
of the (linear) Lax representation~{\rm(\ref{eq:LaxSymPotNizhEq}$'$)} of the real Nyzhnyk equation~{\rm(\ref{eq:SymPotNizhEq}$'$)}
is generated by
the point transformations in the space ${\mathbb R^3_{t,x,y}\times\mathbb R^2_{u,\psi}}$,
where the $(t,x,y,u)$- and $\psi$-components are respectively of the form~\eqref{eq:LaxEqPointSymForm} with $A=1$
and of the form \smash{$\tilde\psi=\bar B\psi$} with an arbitrary nonvanishing constant~\smash{$\bar B$},
together with the transformation $\bar{\mathscr J}$: $\tilde t=t$, $\tilde x=y$, $\tilde y=x$, $\tilde u=u$, $\tilde\psi=\psi$.
\end{theorem}

\begin{proof}
To begin with, we construct megaideals of the pseudoalgebra~$\mathfrak g_{\ref{eq:LaxSymPotNizhEq}'}$.
Since it can be embedded as a pseudosubalgebra into the pseudoalgebra~$\mathfrak g_{\ref{eq:dNLaxPair}}$,
we choose among the megaideals of the pseudoalgebra~$\mathfrak g_{\ref{eq:dNLaxPair}}$ those that are also megaideals of the pseudoalgebra~$\mathfrak g_{\ref{eq:LaxSymPotNizhEq}'}$.
Thus,
$\bar{\mathfrak m}_1=\mathfrak g_{\ref{eq:LaxSymPotNizhEq}'}'$,
$\bar{\mathfrak m}_2'=\mathfrak r_{\ref{eq:LaxSymPotNizhEq}'}$,
$\bar{\mathfrak m}_3=\bar{\mathfrak m}_1\cap\mathfrak r_{\ref{eq:LaxSymPotNizhEq}'}$,
$\bar{\mathfrak m}_4=\mathfrak r_{\ref{eq:LaxSymPotNizhEq}'}'$,
$\bar{\mathfrak m}_5=\mathfrak r_{\ref{eq:LaxSymPotNizhEq}'}''$ and
$\bar{\mathfrak m}_6=\mathfrak z(\bar{\mathfrak m}_1)$.
The fact that the radical~$\mathfrak r_{\ref{eq:LaxSymPotNizhEq}'}$ of the pseudoalgebra~$\mathfrak g_{\ref{eq:LaxSymPotNizhEq}'}$ coincides with~$\bar{\mathfrak m}_2'$ can be shown analogously to the proof of Lemma~9 in~\cite{boyk2024a}.
Since the megaideal~$\bar{\mathfrak m}_7$ of the pseudoalgebra~$\mathfrak g_{\ref{eq:dNLaxPair}}$ has no counterpart among the megaideals of the pseudoalgebra~$\mathfrak g_{\ref{eq:LaxSymPotNizhEq}'}$,
we denote by~$\bar{\mathfrak m}_7$ the megaideal
$\mathfrak z(\mathfrak g_{\ref{eq:LaxSymPotNizhEq}'})=\big\langle Z(1),\bar D^{\rm\psi}\big\rangle$ of the pseudoalgebra~$\mathfrak g_{\ref{eq:LaxSymPotNizhEq}'}$, which under the embedding of the pseudoalgebra~$\mathfrak g_{\ref{eq:LaxSymPotNizhEq}'}$ into~$\mathfrak g_{\ref{eq:dNLaxPair}}$ (see footnote~\ref{fnt:EmbeddingsOfMIAsOfLaxRepresentations}) coincides with the megaideal $\big\langle Z(1),\bar P^v\big\rangle$ of the pseudoalgebra~$\mathfrak g_{\ref{eq:dNLaxPair}}$.
The latter megaideal is inessential in the pseudoalgebra~$\mathfrak g_{\ref{eq:dNLaxPair}}$ since it can be represented as a sum of other megaideals
$\big\langle Z(1),\bar P^v\big\rangle=\bar{\mathfrak m}_6\oplus\bar{\mathfrak m}_7$, but this is not the case for the new megaideal~$\bar{\mathfrak m}_7$ with respect to the pseudoalgebra~$\mathfrak g_{\ref{eq:LaxSymPotNizhEq}'}$.
It is also worth noting that, unlike the pseudoalgebra~$\mathfrak g_{\ref{eq:SymPotNizhEq}'}$, where $\mathfrak m_1=\mathfrak g_{\ref{eq:SymPotNizhEq}'}$, the megaideal~$\bar{\mathfrak m}_1$ of the pseudoalgebra~$\mathfrak g_{\ref{eq:LaxSymPotNizhEq}'}$ is a proper megaideal and this entire pseudoalgebra is the sum of its proper megaideals, meaning it is not essential as a megaideal of itself.

We base our proof on that of Theorem~11 from~\cite{boyk2024a}.
Let a point transformation
\[
\bar\Phi\colon\ (\tilde t,\tilde x,\tilde y,\tilde u,\tilde\psi)=(T,X,Y,U,\Psi)
\]
in the space with coordinates $(t,x,y,u,\psi)$ be a point symmetry of the system~(\ref{eq:LaxSymPotNizhEq}$'$).
Here, the tuple of smooth functions $(T,X,Y,U,\Psi)$ of $(t,x,y,u,\psi)$ has a nonvanishing Jacobian.
It is obvious that the transformation~$\bar{\mathscr J}$ is a point symmetry of the system~(\ref{eq:LaxSymPotNizhEq}$'$).
Therefore, computations can be performed up to composition with this transformation.
Collecting the $\psi$-components in the condition $\bar\Phi_*\bar D^{\rm\psi}\subseteq\bar{\mathfrak m}_7$,
we obtain $T_\psi=X_\psi=Y_\psi=0$, $\psi U_\psi=a_{11}$ and $\psi \Psi_\psi=a_{12}\Psi$, where $a_{11}$ and $a_{12}$ are constants.
Analogously, from the conditions
\[
\bar\Phi_*\bar Z(1)\in\bar{\mathfrak m}_6,\quad
\bar\Phi_*\bar P^z(1)\in\bar{\mathfrak m}_3,\ z\in\{x,y\},\quad
\bar\Phi_*\bar D^t(1)\in\bar{\mathfrak m}_1
\]
it follows that $\Psi_t=\Psi_x=\Psi_y=\Psi_u=0$ and therefore $\Psi_\psi\ne0$ due to the nondegeneracy of the transformation~$\bar\Phi$, which implies $a_{12}\ne 0$.
Integrating the derived equations shows that the functions~$T$, $X$ and $Y$ do not depend on $\psi$, $U=W(t,x,y,u)+a_{11}\ln|\psi|$ and $\Psi=\bar B|\psi|^{a_{12}}$, where $\bar B$ is an arbitrary nonvanishing constant.
Since $\bar\varpi_*\mathfrak g_{\ref{eq:LaxSymPotNizhEq}'}=\mathfrak g_{\ref{eq:SymPotNizhEq}'}$, where $\bar\varpi$ is the natural projection from~$\mathbb R^5_{t,x,y,u,\psi}$ onto~$\mathbb R^4_{t,x,y,u}$, due to the obtained form of the $t$-, $x$-, $y$- and $u$-components of the transformation~$\bar\Phi$, the functions~$T$, $X$, $Y$ and $W$ satisfy the same constraints as in the proof of Theorem~11 from~\cite{boyk2024a} for the functions~$T$, $X$, $Y$ and $U$, respectively.
Thus, up to composition with the transformation~$\bar{\mathscr J}$, the $t$-, $x$-, $y$- and $u$-components have the same form as in Theorem~\ref{thm:SymPotNizhEqCompletePointSymGroup}, except for the additional term $a_{11}\ln|\psi|$ in the $u$-component.

Every point transformation~$\bar\Phi$ whose components have the specified form satisfies the condition $\bar\Phi_*\mathfrak g_{\ref{eq:LaxSymPotNizhEq}'}\subseteq\mathfrak g_{\ref{eq:LaxSymPotNizhEq}'}$.
Therefore, no further constraints can be obtained solely within the framework of the algebraic method and the direct method should be employed subsequently.
We derive expressions for the derivatives of the dependent variables $(\tilde u,\tilde\psi)$ with respect to the independent variables $(\tilde t,\tilde x,\tilde y)$ up to the second order in terms of untilded variables and derivatives, successively substitute the found expressions and the expressions for the higher-order derivatives by virtue of the system~(\ref{eq:LaxSymPotNizhEq}$'$) into the system~(\ref{eq:LaxSymPotNizhEq}$'$) written in tilded variables and split the derived equations with respect to the other (parametric) derivatives of the dependent variables~$(u,\psi)$ up to order two.
The resulting system of equations for the parameters of the point-symmetry transformations of the system~(\ref{eq:LaxSymPotNizhEq}$'$) reduces to the equations
$a_{11}+3a_{12}(1-C^3)=0$, $a_{11}+3a_{12}(1-a_{12})(1-C^3)=0$ and
$a_{11}+3a_{12}-3=0$, $-a_{11}+a_{12}^2-3a_{12}+2=0$.
The solution of the first two equations is $C=1$ and $a_{11}=0$.
Substituting these values into the last two equations and solving them yields $a_{12}=1$.
Thus, the $u$-component of the transformation~$\bar\Phi$ does not depend on $\psi$ and $\Psi=\bar B\psi$.
\end{proof}

In view of footnote~\ref{fnt:EmbeddingsOfMIAsOfLaxRepresentations}, the point-symmetry pseudogroup~$G_{\ref{eq:LaxSymPotNizhEq}'}$ of the (linear) Lax representation~(\ref{eq:LaxSymPotNizhEq}$'$) of the equation~(\ref{eq:SymPotNizhEq}$'$) can be embedded into the point-symmetry pseudogroup~$G_{\ref{eq:dNLaxPair}}$ of the nonlinear Lax representation~\eqref{eq:dNLaxPair} of the equation~\eqref{eq:dN} as a pseudosubgroup singled out by the condition $A=1$, where shifts in~$\partial_\vartheta$ correspond to scale transformations of the variable~$\psi$.

Analogously to the previous cases, there is also a connection between the pseudogroups~$G_{\ref{eq:SymNyzhnykSystem}'}$ and~$G_{\ref{eq:dNSystem}}$, which is described by the following statement.

\begin{theorem}\label{thm:SystemPointSymPseudogroup}
The point-symmetry pseudogroup~$G_{\ref{eq:SymNyzhnykSystem}'}$ of the (real symmetric) Nyzhnyk system~{\rm(\ref{eq:SymNyzhnykSystem}$'$)} is a pseudosubgroup of the pseudogroup~$G_{\ref{eq:dNSystem}}$ singled out by the condition $C=1$.
In other words, it is generated by transformations of the form~\eqref{eq:dNSystemPointSymForm} with $C=1$ together with the transformation~$\hat{\mathscr J}$: $\tilde t=t$, $\tilde x=y$, $\tilde y=x$, $\tilde w=w$, $\tilde v^1=v^2$, $\tilde v^2=v^1$.
%The parameter functions $T$, $X^0$ and $Y^0$ are arbitrary smooth functions of the variable~$t$ with $T_t\neq0$.
\end{theorem}

The proof of this theorem is analogous to the proof of Theorem~14 from~\cite{boyk2024a}, but the vector field~$\hat D^{\rm s}$ absent from~$\mathfrak g_{\ref{eq:SymNyzhnykSystem}'}$ should be replaced by~$\hat D^t(t^2)$.
The megaideals~$\mathfrak m_1$ and~$\mathfrak m_3$ of the pseudoalgebra $\mathfrak g_{\ref{eq:dNSystem}}$ are megaideals of the pseudoalgebra~$\mathfrak g_{\ref{eq:SymNyzhnykSystem}'}$.
In terms of the pseudoalgebra~$\mathfrak g_{\ref{eq:SymNyzhnykSystem}'}$, these megaideals can be represented as $\mathfrak m_1=\mathfrak g_{\ref{eq:SymNyzhnykSystem}'}$ and $\mathfrak m_3=\mathfrak r_{\ref{eq:SymNyzhnykSystem}'}$.
Note that, unlike the case of~$\mathfrak g_{\ref{eq:dNSystem}}$, the megaideal~$\mathfrak m_1$ is an improper megaideal of the pseudoalgebra~$\mathfrak g_{\ref{eq:SymNyzhnykSystem}'}$ and this entire pseudoalgebra is not the sum of its proper megaideals, and therefore it is essential as a megaideal of itself.

Since there exists a connection between the pseudogroups~$G_{\ref{eq:LaxSymPotNizhEq}'}$ and~$G_{\ref{eq:dNLaxPair}}$, it is easy to show that there also exists a~connection between the pseudogroups~$G_{\ref{eq:LaxSymNyzhnykSystem}'}$ and~$G_{\ref{eq:LaxdNSystem}}$.

\begin{theorem}\label{thm:LaxSystemPointSymPseudogroup}
The point-symmetry pseudogroup~$G_{\ref{eq:LaxSymNyzhnykSystem}'}$
of the (real linear) Lax representation~{\rm(\ref{eq:LaxSymNyzhnykSystem}$'$)} of the system~{\rm(\ref{eq:SymNyzhnykSystem}$'$)}
is generated by
the transformation
$\check{\mathscr J}$: $\tilde t=t$, $\tilde x=y$, $\tilde y=x$, $\tilde w=w$, $\tilde v^1=v^2$, ${\tilde v^2=v^1}$, $\tilde \psi=\psi$
and the point transformations in the space \smash{$\mathbb R^3_{t,x,y}\times\mathbb R^4_{w,v^1,v^2,\psi}$}
whose $(t,x,y,w,v^1,v^2)$- and $\psi$-components are respectively of the form~\eqref{eq:LaxdNSystemPointSymForm} with $A=1$
and of the form \smash{$\tilde\psi=\check B\psi$} with an arbitrary nonvanishing constant~\smash{$\check B$}.
\end{theorem}

\begin{proof}
We base the proof on that of the preceding theorem, extending the space by one coordinate~$\psi$.
Among the megaideals of the pseudoalgebra~$\mathfrak g_{\ref{eq:LaxdNSystem}}$, we select those that are megaideals of the pseudoalgebra~$\mathfrak g_{\ref{eq:LaxSymNyzhnykSystem}'}$:
$\check{\mathfrak m}_1=\mathfrak g_{\ref{eq:LaxSymNyzhnykSystem}'}'$,
$\check{\mathfrak m}_2'=\mathfrak r_{\ref{eq:LaxSymNyzhnykSystem}'}$,
$\check{\mathfrak m}_3=\check{\mathfrak m}_1\cap\check{\mathfrak m}_2'$,
$\check{\mathfrak m}_4:=\mathfrak z(\mathfrak g_{\ref{eq:LaxSymNyzhnykSystem}'})$.
According to footnote~\ref{fnt:EmbeddingsOfMIAsOfLaxRepresentations}, the point-symmetry pseudogroup~$G_{\ref{eq:LaxSymNyzhnykSystem}'}$ of the (linear) Lax representation~(\ref{eq:LaxSymNyzhnykSystem}$'$) of the system~(\ref{eq:SymNyzhnykSystem}$'$) is a pseudosubgroup of the point-symmetry pseudogroup~$G_{\ref{eq:LaxdNSystem}}$ of the nonlinear Lax representation~\eqref{eq:LaxdNSystem} of the system~\eqref{eq:dNSystem} singled out by the condition $A=1$, where scale transformations of the variable~$\psi$ correspond to shifts in~$\partial_\vartheta$.
\end{proof}

\begin{remark}
Only certain discrete point symmetries of the dispersionless models remain discrete point symmetries of their dispersive counterparts.
Namely, such independent discrete symmetries are the transformations~$\mathscr J$ and~$\mathscr I^{\rm i}$ for the Nyzhnyk equation~(\ref{eq:SymPotNizhEq}$'$), the transformations~$\bar{\mathscr J}$ and~$\bar{\mathscr I}^{\rm i}$ for the Lax representation~(\ref{eq:LaxSymPotNizhEq}$'$) of this equation, the transformations~$\hat{\mathscr J}$ and~$\hat{\mathscr I}^{\rm i}$ for the Nyzhnyk system~(\ref{eq:SymNyzhnykSystem}$'$) and the transformations~$\check{\mathscr J}$ and~$\check{\mathscr I}^{\rm i}$ for the Lax representation~(\ref{eq:LaxSymNyzhnykSystem}$'$) of this system.
\end{remark}

\begin{remark}
The complexification of point-symmetry pseudogroups in the dispersive case is performed analogously to the dispersionless case according to Remark~\ref{rem:OnRealAndComplexCasesForPointSyms}, with the exception that only the parameter function~$T$ requires reparameterization.
Moreover, under the complexification of the dispersive models, only the permutation transformations~$\mathscr J$, $\bar{\mathscr J}$, $\hat{\mathscr J}$ and~$\check{\mathscr J}$ remain as their independent discrete point symmetries.
\end{remark}

\begin{remark}
When constructing the (pseudo)groups of point symmetries for a series of related but different types of models with similar Lie invariance (pseudo)algebras, the use of the algebraic method is much more convenient and appropriate than the application of the direct method, since the corresponding computations share a substantial common part.
Therefore, after finding the point-symmetry (pseudo)group of one of such models by the algebraic method, analogous computations for the other models can be obtained via minor modifications.
The foundational computations for the symmetric Nyzhnyk models are those of the point-symmetry pseudogroups $G_{\ref{eq:dN}}$ and~$G_{\ref{eq:dNSystem}}$ of the (real) dispersionless Nyzhnyk equation~\eqref{eq:dN} and the (real) dispersionless Nyzhnyk system~\eqref{eq:dNSystem}.
The construction of the point-symmetry pseudogroups~$G_{\ref{eq:dNLaxPair}}$ and~$G_{\ref{eq:LaxdNSystem}}$ of the nonlinear Lax representations~\eqref{eq:dNLaxPair} and~\eqref{eq:LaxdNSystem} of the equation~\eqref{eq:dN} and the system~\eqref{eq:dNSystem} relies on the computations for this equation and on all the specified computations, respectively.
Analogous computations for the symmetric dispersive Nyzhnyk models are obtained by modifying the computations for the corresponding dispersionless Nyzhnyk models.
If the direct method were applied, the necessary computations for each of the Nyzhnyk models would have to be carried out independently and in full scope.
\end{remark}

\begin{remark}
In contrast to the dispersionless Nyzhnyk equation~\eqref{eq:dN}, the maximal Lie invariance pseudoalgebra~$\mathfrak g_{\ref{eq:SymPotNizhEq}'}$ of the (dispersive) Nyzhnyk equation~(\ref{eq:SymPotNizhEq}$'$) does not completely determine its point-symmetry pseudogroup~$G_{\ref{eq:SymPotNizhEq}'}$.
This is due to the fact that the pseudoalgebra~$\mathfrak g_{\ref{eq:SymPotNizhEq}'}$ is a proper codimension-one megaideal of the pseudoalgebra~$\mathfrak g_{\ref{eq:dN}}$.
Namely, it does not contain the vector field~$D^{\rm s}$.
At the same time, as a megaideal of the pseudoalgebra~$\mathfrak g_{\ref{eq:dN}}$, it admits all automorphisms of this pseudoalgebra, including those induced by the scale transformations generated by the vector field~$D^{\rm s}$.
Therefore, within a purely algebraic approach, one can only show that the pseudogroup~$G_{\ref{eq:SymPotNizhEq}'}$ is contained in the pseudogroup~$G_{\ref{eq:dN}}$.
Although the Lie invariance pseudoalgebras of all other models, even the dispersionless Nyzhnyk models, do not completely determine the corresponding point-symmetry pseudogroups, the role of the direct method still increases for dispersive models compared to their dispersionless analogues.
\end{remark}

\section{Defining subalgebras}\label{sec:DefiningSubalgs}

In \cite[Sections~4 and~5]{boyk2024a}, it was studied whether the pseudosubalgebras of the maximal Lie invariance pseudoalgebra~$\mathfrak g_{\ref{eq:dN}}$ of the dispersionless Nyzhnyk equation~\eqref{eq:dN}, which naturally arise during the computation of the point-symmetry pseudogroup~$G_{\ref{eq:dN}}$ of this equation, determine the point or contact diffeomorphisms that respectively stabilize this pseudoalgebra or its first prolongation.
First, let us clarify the definition of this stabilization from~\cite{boyk2024a}, distinguishing between the case of Lie algebras and the case of Lie pseudoalgebras.

\begin{definition}\label{def:DefiningSubalg}
A proper subalgebra~$\mathfrak s$ of a Lie algebra~$\mathfrak g$ of vector fields
is called a \emph{subalgebra defining the diffeomorphisms that stabilize the algebra~$\mathfrak g$}
if the conditions $\Phi_*\mathfrak g\subseteq\mathfrak g$ and $\Phi_*\mathfrak s\subseteq\mathfrak g$
are equivalent for any diffeomorphism~$\Phi$ of a manifold~$M$.
\end{definition}

It is straightforward to formulate an analogous definition for pseudoalgebras, provided that the action of local diffeomorphisms on partial vector fields is correctly defined.
Let~$\Omega_1$ and $\Omega_2$ be open subsets of the base manifold~$M$, let~$\Phi$ be a local diffeomorphism of~$M$ with domain~$\Omega_1$ and let~$Q$ be a partial vector field on~$M$ with domain~$\Omega_2$.
Then, within the framework of the partial mapping approach, it is natural to call $\big(\Phi|_{\Omega_1\cap\Omega_2}^{}\big)_*^{}\big(Q|_{\Omega_1\cap\Omega_2}^{}\big)$ the \emph{pushforward of the partial vector field~$Q$ by the local diffeomorphism~$\Phi$}.

\begin{definition}\label{def:DefiningPseudosubalg}
A proper (pseudo)subalgebra~$\mathfrak s$ of a Lie pseudoalgebra~$\mathfrak g$ of partial vector fields is called
a \emph{(pseudo)subalgebra defining the local diffeomorphisms that stabilize the pseudoalgebra~$\mathfrak g$}
if the conditions $\Phi_*\mathfrak g\subseteq\mathfrak g$ and $\Phi_*\mathfrak s\subseteq\mathfrak g$
are equivalent for any local diffeomorphism~$\Phi$ of a manifold~$M$.
\end{definition}

The equivalence of the conditions $\Phi_*\mathfrak g\subseteq\mathfrak g$ and $\Phi_*\mathfrak s\subseteq\mathfrak g$ reduces to the implication $\Phi_*\mathfrak s\subseteq\mathfrak g\Rightarrow\Phi_*\mathfrak g\subseteq\mathfrak g$, since the reverse implication holds trivially.

To find the point-symmetry pseudogroup~$G_{\ref{eq:dN}}$ of the (real) dispersionless Nyzhnyk equation~\eqref{eq:dN} in~\cite{boyk2024a}, only 11 (linearly independent) vector fields from the infinite-dimensional pseudoalgebra~$\mathfrak g_{\ref{eq:dN}}$ were used:
$Z(1)$, $Z(t)$, $R^x(1)$, $R^y(1)$, $P^x(1)$, $P^y(1)$, $P^x(t)$, $P^y(t)$, $D^{\rm s}$, $D^t(1)$ and $D^t(t)$, whose linear span is an 11-dimensional subalgebra~$\mathfrak s_{\ref{eq:dN}}^0$ of the pseudoalgebra~$\mathfrak g_{\ref{eq:dN}}$.
More precisely, two subalgebras were considered simultaneously,
\begin{gather*}
\mathfrak s_{\ref{eq:dN}}^1:=\big\langle Z(1),\,Z(t),\,R^x(1),\,R^y(1),\,P^x(1),\,P^y(1),\,P^x(t),\,P^y(t),\,
D^{\rm s}\big\rangle,\\
\mathfrak s_{\ref{eq:dN}}^2:=\big\langle Z(1),\,Z(t),\,R^x(1),\,R^y(1),\,P^x(1),\,P^y(1),\,P^x(t),\,P^y(t),\,
D^t(1),\,D^t(t)\big\rangle,
\end{gather*}
for which $\mathfrak s_{\ref{eq:dN}}^1+\mathfrak s_{\ref{eq:dN}}^2=\mathfrak s_{\ref{eq:dN}}^0$. In \cite[Sections~4 and~5]{boyk2024a}, it was shown which of these subalgebras of the pseudoalgebra~$\mathfrak g_{\ref{eq:dN}}$ determine the diffeomorphisms stabilizing this pseudoalgebra.
Namely, the following statements were proved~\cite[Theorem~6]{boyk2024a}.

\begin{theorem}\label{thm:dNDefSubalgs}
(i) The subalgebra~$\mathfrak s_{\ref{eq:dN}}^2$ of the pseudoalgebra~$\mathfrak g_{\ref{eq:dN}}$ determines
the diffeomorphisms stabilizing this pseudoalgebra.

(ii) The subalgebra~$\mathfrak s_{\ref{eq:dN}}^1$ and even the subalgebra~$\mathfrak s_{\ref{eq:dN}}^1+\langle D^t(1)\rangle$ of the pseudoalgebra~$\mathfrak g_{\ref{eq:dN}}$ do not possess this property.
\end{theorem}

Additionally, we introduce the notation
\begin{gather*}
\mathfrak s_{\ref{eq:dN}}^3:=\big\langle Z(1),\,Z(t),\,Z(t^2),\,R^x(1),\,R^y(1),\,R^x(t),\,R^y(t)\big\rangle.
\end{gather*}

\begin{corollary}\label{cor:dNContDefSubalgs}
The first prolongation of the subalgebra~$\mathfrak s_{\ref{eq:dN}}^2+\mathfrak s_{\ref{eq:dN}}^3$, which is a subalgebra of the pseudoalgebra~$\mathfrak g^{\rm c}_{\ref{eq:dN}}=\mathfrak g_{\ref{eq:dN}(1)}$, together with its restrictions, determines the diffeomorphisms of the corresponding first-order jet space that stabilize~$\mathfrak g^{\rm c}_{\ref{eq:dN}}$.
\end{corollary}

Based on the proof of Theorem~6 from~\cite{boyk2024a}, presented above as Theorem~\ref{thm:dNDefSubalgs}(i), one can analogously answer the question of whether the subalgebras of the maximal Lie invariance pseudoalgebra~$\mathfrak g_{\ref{eq:dNLaxPair}}$ of the (real) nonlinear Lax representation~\eqref{eq:dNLaxPair} of the equation~\eqref{eq:dN}, which naturally arise when constructing its point-symmetry pseudogroup~$G_{\ref{eq:dNLaxPair}}$ by the algebraic method, determine the point diffeomorphisms stabilizing the pseudoalgebra~$\mathfrak g_{\ref{eq:dNLaxPair}}$.
We denote
\begin{gather*}
\mathfrak s_{\ref{eq:dNLaxPair}}^1:=\big\langle
\bar Z(1),\,\bar Z(t),\,\bar R^x(1),\,\bar R^y(1),\,
\bar P^x(1),\,\bar P^y(1),\,\bar P^x(t),\,\bar P^y(t),\,\bar D^{\rm s},\,\bar P^\vartheta\big\rangle,\\
\mathfrak s_{\ref{eq:dNLaxPair}}^2:=\big\langle
\bar Z(1),\, \bar Z(t),\,\bar R^x(1),\,\bar R^y(1),\,
\bar P^x(1),\,\bar P^y(1),\,\bar P^x(t),\,\bar P^y(t),\,\bar D^t(1),\,\bar D^t(t),\,\bar P^\vartheta\big\rangle.
\end{gather*}

\begin{theorem}\label{thm:dNLaxPairDefSubalgs}
The subalgebra~$\mathfrak s_{\ref{eq:dNLaxPair}}^2$ of the pseudoalgebra~$\mathfrak g_{\ref{eq:dNLaxPair}}$ determines the diffeomorphisms stabilizing this pseudoalgebra.
At the same time, the subalgebra~$\mathfrak s_{\ref{eq:dNLaxPair}}^1$
and even the subalgebra ${\mathfrak s_{\ref{eq:dNLaxPair}}^1+\langle\bar D^t(1)\rangle}$ of the same pseudoalgebra do not possess this property.
\end{theorem}

\begin{proof}
It suffices to modify the proof of Theorem~6 from~\cite{boyk2024a}.
Instead of the pseudoalgebra~$\mathfrak g_{\ref{eq:dN}}$ and a point transformation~$\Phi$ in the space with coordinates $(t,x,y,u)$, we consider the pseudoalgebra~$\mathfrak g_{\ref{eq:dNLaxPair}}$ and a point transformation
\[
\bar\Phi\colon\ (\tilde t,\tilde x,\tilde y,\tilde u,\tilde\vartheta)=(T,X,Y,U,\Theta)
\]
in the extended space with coordinates $(t,x,y,u,\vartheta)$, where $(T,X,Y,U,\Theta)$ is a tuple of smooth functions of $(t,x,y,u,\vartheta)$ with a nonvanishing Jacobian.
Suppose that $\bar\Phi_*\mathfrak s_{\ref{eq:dNLaxPair}}^1\subseteq\mathfrak g_{\ref{eq:dNLaxPair}}$ (or, correspondingly, $\bar\Phi_*\mathfrak s_{\ref{eq:dNLaxPair}}^2\subseteq\mathfrak g_{\ref{eq:dNLaxPair}}$). We denote the basis elements
$\bar Z(1)$, $\bar Z(t)$, $\bar R^z(1)$, $\bar P^z(1)$, $\bar P^z(t)$, $\bar D^{\rm s}$, $\bar D^t(1)$, $\bar D^t(t)$ and $\bar P^\vartheta$ with $z\in\{x,y\}$ of the subalgebras~$\mathfrak s_{\ref{eq:dNLaxPair}}^1$ and~$\mathfrak s_{\ref{eq:dNLaxPair}}^2$ as $Q^1$, $Q^2$, $Q^{3z}$, $Q^{4z}$, $Q^{5z}$, $Q^6$, $Q^7$, $Q^8$ and $Q^0$, respectively.
For each basis element~$Q$ of the subalgebra~$\mathfrak s_{\ref{eq:dNLaxPair}}^1$ (and then independently for the subalgebra~$\mathfrak s_{\ref{eq:dNLaxPair}}^2$), we apply the condition $\bar\Phi_*Q\in\mathfrak g_{\ref{eq:dNLaxPair}}$, which yields the equations
\begin{gather}\label{eq:dNLaxPairMainPushforwardsMod}
\bar\Phi_*Q^\kappa=
\nu^\kappa\tilde P^\vartheta+
\lambda^\kappa\tilde D^{\rm s}+\tilde D^t(\tilde\tau^\kappa)+\tilde P^x(\tilde\chi^\kappa)+\tilde P^y(\tilde\rho^\kappa)
+\tilde R^x(\tilde\alpha^\kappa)+\tilde R^y(\tilde\beta^\kappa)+\tilde Z(\tilde\sigma^\kappa),
\end{gather}
where $\nu^\kappa$ and $\lambda^\kappa$ are constants, $\tilde\tau^\kappa$, $\tilde\chi^\kappa$, $\tilde\rho^\kappa$, $\tilde\alpha^\kappa$, $\tilde\beta^\kappa$ and $\tilde\sigma^\kappa$ are arbitrary smooth functions of~$\tilde t$ and the index~$\kappa$ runs through the sets $\{1,2,3z,4z,5z,6,0\}$ and $\{1,2,3z,4z,5z,7,8,0\}$ for the subalgebras~$\mathfrak s_{\ref{eq:dNLaxPair}}^1$ and~$\mathfrak s_{\ref{eq:dNLaxPair}}^2$, respectively.
Here and below, $z$ runs through the set $\{x,y\}$.
In the notation of the vector fields on the right-hand side of the equations~\eqref{eq:dNLaxPairMainPushforwardsMod}, a tilde is used to emphasize that these vector fields are represented in the transformed coordinates, while a bar is omitted.

First, we consider the common basis elements of the subalgebras~$\mathfrak s_{\ref{eq:dNLaxPair}}^1$ and~$\mathfrak s_{\ref{eq:dNLaxPair}}^2$.
Let us analyze the equation~\eqref{eq:dNLaxPairMainPushforwardsMod} with $\kappa=1,2$.
Collecting the $t$-components, we obtain the equations $T_u=\tilde\tau^1(T)$ and $tT_u=\tilde\tau^2(T)$.
Suppose that $T_u\ne0$.
Then $\tilde\tau^1\tilde\tau^2\ne0$ and the combination of the last two equations reduces to the equation $t=f(T)$ with $f(\tilde t):=\tilde\tau^2(\tilde t)/\tilde\tau^1(\tilde t)$, whence we obtain $1=f_{\tilde t}(T)T_t$ by differentiation with respect to~$t$.
Therefore, the derivative $f_{\tilde t}$ is nonvanishing, and according to the inverse function theorem, we find that $T=T(t)$, which contradicts our assumption.
Hence, $T_u=0$ and $\tilde\tau^1=\tilde\tau^2=0$.

Suppose that $(X_u,Y_u)\ne(0,0)$.
We analyze the equations~\eqref{eq:dNLaxPairMainPushforwardsMod} for $\kappa=1,2$, splitting it componentwise.
Without loss of generality, we can assume that $X_u\ne0$, since the case $Y_u\ne0$ is completely analogous.
Collecting the $x$-components yields the equations $X_u=\lambda^1X+\tilde\chi^1(T)$ and $tX_u=\lambda^2X+\tilde\chi^2(T)$.
Their consequence, $(\lambda^1t-\lambda^2)X+t\tilde\chi^1(T)-\tilde\chi^2(T)=0$, can be split with respect to $X$, which gives $\lambda^1=\lambda^2=0$ and $t\tilde\chi^1=\tilde\chi^2$.
Then the $y$-components simplify to $Y_u=\tilde\rho^1(T)$ and $tY_u=\tilde\rho^2(T)$, and thus they can be combined into the equation $t\tilde\rho^1=\tilde\rho^2$.
Collecting the $u$-components, we obtain the equations
\begin{gather*}
 U_u=-\frac12\tilde\chi^1_{\tilde t}(T)X^2-\frac12\tilde\rho^1_{\tilde t}(T)Y^2
+\tilde\alpha^1(T)X+\tilde\beta^1(T)Y+\tilde\sigma^1(T),
\\
tU_u=-\frac12\tilde\chi^2_{\tilde t}(T)X^2-\frac12\tilde\rho^2_{\tilde t}(T)Y^2
+\tilde\alpha^2(T)X+\tilde\beta^2(T)Y+\tilde\sigma^2(T).
\end{gather*}
Subtracting the first equation multiplied by $t$ from the second one and splitting the resulting equation with respect to $X^2$ and $Y^2$ (which is possible since $T$, $X$ and $Y$ are functionally independent) gives $t\tilde\chi^1_{\tilde t}(T)=\tilde\chi^2_{\tilde t}(T)$ and $t\tilde\rho^1_{\tilde t}(T)=\tilde\rho^2_{\tilde t}(T)$.
The differential consequences of these equations pairwise with the equations $t\tilde\chi^1(T)=\tilde\chi^2(T)$ and $t\tilde\rho^1(T)=\tilde\rho^2(T)$ are the equations $\tilde\chi^1=\tilde\chi^2=0$ and $\tilde\rho^1=\tilde\rho^2=0$, respectively.
Thus, we have $X_u=Y_u=0$, which contradicts our assumption.
Therefore, in fact, $X_u=Y_u=0$, $\lambda^1=\lambda^2=0$, $\tilde\chi^1=\tilde\chi^2=0$ and $\tilde\rho^1=\tilde\rho^2=0$.
Collecting the $\vartheta$-components yields $\nu^1=\nu^2=0$ and $\Theta_u=0$, and therefore $U_u\ne0$ due to the nondegeneracy of the transformation~$\Phi$.

Subtracting the first equation multiplied by $t$ from the second one and splitting the resulting equation with respect to $X^2$ and $Y^2$ (which is possible since $T$, $X$ and $Y$ are functionally independent) gives $t\tilde\chi^1_{\tilde t}(T)=\tilde\chi^2_{\tilde t}(T)$ and $t\tilde\rho^1_{\tilde t}(T)=\tilde\rho^2_{\tilde t}(T)$.
The differential consequences of these equations pairwise with the equations $t\tilde\chi^1(T)=\tilde\chi^2(T)$ and $t\tilde\rho^1(T)=\tilde\rho^2(T)$ are the equations $\tilde\chi^1=\tilde\chi^2=0$ and $\tilde\rho^1=\tilde\rho^2=0$, respectively.
Thus, we have $X_u=Y_u=0$, which contradicts our assumption.
Therefore, in fact, $X_u=Y_u=0$, $\lambda^1=\lambda^2=0$, $\tilde\chi^1=\tilde\chi^2=0$ and $\tilde\rho^1=\tilde\rho^2=0$.
Collecting the $\vartheta$-components yields $\nu^1=\nu^2=0$ and $\Theta_u=0$, and therefore $U_u\ne0$ due to the nondegeneracy of the transformation~$\Phi$.

By collecting the $t$-components in the equations~\eqref{eq:dNLaxPairMainPushforwardsMod} for $\kappa=4z,5z$, we find the equations $T_z=\tilde\tau^{4z}(T)$ and $tT_z=\tilde\tau^{5z}(T)$, whence $t\tilde\tau^{4z}(T)=\tilde\tau^{5z}(T)$.
Suppose that $T_z\ne0$.
Then $\tilde\tau^{4z}(T)\ne0$ and $t=f(T)$ along with $f=\tilde\tau^{5z}/\tilde\tau^{4z}$.
Using similar arguments as at the beginning of the proof, we obtain $T=T(t)$, which contradicts our assumption.
Hence, $T_z=0$ and $\tilde\tau^{4z}=\tilde\tau^{5z}=0$.

Componentwise splitting of the equations~\eqref{eq:dNLaxPairMainPushforwardsMod} with $\kappa=3z$ yields $\tilde\tau^{3z}=\tilde\chi^{3z}=\tilde\rho^{3z}=0$, $\lambda^{3z}=\nu^{3z}=0$ and $zU_u=\tilde\alpha^{3z}(T)X+\tilde\beta^{3z}(T)Y+\tilde\sigma^{3z}(T)$.
In view of the single unstated essential consequence $U_u=\tilde\alpha^1(T)X+\tilde\beta^1(T)Y+\tilde\sigma^1(T)$ of the equations~\eqref{eq:dNLaxPairMainPushforwardsMod} with $\kappa=1,2$, this means that $x$ and~$y$ can be represented as linear fractional functions of $(X,Y)$ with coefficients depending on~$T$.
From the first-order differential consequences of these representations with respect to~$x$ and~$y$, the inequality $X_xY_y-X_yY_x\ne0$ follows.
Therefore, in particular, $(X_x, X_y)\ne(0,0)$.

Let us collect the $x$-components in the equations~\eqref{eq:dNLaxPairMainPushforwardsMod} for $\kappa=4z,5z$. Eliminating~$X_z$ from the obtained equations $X_z=\lambda^{4z}X+\tilde\chi^{4z}(T)$ and $tX_z=\lambda^{5z}X+\tilde\chi^{5z}(T)$ and splitting the result first with respect to~$X$ and then, where possible, with respect to~$t$, we get $\lambda^{4z}=\lambda^{5z}=0$ and $t\tilde\chi^{4z}=\tilde\chi^{5z}$.
We choose $z\in\{x,y\}$ such that $X_z=\tilde\chi^{4z}(T)\ne0$.
Then $t=f(T)$, where $f:=\tilde\chi^{5z}/\tilde\chi^{4z}$.
This implies that the function~$T$ depends only on $t$ with $T_t\ne0$.
Solving the specified representations for~$x$ and~$y$ as equations with respect to~$X$ and~$Y$, we obtain $X={\rm N}^X/{\rm D}$ and $Y={\rm N}^Y/{\rm D}$, where ${\rm N}^X=X^1(t)x+X^2(t)y+X^0(t)$, ${\rm N}^Y=Y^1(t)x+Y^2(t)y+Y^0(t)$ and ${\rm D}=K^1(t)x+K^2(t)y+K^0(t)$ for certain smooth functions $X^0$, $X^1$, $X^2$, $Y^0$, $Y^1$, $Y^2$, $K^0$, $K^1$ and $K^2$ of~$t$.
Under the obtained constraints, the only consequences of the $y$-components of the equations~\eqref{eq:dNLaxPairMainPushforwardsMod} with $\kappa=4z,5z$ are the equations $t\tilde\rho^{4z}=\tilde\rho^{5z}$.

The equations derived by collecting the $x$- and $y$-components in the equations~\eqref{eq:dNLaxPairMainPushforwardsMod} with $\kappa=4z$ reduce to
\begin{gather}\label{eq:Mu4z}
\begin{split}
&X^1{\rm D}-K^1{\rm N}^X=\tilde\chi^{4x}(T){\rm D}^2,\quad
 X^2{\rm D}-K^2{\rm N}^X=\tilde\chi^{4y}(T){\rm D}^2,\\
&Y^1{\rm D}-K^1{\rm N}^Y=\tilde\rho^{4x}(T){\rm D}^2,\quad
 Y^2{\rm D}-K^2{\rm N}^Y=\tilde\rho^{4y}(T){\rm D}^2.
\end{split}
\end{gather}
Suppose that $(K^1,K^2)\ne(0,0)$.
Then either $X^1K^2-X^2K^1\ne0$ or $Y^1K^2-Y^2K^1\ne0$, since otherwise the Jacobian of the functions~$T$, $X$ and~$Y$ vanishes.
Let $X^1K^2-X^2K^1\ne0$.
The case $Y^1K^2-Y^2K^1\ne0$ can be treated analogously up to the permutation of the variables~$x$ and~$y$.
The Jacobian of the functions~${\rm N}^X$ and~${\rm D}$ with respect to $(x,y)$ is nonvanishing and splitting the first two equations in~\eqref{eq:Mu4z} with respect to $({\rm N}^X,{\rm D})$ yields, in particular, the equations $X^1=X^2=0$, which contradict the inequality $X^1K^2-X^2K^1\ne0$.
Therefore, the assumption $(K^1,K^2)\ne(0,0)$ is false.
Thus, $K^1=K^2=0$, meaning that the functions~$X$ and~$Y$ are affine with respect to $(x,y)$ with coefficients depending on~$t$.
We redenote $X^k/K^0$ as~$X^k$ and $Y^k/K^0$ as~$Y^k$ for $k=0,1,2$.
Collecting the $\partial_\vartheta$-components in the equations~\eqref{eq:dNLaxPairMainPushforwardsMod} with $\kappa=4z,5z$, we find the conditions $\nu^{4z}=\nu^{5z}=0$ and $\Theta_z=0$.
The remaining analysis at this stage corresponds to the part with $i=4,5$ of the proof of Theorem~2(i) from~\cite{boyk2024a}.
The only difference is that the function $W^0=W^0(t)$ from that proof now depends also on~$\partial_\vartheta$.

Splitting componentwise the equation~\eqref{eq:dNLaxPairMainPushforwardsMod} with $\kappa=0$, we additionally find that $U_\vartheta=0$ %$\tilde\tau^0=\tilde\chi^0=\tilde\rho^0=\tilde\alpha^0=\tilde\beta^0=\tilde\sigma^0=0$, $\lambda^0=0$,
and $\Theta_\vartheta=\nu^0$, where $\nu^0$ is a constant that is nonvanishing due to the nondegeneracy of the transformation~$\Phi$.

The analysis of the equations~\eqref{eq:dNLaxPairMainPushforwardsMod} with $\kappa=6,7,8$ is also analogous to the corresponding part of the proof of Theorem~6 from~\cite{boyk2024a}.
The differences are the appearance of additional equations derived when collecting the $\partial_\vartheta$-component and the dependence of $W^0$ on~$\partial_\vartheta$.
In the equations~\eqref{eq:dNLaxPairMainPushforwardsMod} with $\kappa=7,8$, we collect the $\tilde t$-, $\tilde x$-, $\tilde u$- and $\partial_\vartheta$-components.
Combining the obtained equations yields $t\tilde\tau^7=\tilde\tau^8$, $\lambda^7=\lambda^8=0$, $\nu^7=\nu^8=0$, $W^0_\vartheta=0$ and $\Theta_t=0$.
Taking only the equation~\eqref{eq:dNLaxPairMainPushforwardsMod} with $\kappa=6$ instead of the equations~\eqref{eq:dNLaxPairMainPushforwardsMod} with $\kappa=7,8$, one cannot obtain the condition $W^0_\vartheta=0$, even if the equation~\eqref{eq:dNLaxPairMainPushforwardsMod} with $\kappa=7$ is additionally considered.
Thus, the condition $\Phi_*\mathfrak s_{\ref{eq:dNLaxPair}}^2\subseteq\mathfrak g_{\ref{eq:dNLaxPair}}$ implies $\Phi_*\mathfrak g_{\ref{eq:dNLaxPair}}\subseteq\mathfrak g_{\ref{eq:dNLaxPair}}$.
In other words, the subalgebra~$\mathfrak s_{\ref{eq:dNLaxPair}}^2$ determines the diffeomorphisms stabilizing the pseudoalgebra~$\mathfrak g_{\ref{eq:dNLaxPair}}$, which cannot be said about the subalgebra~$\mathfrak s_{\ref{eq:dNLaxPair}}^1$ or even about the subalgebra~$\mathfrak s_{\ref{eq:dNLaxPair}}^1+\langle\bar D^t(1)\rangle$.
\end{proof}

In view of the results obtained above, it is also worth analyzing which subalgebras of the maximal Lie invariance pseudoalgebras~$\mathfrak g_{\ref{eq:dNSystem}}$ and~$\mathfrak g_{\ref{eq:LaxdNSystem}}$ of the (real) dispersionless Nyzhnyk system~\eqref{eq:dNSystem} and its Lax representation~\eqref{eq:LaxdNSystem}, respectively, naturally arising when constructing their point-symmetry pseudogroups~$G_{\ref{eq:dNSystem}}$ and~$G_{\ref{eq:LaxdNSystem}}$ by the algebraic method, determine the point diffeomorphisms that stabilize the corresponding pseudoalgebras~$\mathfrak g_{\ref{eq:dNSystem}}$ and~$\mathfrak g_{\ref{eq:LaxdNSystem}}$.
We introduce the notation
\begin{gather*}
\mathfrak s_{\ref{eq:dNSystem}}^1:=\mathfrak s_{\ref{eq:dNSystem}}^0+\big\langle\hat D^t(t^2)\big\rangle,\quad
\mathfrak s_{\ref{eq:dNSystem}}^2:=\mathfrak s_{\ref{eq:dNSystem}}^0+\big\langle\hat D^{\rm s}\big\rangle\\\quad\mbox{with}\quad
\mathfrak s_{\ref{eq:dNSystem}}^0:=
\big\langle\hat P^x(1),\,\hat P^y(1),\,\hat P^x(t),\,\hat P^y(t),\,
\hat P^x(t^2),\,\hat P^y(t^2),\,\hat D^t(1),\,\hat D^t(t)\big\rangle,
\\[1ex]
\mathfrak s_{\ref{eq:LaxdNSystem}}^1:=\mathfrak s_{\ref{eq:LaxdNSystem}}^0+\big\langle\check D^t(t^2)\big\rangle,\quad
\mathfrak s_{\ref{eq:LaxdNSystem}}^2:=\mathfrak s_{\ref{eq:LaxdNSystem}}^0+\big\langle\check D^{\rm s}\big\rangle\\\quad\mbox{with}\quad
\mathfrak s_{\ref{eq:LaxdNSystem}}^0:=
\big\langle\check P^x(1),\,\check P^y(1),\,\check P^x(t),\,\check P^y(t),\,
\check P^x(t^2),\,\check P^y(t^2),\,\check D^t(1),\,\check D^t(t),\,\check P^\vartheta\big\rangle
\end{gather*}
for the specified subalgebras of the pseudoalgebras~$\mathfrak g_{\ref{eq:dNSystem}}$ and~$\mathfrak g_{\ref{eq:LaxdNSystem}}$, respectively.

\begin{theorem}\label{thm:dNSystemDefSubalgs}
The subalgebras~$\mathfrak s_{\ref{eq:dNSystem}}^1$ and~$\mathfrak s_{\ref{eq:dNSystem}}^2$ of the pseudoalgebra~$\mathfrak g_{\ref{eq:dNSystem}}$ determine the diffeomorphisms stabilizing this pseudoalgebra.
\end{theorem}

\begin{proof}
We follow the proof of Theorem~14 from~\cite{boyk2024a}.
Consider a point transformation
\[
\hat\Phi\colon\ (\tilde t,\tilde x,\tilde y,\tilde w,\tilde v^1,\tilde v^2)=(T,X,Y,W,V^1,V^2)
\]
in the space with coordinates $(t,x,y,w,v^1,v^2)$ and denote the basis elements $\hat P^z(1)$, $\hat P^z(t)$, $\hat P^z(t^2)$, $\hat D^t(1)$, $\hat D^t(t)$, $\hat D^t(t^2)$ and $\hat D^{\rm s}$ with $z\in\{x,y\}$ of the subalgebra~$\mathfrak s_{\ref{eq:dNSystem}}^1+\mathfrak s_{\ref{eq:dNSystem}}^2$ as $Q^{1z}$, $Q^{2z}$, $Q^{3z}$, $Q^4$, $Q^5$, $Q^6$ and $Q^7$, respectively.
Also, let $v^x:=v^1$ and $v^y:=v^2$.
For each basis element~$Q$ of the subalgebra~$\mathfrak s_{\ref{eq:dNSystem}}^1$ or the subalgebra~$\mathfrak s_{\ref{eq:dNSystem}}^2$, we apply the condition $\hat\Phi_*Q\in\mathfrak g_{\ref{eq:dNSystem}}$.
In other words, we solve the equation
\begin{gather}\label{eq:dNSystemMainPushforwardsMod}
\hat\Phi_*Q^\kappa=\tilde D^t(\tilde\tau^\kappa)+\lambda^\kappa\tilde D^{\rm s}+\tilde P^x(\tilde\chi^\kappa)+\tilde P^y(\tilde\rho^\kappa),
\end{gather}
where
$\lambda^\kappa$ are constants,
$\tilde\tau^\kappa$, $\tilde\chi^\kappa$ and $\tilde\rho^\kappa$ are arbitrary smooth functions of~$\tilde t$ and the index~$\kappa$ runs through the set $\big\{1z,2z,3z,4,5,6\mid z\in\{x,y\}\big\}$ or the set $\big\{1z,2z,3z,4,5,7\mid z\in\{x,y\}\big\}$.
In the notation of the vector fields on the right-hand side of the equation~\eqref{eq:dNSystemMainPushforwardsMod}, a tilde is used to emphasize that these vector fields are represented in the transformed coordinates, while the hat is omitted.
Note that for each value of~$\kappa$, at least one of the parameters on the right-hand side of the corresponding equation~\eqref{eq:dNSystemMainPushforwardsMod} should be nonvanishing.

To begin with, let us consider the equations related to the common part
$\mathfrak s_{\ref{eq:dNSystem}}^0=\mathfrak s_{\ref{eq:dNSystem}}^1\cap\mathfrak s_{\ref{eq:dNSystem}}^2$
of the subalgebras~$\mathfrak s_{\ref{eq:dNSystem}}^1$ and~$\mathfrak s_{\ref{eq:dNSystem}}^2$.
We split componentwise the identities
\begin{gather}\label{eq:dNSystemMainPushforwardsComb1}
\Phi_*(Q^{3z}-2tQ^{2z}+t^2Q^{1z})=
\Phi_*Q^{3z}-2\Phi_*(t)\Phi_*Q^{2z}+\Phi_*(t^2)\Phi_*Q^{1z}=0,\quad z\in\{x,y\}.
\end{gather}
Collecting the $t$-components in these identities yields the equations
\[
T_z=\tau^{1z}(T),\quad
tT_z-T_{v^z}=\tau^{2z}(T),\quad
t^2\tilde\tau^{1z}(T)-2t\tilde\tau^{2z}(T)+\tilde\tau^{3z}(T)=0.
\]
If $(\tilde\tau^{1z},\tilde\tau^{2z})\ne(0,0)$ for some~$z$, then the last equality implies that~$t$ can be expressed in terms of~$T$, and thus~$T$ is a function of~$t$ alone, whence $T_z=T_{v^z}=0$, which means $\tilde\tau^{1z}=\tilde\tau^{2z}=0$, leading to a contradiction.
Consequently, $T_z=T_{v^z}=0$ and $\tilde\tau^{1z}=\tilde\tau^{2z}=\tilde\tau^{3z}=0$.
Similarly, from the $x$- and $y$-components in the identities, after splitting with respect to~$\tilde x$ and~$\tilde y$, we obtain the equations
\begin{gather*}
\lambda^{1z}=\lambda^{2z}=\lambda^{3z}=0,
\\
\tilde\chi^{3z}(T)-2t\tilde\chi^{2z}(T)+t^2\tilde\chi^{1z}(T)=0,\quad
\tilde\chi^{3z}_{\tilde t}(T)-2t\tilde\chi^{2z}_{\tilde t}(T)+t^2\tilde\chi^{1z}_{\tilde t}(T)=0,
\\
\tilde\rho^{3z}(T)-2t\tilde\rho^{2z}(T)+t^2\tilde\rho^{1z}(T)=0,\quad
\tilde\rho^{3z}_{\tilde t}(T)-2t\tilde\rho^{2z}_{\tilde t}(T)+t^2\tilde\rho^{1z}_{\tilde t}(T)=0,
\end{gather*}
whose differential consequences have the form
$\tilde\chi^{2z}(T)=t\tilde\chi^{1z}(T)$ and
$\tilde\rho^{2z}(T)=t\tilde\rho^{1z}(T)$.
From the last two equations, by analogy with the previous proofs, we find that $T=T(t)$ with $T_t\ne0$.
Then the equation~\eqref{eq:dNSystemMainPushforwardsMod}$_{1z}$ and the combination $\Phi_*(t)\eqref{eq:dNSystemMainPushforwardsMod}_{1z}-\eqref{eq:dNSystemMainPushforwardsMod}_{2z}$ split into
 \begin{gather*}
X_z=\tilde\chi^{1z}(T),\quad X_{v^z}=t\tilde\chi^{1z}(T)-\tilde\chi^{2z}(T)=0,\\
Y_z=\tilde\rho^{1z}(T),\quad Y_{v^z}=t\tilde\rho^{1z}(T)-\tilde\rho^{2z}(T)=0,\quad
W_z=W_{v^z}=0,\\
V^1_z=-\tilde\chi^{1z}_{\tilde t}(T)=-\frac{X_{zt}}{T_t},\quad
V^1_{v^z}=\tilde\chi^{2z}_{\tilde t}(T)-t\tilde\chi^{1z}_{\tilde t}(T)=\frac{\tilde\chi^{1z}(T)}{T_t}=\frac{X_z}{T_t},\\
V^2_z=-\tilde\rho^{1z}_{\tilde t}(T)=-\frac{Y_{zt}}{T_t},\quad
V^2_{v^z}=\tilde\rho^{2z}_{\tilde t}(T)-t\tilde\rho^{1z}_{\tilde t}(T)=\frac{\tilde\rho^{1z}(T)}{T_t}=\frac{Y_z}{T_t}.
\end{gather*}
Taking into account all the obtained equations, we derive the following form for the components of the transformation~$\Phi$:
\begin{gather}\label{eq:dNSystemPreliminaryFormOfPointSyms}
\begin{split}&
%T=T(t),\quad
X=X^1(t)x+X^2(t)y+X^0(t,w),\quad
Y=Y^1(t)x+Y^2(t)y+Y^0(t,w),\\&
V^1=\frac{X^1}{T_t}v^1+\frac{X^2}{T_t}v^2-\frac{X^1_t}{T_t}x-\frac{X^2_t}{T_t}y+\breve V^1(t,w),\\&
V^2=\frac{Y^1}{T_t}v^1+\frac{Y^2}{T_t}v^2-\frac{Y^1_t}{T_t}x-\frac{Y^2_t}{T_t}y+\breve V^2(t,w),\quad
W=W(t,w),
\end{split}
\end{gather}
where $X^1$, $X^2$, $X^0$, $Y^1$, $Y^2$, $Y^0$, $W$, $\breve V^1$ and $\breve V^2$ are arbitrary smooth functions of their arguments with $W_w(X^1Y^2-X^2Y^1)\ne0$; for comparison, see the system~(18) in~\cite{boyk2024a}.

The equation~\eqref{eq:dNSystemMainPushforwardsMod}$_4$ and $3\eqref{eq:dNSystemMainPushforwardsMod}_5-3\Phi_*(t)\eqref{eq:dNSystemMainPushforwardsMod}_4$ split into the system
\begin{gather}\label{eq:dNSystemMainPushforwardsComb2Split}
\begin{split}&
\lambda^4=\lambda^5=0,\quad\tilde\tau^4(T)=T_t,\\&
X_t= \tfrac13\tilde\tau^4_{\tilde t}(T)X+\tilde\chi^4(T),\quad
Y_t= \tfrac13\tilde\tau^4_{\tilde t}(T)Y+\tilde\rho^4(T),\quad
W_t=-\tfrac23\tilde\tau^4_{\tilde t}(T)W,\\&
V^1_t=-\tfrac23\tilde\tau^4_{\tilde t}(T)V^1-\tfrac13\tilde\tau^4_{\tilde t\tilde t}(T)X-\tilde\chi^4_{\tilde t}(T),\quad
V^2_t=-\tfrac23\tilde\tau^4_{\tilde t}(T)V^2-\tfrac13\tilde\tau^4_{\tilde t\tilde t}(T)Y-\tilde\rho^4_{\tilde t}(T),
\\[1ex]&
\tilde\tau^5(T)=tT_t,\quad
xX_x+yX_y-2wX_w=X+3\tilde\chi^5(T)-3t\tilde\chi^4(T),\\&
xY_x+yY_y-2wY_w=Y+3\tilde\rho^5(T)-3t\tilde\rho^4(T),\quad
wW_w=W,\\&
xV^1_x+yV^1_y-2wV^1_w-2v^1V^1_{v^1}-2v^2V^1_{v^2}=-2V^1+(T_t^{-1})_tX-3\tilde\chi^5_{\tilde t}(T)+3t\tilde\chi^4_{\tilde t}(T),\\&
xV^2_x+yV^2_y-2wV^2_w-2v^1V^2_{v^1}-2v^2V^2_{v^2}=-2V^2+(T_t^{-1})_tY-3\tilde\rho^5_{\tilde t}(T)+3t\tilde\rho^4_{\tilde t}(T).
\end{split}
\end{gather}
This system has already been partially simplified by taking into account its equations $\tilde\tau^4(T)=T_t$ and $\tilde\tau^5(T)=tT_t$, which imply
\[
\tilde\tau^4_{\tilde t}(T)=\frac{T_{tt}}{T_t},\quad
\tilde\tau^5_{\tilde t}(T)-t\tilde\tau^4_{\tilde t}(T)=1,\quad
\tilde\tau^4_{\tilde t\tilde t}(T)=\frac1{T_t}\left(\frac{T_{tt}}{T_t}\right)_{\!t},\quad
\tilde\tau^5_{\tilde t\tilde t}(T)-t\tilde\tau^4_{\tilde t\tilde t}(T)=-\left(\frac1{T_t}\right)_{\!t}.
\]
Substituting the form~\eqref{eq:dNSystemPreliminaryFormOfPointSyms} of the components of the transformation~$\Phi$ into this system, splitting the extended system with respect to~$(x,y,v^1,v^2)$, and solving the resulting equations, we obtain (here and below $j=1,2$)
\begin{gather*}
X^j_t=\frac{T_{tt}}{3T_t}X^j,\quad
X^0_t=\frac{T_{tt}}{3T_t}X^0+\tilde\chi^4(T),\quad
-2wX^0_w=X^0+3\tilde\chi^5(T)-3t\tilde\chi^4(T),
\\
Y^j_t=\frac{T_{tt}}{3T_t}Y^j,\quad
Y^0_t=\frac{T_{tt}}{3T_t}Y^0+\tilde\rho^4(T),\quad
-2wY^0_w=Y^0+3\tilde\rho^5(T)-3t\tilde\rho^4(T),
\\
\breve V^1_t=-\frac{2T_{tt}}{3T_t}\breve V^1-\left(\frac{T_{tt}}{T_t}\right)_{\!t}\frac{X^0}{3T_t}-\tilde\chi^4_{\tilde t}(T),\quad
w\breve V^1_w=\breve V^1+\frac{X^0_t}{T_t},
\\
\breve V^2_t=-\frac{2T_{tt}}{3T_t}\breve V^2-\left(\frac{T_{tt}}{T_t}\right)_{\!t}\frac{Y^0}{3T_t}-\tilde\rho^4_{\tilde t}(T),\quad
w\breve V^2_w=\breve V^2+\frac{Y^0_t}{T_t},
\\
W_t=-\frac{2T_{tt}}{3T_t}W,\quad
wW_w=W.
\end{gather*}

The completion of the proof is carried out separately for the subalgebras~$\mathfrak s_{\ref{eq:dNSystem}}^1$ and~$\mathfrak s_{\ref{eq:dNSystem}}^2$ by splitting the equation~\eqref{eq:dNSystemMainPushforwardsMod}$_6$ or~\eqref{eq:dNSystemMainPushforwardsMod}$_7$, respectively.
In both cases, using the results obtained above, we have $X^0=X^0(t)$ and $Y^0=Y^0(t)$.
Then the equations for the other parameter functions in the expressions for the components of the transformation~$\Phi$ integrate to
\begin{gather}\label{eq:dNSystemPreliminaryFormOfPointSyms2}
X^j=A_jT_t^{1/3},\ \
Y^j=B_jT_t^{1/3},\ \
W=\frac{Cw}{T_t^{2/3}},\ \
\breve V^1=\frac{E_1w}{T_t^{2/3}}-\frac{X^0_t}{T_t},\ \
\breve V^2=\frac{E_2w}{T_t^{2/3}}-\frac{Y^0_t}{T_t},
\end{gather}
where $A_j$, $B_j$, $C$ and $E_j$, $j=1,2$, are constants with $C(A_1B_2-A_2B_1)\ne0$.
All the obtained diffeomorphisms stabilize the pseudoalgebra~$\mathfrak g_{\ref{eq:dNSystem}}$.
\end{proof}

\begin{theorem}\label{thm:LaxdNSystemDefSubalgs}
The subalgebras~$\mathfrak s_{\ref{eq:LaxdNSystem}}^1$ and~$\mathfrak s_{\ref{eq:LaxdNSystem}}^2$ of the pseudoalgebra~$\mathfrak g_{\ref{eq:LaxdNSystem}}$ determine the diffeomorphisms stabilizing this pseudoalgebra.
\end{theorem}

\begin{proof}
The proof of this theorem can be obtained by modifying the proof of Theorem~\ref{thm:dNSystemDefSubalgs} in an analogous way to how the proof of Theorem~6 in~\cite{boyk2024a} was modified above into the proof of Theorem~\ref{thm:dNLaxPairDefSubalgs}.
One should take into account
the appearance of the additional dependent variable~$\vartheta$,
the prolongation of the vector fields from~$\mathfrak g_{\ref{eq:dNSystem}}$ to~$\vartheta$
and the appearance of the one more generating vector field~$\check P^\vartheta$ in~$\mathfrak g_{\ref{eq:LaxdNSystem}}$.

Consider a general point transformation~$\check\Phi$ as at the beginning of the proof of Theorem~\ref{thm:LaxdNSystemPointSymPseudogroup}.
We denote the basis elements $\check P^z(1)$, $\check P^z(t)$, $\check P^z(t^2)$, $\check D^t(1)$, $\check D^t(t)$, $\check D^t(t^2)$, $\check D^{\rm s}$ and $\check P^\vartheta$ with $z\in\{x,y\}$ of the subalgebra~$\mathfrak s_{\ref{eq:LaxdNSystem}}^1+\mathfrak s_{\ref{eq:LaxdNSystem}}^2$ as $Q^{1z}$, $Q^{2z}$, $Q^{3z}$, $Q^4$, $Q^5$, $Q^6$, $Q^7$ and $Q^0$, respectively.
For each basis element~$Q$ of the subalgebra~$\mathfrak s_{\ref{eq:LaxdNSystem}}^1$ or the subalgebra~$\mathfrak s_{\ref{eq:LaxdNSystem}}^2$, we apply the condition $\check\Phi_*Q\in\mathfrak g_{\ref{eq:LaxdNSystem}}$, which yields the equation
\begin{gather}\label{eq:eq:LaxdNSystemMainPushforwardsMod}
\check\Phi_*Q^\kappa=
\nu^\kappa\tilde P^\vartheta+
\lambda^\kappa\tilde D^{\rm s}+\tilde D^t(\tilde\tau^\kappa)+\tilde P^x(\tilde\chi^\kappa)+\tilde P^y(\tilde\rho^\kappa),
\end{gather}
where
$\nu^\kappa$ and $\lambda^\kappa$ are constants,
$\tilde\tau^\kappa$, $\tilde\chi^\kappa$ and $\tilde\rho^\kappa$ are arbitrary smooth functions of~$\tilde t$, and the index~$\kappa$ runs through the set $\big\{1z,2z,3z,4,5,6,0\mid z\in\{x,y\}\big\}$ or the set $\big\{1z,2z,3z,4,5,7,0\mid z\in\{x,y\}\big\}$.
As before, in the notation of the vector fields on the right-hand side of the equation~\eqref{eq:eq:LaxdNSystemMainPushforwardsMod}, a tilde is used to emphasize that these vector fields are represented in the transformed coordinates, while the check is omitted.

First, let us consider the equations associated with the common part $\mathfrak s_{\ref{eq:LaxdNSystem}}^0=\mathfrak s_{\ref{eq:LaxdNSystem}}^1\cap\mathfrak s_{\ref{eq:LaxdNSystem}}^2$ of the subalgebras~$\mathfrak s_{\ref{eq:LaxdNSystem}}^1$ and~$\mathfrak s_{\ref{eq:LaxdNSystem}}^2$.
Sequential splitting of an identity analogous to~\eqref{eq:dNSystemMainPushforwardsComb1} and the equations~\eqref{eq:eq:LaxdNSystemMainPushforwardsMod}$_{1z}$ and $\Phi_*(t)\eqref{eq:eq:LaxdNSystemMainPushforwardsMod}_{1z}-\eqref{eq:eq:LaxdNSystemMainPushforwardsMod}_{2z}$ additionally yields the equations $\nu^{1z}=\nu^{2z}=\nu^{3z}=0$ and $\Theta_z=\Theta_{v^z}=0$.
Hence, $\Theta=\Theta(t,w,\vartheta)$ and the $(t,x,y,w,v^1,v^2)$-components of the transformation~$\check\Phi$ have the form~\eqref{eq:dNSystemPreliminaryFormOfPointSyms}, with the sole difference that, in addition to $(t,w)$, the parameter functions $X^0$, $Y^0$, $W$, $\breve V^1$ and $\breve V^2$ also depend on~$\partial_\vartheta$, with $W_w\Theta_\vartheta-W_\vartheta\Theta_w\ne0$.
Splitting the equation~\eqref{eq:eq:LaxdNSystemMainPushforwardsMod}$_4$ and $3\eqref{eq:eq:LaxdNSystemMainPushforwardsMod}_5-3\Phi_*(t)\eqref{eq:eq:LaxdNSystemMainPushforwardsMod}_4$, we obtain the extension of the system~\eqref{eq:dNSystemMainPushforwardsComb2Split} by the equations $\Theta_t=\nu^4$ and $t\Theta_t-\tfrac23w\Theta_w=\nu^5$, while the equation~\eqref{eq:eq:LaxdNSystemMainPushforwardsMod}$_0$ splits into
$\tilde\tau^0=0$,
$X_\vartheta=\lambda^0X+\tilde\chi^0$,
$Y_\vartheta=\lambda^0Y+\tilde\rho^0$,
$W_\vartheta=\lambda^0W$,
\smash{$\breve V^1_\vartheta=\lambda^0\breve V^1-\tilde\chi^0_t$},
\smash{$\breve V^2_\vartheta=\lambda^0\breve V^2-\tilde\rho^0_t$},
$\Theta_\vartheta=\frac32\lambda^0\Theta+\nu^0$.
Splitting the latter equations with respect to~$X$ and~$Y$ along $(x,y)$ in view of the representations~\eqref{eq:dNSystemPreliminaryFormOfPointSyms}, we deduce that $\lambda^0=0$, and therefore $W_\vartheta=0$.

Let us now consider separately the equations corresponding to the complements of $\mathfrak s_{\ref{eq:LaxdNSystem}}^0$ in~$\mathfrak s_{\ref{eq:LaxdNSystem}}^1$ and in~$\mathfrak s_{\ref{eq:LaxdNSystem}}^2$.

Componentwise splitting of the combination of equations
$\Phi_*(t^2)\eqref{eq:eq:LaxdNSystemMainPushforwardsMod}_4
-2\Phi_*(t)\eqref{eq:eq:LaxdNSystemMainPushforwardsMod}_5
+\eqref{eq:eq:LaxdNSystemMainPushforwardsMod}_6$ yields
\begin{gather*}
\nu^4t^2-2\nu^5t+\nu^6=0,\\
t^2\tilde\tau^4(T)-2t\tilde\tau^5(T)+\tilde\tau^6(T)=0,\quad
t^2\tilde\tau^4_{\tilde t}(T)-2t\tilde\tau^5_{\tilde t}(T)+\tilde\tau^6_{\tilde t}(T)+3\lambda^6=0,\\
t^2\tilde\chi^4(T)-2t\tilde\chi^5(T)+\tilde\chi^6(T)=0,\quad
t^2\tilde\rho^4(T)-2t\tilde\rho^5(T)+\tilde\rho^6(T)=0,\\
2(xV^1_{v^1}+yV^1_{v^2})=
\big(t^2\tilde\tau^4_{\tilde t\tilde t}(T)-2t\tilde\tau^5_{\tilde t\tilde t}(T)+\tilde\tau^6_{\tilde t\tilde t}(T)\big)X
+t^2\tilde\chi^4_{\tilde t}(T)-2t\tilde\chi^5_{\tilde t}(T)+\tilde\chi^6_{\tilde t}(T),\\
2(xV^2_{v^1}+yV^2_{v^2})=
\big(t^2\tilde\tau^4_{\tilde t\tilde t}(T)-2t\tilde\tau^5_{\tilde t\tilde t}(T)+\tilde\tau^6_{\tilde t\tilde t}(T)\big)Y
+t^2\tilde\rho^4_{\tilde t}(T)-2t\tilde\rho^5_{\tilde t}(T)+\tilde\rho^6_{\tilde t}(T),
\end{gather*}
whence
\begin{gather*}
\nu^4=\nu^5=\nu^6=0,\quad \tau^6(T)=t^2T_t,\quad \lambda^6=0,\quad
t^2\tilde\tau^4_{\tilde t\tilde t}(T)-2t\tilde\tau^5_{\tilde t\tilde t}(T)+\tilde\tau^6_{\tilde t\tilde t}(T)=\frac2{T_t},\\
t^2\tilde\chi^4_{\tilde t}(T)-2t\tilde\chi^5_{\tilde t}(T)+\tilde\chi^6_{\tilde t}(T)=
\frac2{T_t}(\tilde\chi^5_{\tilde t}(T)-t\tilde\chi^4_{\tilde t}(T)),\\
t^2\tilde\rho^4_{\tilde t}(T)-2t\tilde\rho^5_{\tilde t}(T)+\tilde\rho^6_{\tilde t}(T)=
\frac2{T_t}(\tilde\rho^5_{\tilde t}(T)-t\tilde\rho^4_{\tilde t}(T)).
\end{gather*}
Consequently,
$X^0=t\tilde\chi^4_{\tilde t}(T)-\tilde\chi^5_{\tilde t}(T)$ and
$Y^0=t\tilde\rho^4_{\tilde t}(T)-\tilde\rho^5_{\tilde t}(T)$,
and therefore $\tilde\chi^0=\tilde\rho^0=0$.

From the equation~\eqref{eq:eq:LaxdNSystemMainPushforwardsMod}$_7$, it follows that
$\tilde\tau^7=0$, $wW_w=\lambda^7W$ (whence $\lambda^7=1$),
$wX^0_w+\tfrac32\vartheta X^0_\vartheta=X^0+\tilde\chi^7$,
$wY^0_w+\tfrac32\vartheta Y^0_\vartheta=Y^0+\tilde\rho^7$ and
$w\Theta_w+\tfrac32\vartheta\Theta_\vartheta=\tfrac32\Theta+\nu^7$,
which implies  $\tilde\chi^0=\tilde\rho^0=0$ and $\nu^4=0$.

Thus, taking into account all the derived equations, in both cases we also have
$X^0_w=Y^0_w=0$,
$X^0_\vartheta=Y^0_\vartheta=0$,
$\breve V^1_\vartheta=\breve V^2_\vartheta=0$,
$\Theta_t=\Theta_w=0$ and $\Theta_\vartheta=\const$.
All diffeomorphisms satisfying these and all the equations above
stabilize the pseudoalgebra~$\mathfrak g_{\ref{eq:LaxdNSystem}}$.
\end{proof}

Furthermore, since $\mathfrak s_{\ref{eq:dN}}^2$ and $\mathfrak s_{\ref{eq:dN}}^2+\mathfrak s_{\ref{eq:dN}}^3$ are subalgebras of the maximal Lie invariance pseudoalgebra~$\mathfrak g_{\ref{eq:SymPotNizhEq}'}$ of the (symmetric) potential Nyzhnyk equation~(\ref{eq:SymPotNizhEq}$'$), which coincides with the derived  pseudoalgebra~$\mathfrak g_{\ref{eq:dN}}'$ of the pseudoalgebra~$\mathfrak g_{\ref{eq:dN}}$ (i.e., it is its megaideal), \emph{Theorem~\ref{thm:dNDefSubalgs}(i) and its Corollary~\ref{cor:dNContDefSubalgs} also hold for the pseudoalgebra~$\mathfrak g_{\ref{eq:SymPotNizhEq}'}$}.
Similarly, up to replacing $\bar P^\vartheta$ by $\bar D^{\rm\psi}$, the maximal Lie invariance pseudoalgebra~$\mathfrak g_{\text{\ref{eq:LaxSymPotNizhEq'}}}$ of the linear Lax representation~\eqref{eq:LaxSymPotNizhEq'} of this equation embeds into the pseudoalgebra~$\mathfrak g_{\ref{eq:dNLaxPair}}$, coincides with its derived  pseudoalgebra (meaning it is its megaideal), and also contains $\mathfrak s_{\ref{eq:dNLaxPair}}^2$ as its subalgebra.
Therefore, \emph{the first statement of Theorem~\ref{thm:dNLaxPairDefSubalgs} holds for the pseudoalgebra~$\mathfrak g_{\text{\ref{eq:LaxSymPotNizhEq'}}}$}.

By analyzing the proofs of Theorems~\ref{thm:dNSystemDefSubalgs} and~\ref{thm:LaxdNSystemDefSubalgs}, similar assertions can be obtained for the maximal Lie invariance pseudoalgebras~$\mathfrak g_{\ref{eq:SymNyzhnykSystem}'}$ and~$\mathfrak g_{\ref{eq:LaxSymNyzhnykSystem}'}$ of the (real symmetric) Nyzhnyk system~(\ref{eq:SymNyzhnykSystem}$'$) and its (linear) Lax representation~(\ref{eq:LaxSymNyzhnykSystem}$'$), respectively.
The pseudoalgebra~$\mathfrak g_{\ref{eq:SymNyzhnykSystem}'}$ is a megaideal of the pseudoalgebra~$\mathfrak g_{\ref{eq:dNSystem}}$ as its derived algebra, and under the substitution $\psi={\rm e}^\vartheta$, the pseudoalgebra~$\mathfrak g_{\ref{eq:LaxSymNyzhnykSystem}'}$ coincides with the derived  pseudoalgebra of the pseudoalgebra~$\mathfrak g_{\ref{eq:LaxdNSystem}}$.
Therefore, given the results obtained above, it is straightforward to single out the respective subalgebras of the pseudoalgebras~$\mathfrak g_{\ref{eq:SymNyzhnykSystem}'}$ and~$\mathfrak g_{\ref{eq:LaxSymNyzhnykSystem}'}$:
\begin{gather*}
\mathfrak s_{\ref{eq:SymNyzhnykSystem}'}:=
\big\langle\hat P^x(1),\,\hat P^y(1),\,\hat P^x(t),\,\hat P^y(t),\,
\hat P^x(t^2),\,\hat P^y(t^2),\,\hat D^t(1),\,\hat D^t(t),\,\hat D^t(t^2)\big\rangle,\\
\mathfrak s_{\ref{eq:LaxSymNyzhnykSystem}'}:=
\big\langle\check P^x(1),\,\check P^y(1),\,\check P^x(t),\,\check P^y(t),\,
\check P^x(t^2),\,\check P^y(t^2),\,\hat D^t(1),\,\check D^t(t),\,\check D^t(t^2),\,\check D^{\rm\psi}\big\rangle,
\end{gather*}
which determine the point diffeomorphisms stabilizing these pseudoalgebras.
It is precisely the subalgebras~$\mathfrak s_{\ref{eq:SymNyzhnykSystem}'}$ and~$\mathfrak s_{\ref{eq:LaxSymNyzhnykSystem}'}$ that arise naturally when constructing the point-symmetry pseudogroups~$G_{\ref{eq:SymNyzhnykSystem}'}$ and~$G_{\ref{eq:LaxSymNyzhnykSystem}'}$ of the systems~(\ref{eq:SymNyzhnykSystem}$'$) and~(\ref{eq:LaxSymNyzhnykSystem}$'$) by the algebraic method.

\begin{corollary}\label{cor:SymNyzhnykSystemAndItsLaxDefSubalgs}
(i) The subalgebra~$\mathfrak s_{\ref{eq:SymNyzhnykSystem}'}$ of the pseudoalgebra~$\mathfrak g_{\ref{eq:SymNyzhnykSystem}'}$ determines the diffeomorphisms stabilizing this pseudoalgebra.

(ii) The subalgebra~$\mathfrak s_{\ref{eq:LaxSymNyzhnykSystem}'}$ of the pseudoalgebra~$\mathfrak g_{\ref{eq:LaxSymNyzhnykSystem}'}$ determines the diffeomorphisms stabilizing this pseudoalgebra.
\end{corollary}

\section{Defining geometric properties}\label{sec:DefiningGeomProperties}

All third-order partial differential equations with three independent variables that admit the same Lie invariance pseudoalgebra as the dispersionless Nyzhnyk equation~\eqref{eq:dN} were first constructed in~\cite[Section~8]{boyk2024a}.
A collection of geometric properties uniquely defining this equation was also found therein.

\begin{lemma}\label{lem:dNEqsWithSameIA}
A partial differential equation of order at most three with three independent variables is invariant with respect to the pseudoalgebra~$\mathfrak g_{\ref{eq:dN}}$ if and only if it is of the form
\begin{gather}\label{eq:dNEqsWithSameIA}
\!\!\!\!u_{txy}=(u_{xx}u_{xy})_x+(u_{xy}u_{yy})_y+u_{xy}u_{xyy}H(\omega_1,\omega_2),\quad
\omega_1:=\frac{u_{xxx}-u_{yyy}}{u_{xyy}},\quad
\omega_2:=\frac{u_{xxy}}{u_{xyy}},
\end{gather}
where $H:=H(\omega_1,\omega_2)$ is an arbitrary smooth function of its arguments.
\end{lemma}

\begin{lemma}\label{lem:dNEqsWithSameIAandCLChars}
(i) An equation of the form~\eqref{eq:dNEqsWithSameIA} admits the conservation-law characteristic~$1$ and thus it is in a conserved form if and only if $H$ is an affine function of~$(\omega_1,\omega_2)$, i.e., $H=a\omega_1+b\omega_2+c$ for some constants~$a$, $b$ and~$c$, and the equation takes the form
\begin{gather}\label{eq:dNEqsWithSameIAandCLChar1}
u_{txy}=(u_{xx}u_{xy})_x+(u_{xy}u_{yy})_y+u_{xy}\big(a(u_{xxx}-u_{yyy})+bu_{xxy}+cu_{xyy}\big).
\end{gather}

(ii) An equation of the form~\eqref{eq:dNEqsWithSameIAandCLChar1} admits the conservation-law characteristic~$u_{xx}$ or~$u_{yy}$ if and only if $a=b=0$ or $a=c=0$, respectively.
\end{lemma}

\begin{theorem}\label{thm:dNDefiningGeometricProperties}
A partial differential equation of order $r\in\{1,2,3\}$ with three independent variables admits the pseudoalgebra~$\mathfrak g_{\ref{eq:dN}}$ as its Lie invariance algebra and the conservation-law characteristics~$1$, $u_{xx}$ and~$u_{yy}$ if and only if it coincides with the dispersionless Nyzhnyk equation~\eqref{eq:dN}.
\end{theorem}

\begin{remark}
The Nyzhnyk equation~\eqref{eq:SymPotNizhEq} has a significantly higher order (five instead of three) and is therefore defined on a jet space of a much higher dimension than its dispersionless counterpart~\eqref{eq:dN}.
Moreover, $\mathfrak g_{\ref{eq:SymPotNizhEq}}\subsetneq\mathfrak g_{\ref{eq:dN}}$.
Thus, to single out the Nyzhnyk equation~\eqref{eq:SymPotNizhEq} from among partial differential equations of the same order, one requires significantly more geometric properties in addition to invariance under~$\mathfrak g_{\ref{eq:SymPotNizhEq}}$ compared to Theorem~\ref{thm:dNDefiningGeometricProperties}, and it is obvious that these properties are not exhausted by simple conservation laws.
In view of these arguments, extending Theorem~\ref{thm:dNDefiningGeometricProperties} to asymmetric models, including the dispersive case, appears more realistic and promising.
\end{remark}

\section{One- and two-dimensional subalgebras}\label{sec:OneAndTwoDimSubalg}

One- and two-dimensional subalgebras of the maximal Lie invariance pseudoalgebra of the dispersionless Nyzhnyk equation~\eqref{eq:dN} and one-dimensional subalgebras of the maximal Lie invariance pseudoalgebra of its nonlinear Lax representation~\eqref{eq:dNLaxPair} were exhaustively classified in~\cite{vinn2024a}.
This enabled a complete classification of the Lie reductions of this equation and the construction of large families of new exact solutions.
We present these subalgebra classification results here in order to subsequently extend them to the Nyzhnyk equation~\eqref{eq:SymPotNizhEq}.

Throughout this section, the parameter functions~$\rho$, $\alpha$, $\beta$, $\tilde\rho$, $\beta^1$, $\beta^2$ and $\sigma$ run through the set of smooth functions of~$t$ with $\rho\ne0$.

\begin{lemma}\label{lem:dN1DInequivSubalgs}
A complete list of $G_{\ref{eq:dN}}$-inequivalent one-dimensional subalgebras of the pseudoalgebra~$\mathfrak g_{\ref{eq:dN}}$
consists of the following families of subalgebras:
\begin{gather*}
\mathfrak s_{1.1}^\delta=\big\langle D^t(1)+\delta D^{\rm s}\big\rangle,\quad
\mathfrak s_{1.2}       =\big\langle D^{\rm s}              \big\rangle,\quad
\mathfrak s_{1.3}^\rho  =\big\langle P^x(1)+P^y(\rho)       \big\rangle,\quad
\mathfrak s_{1.4}^\beta =\big\langle P^x(1)+R^y(\beta)      \big\rangle,\\
\mathfrak s_{1.5}^\beta =\big\langle R^x(1)+R^y(\beta)      \big\rangle,\quad
\mathfrak s_{1.6}       =\big\langle Z(t)                   \big\rangle,\quad
\mathfrak s_{1.7}       =\big\langle Z(1)                   \big\rangle,
\end{gather*}
where $\delta\in\{0,1\}$ $(\!{}\bmod G_{\ref{eq:dN}})$.
%, and $\rho$ and $\beta$ run through the set of smooth functions of~$t$ with $\rho\ne0$.
\end{lemma}

\begin{lemma}\label{lem:dN2DInequivSubalgs}
A complete list of $G_{\ref{eq:dN}}$-inequivalent two-dimensional subalgebras of the pseudoalgebra~$\mathfrak g_{\ref{eq:dN}}$
is constituted by the nonabelian algebras
\begin{gather*}
\mathfrak s_{2.1}^\lambda=\big\langle D^t(1),\,D^t(t)+\lambda D^{\rm s}\big\rangle,\quad
\mathfrak s_{2.2}^\nu    =\big\langle D^t(1),\,D^t(t)-\tfrac13 D^{\rm s}+P^x(1)+P^y(\nu)\big\rangle,\\
\mathfrak s_{2.3}^\nu    =\big\langle D^t(1),\,D^t(t)+\tfrac16 D^{\rm s}+R^x(1)+R^y(\nu)\big\rangle,\quad
\mathfrak s_{2.4}        =\big\langle D^t(1),\,D^t(t)+Z(1)\big\rangle,\\
\mathfrak s_{2.5}^{\lambda\mu}=\big\langle D^t(1)+\lambda D^{\rm s},\,P^x({\rm e}^{(\lambda-1)t})+\mu P^y({\rm e}^{(\lambda-1)t})\big\rangle,\\
\mathfrak s_{2.6}^{\lambda\delta}=\big\langle D^t(1)+\lambda D^{\rm s},\,P^x({\rm e}^{(\lambda-1)t})+\delta R^y({\rm e}^{(2\lambda-1)t})\big\rangle,\\
\mathfrak s_{2.7}^{\lambda\nu}=\big\langle D^t(1)+\lambda D^{\rm s},\,R^x({\rm e}^{(2\lambda-1)t})+\nu R^y({\rm e}^{(2\lambda-1)t})\big\rangle,\quad
\mathfrak s_{2.8}^\lambda     =\big\langle D^t(1)+\lambda D^{\rm s},\,Z({\rm e}^{(3\lambda-1)t})\big\rangle,\\
\mathfrak s_{2.9}^{\tilde\rho}=\big\langle D^{\rm s},\,P^x(1)+P^y(\tilde\rho)\big\rangle,\quad
\mathfrak s_{2.10}^\beta      =\big\langle D^{\rm s},\,R^x(1)+R^y(\beta)\big\rangle,\\
\mathfrak s_{2.11}            =\big\langle D^{\rm s},\,Z(t)\big\rangle,\quad
\mathfrak s_{2.12}            =\big\langle D^{\rm s},\,Z(1)\big\rangle,\quad
\end{gather*}
and the abelian algebras
\begin{gather*}
\mathfrak s_{2.13}       =\big\langle D^t(1),\,D^{\rm s}\big\rangle,\quad
\mathfrak s_{2.14}^{\delta\nu\delta'}=\big\langle D^t(1)+\delta D^{\rm s},\,
P^x({\rm e}^{\delta t})+\nu P^y({\rm e}^{\delta t})+\delta'R^y({\rm e}^{2\delta t})\big\rangle,
\\
\mathfrak s_{2.15}^{\delta\nu}=\big\langle D^t(1)+\delta D^{\rm s},\,R^x({\rm e}^{2\delta t})+\nu R^y({\rm e}^{2\delta t})\big\rangle,\quad
\mathfrak s_{2.16}^\delta     =\big\langle D^t(1)+\delta D^{\rm s},\,Z({\rm e}^{3\delta t}) \big\rangle,\\
\mathfrak s_{2.17}^{\rho\alpha\beta}=\big\langle P^x(1)+R^y(\beta),\,P^y(\rho)+R^x(\rho\beta)\big\rangle,\\
\mathfrak s_{2.18}^{\rho\beta\sigma}=\big\langle P^x(1)+P^y(\rho),\,-R^x(\rho\beta)+R^y(\beta)+Z(\sigma)\big\rangle_{(\beta,\sigma)\ne(0,0)},\\
\mathfrak s_{2.19}^{\beta^1\beta^2} =\big\langle P^x(1)+R^y(\beta^1),\,R^y(\beta^2)\big\rangle_{\beta^2\ne0},\quad
\mathfrak s_{2.20}^{\beta\sigma}    =\big\langle P^x(1)+R^y(\beta),\,Z(\sigma)\big\rangle_{\sigma\ne0},\\
\mathfrak s_{2.21}^{\alpha\beta^1\beta^2} =\big\langle R^x(1)+R^y(\beta^1),\,R^x(\alpha)+R^y(\beta^2)\big\rangle_{\beta^2\ne\alpha\beta^1},\\
\mathfrak s_{2.22}^{\alpha\beta\sigma}    =\big\langle R^x(1)+R^y(\beta),\,R^x(\alpha)+R^y(\alpha\beta)+Z(\sigma)\big\rangle_{\alpha_t\ne0},\\
\mathfrak s_{2.23}^{\beta\sigma}          =\big\langle R^x(1)+R^y(\beta),\,Z(\sigma)\big\rangle_{\sigma\ne0},\ \
\mathfrak s_{2.24}^\sigma =\big\langle Z(t),\,Z(\sigma)\big\rangle_{\sigma_{tt}\ne0},\ \
\mathfrak s_{2.25}^\sigma =\big\langle Z(1),\,Z(\sigma)\big\rangle_{\sigma_t\ne0},
\end{gather*}
where
%$\rho$, $\tilde\rho$, $\alpha$, $\beta$, $\beta^1$, $\beta^2$ and $\sigma$ run the set of smooth functions from~$t$ with $\rho\ne0$,
$\lambda\in\mathbb R$, $\mu\in[-1,1]\setminus\{0\}$ and $\nu\in[-1,1]$ $(\!{}\bmod G_{\ref{eq:dN}})$, $\delta,\delta'\in\{0,1\}$ $(\!{}\bmod G_{\ref{eq:dN}})$,
and in addition the conditions specified after the corresponding subalgebras hold.
\end{lemma}

\begin{lemma}\label{lem:dNLaxPair1DInequivSubalgs}
A complete list of $G_{\ref{eq:dNLaxPair}}$-inequivalent one-dimensional subalgebras of the pseudoalgebra~$\mathfrak g_{\ref{eq:dNLaxPair}}$
is exhausted by the following families of subalgebras:
\begin{gather*}
\bar{\mathfrak s}_{1.1}^{\delta\delta'}=\big\langle \bar D^t(1)+\delta\bar D^{\rm s}+\delta'\bar P^v\big\rangle_{\delta\delta'=0},\quad
\bar{\mathfrak s}_{1.2}                =\big\langle \bar D^{\rm s}                                  \big\rangle,\quad
\bar{\mathfrak s}_{1.3}^{\rho\delta}   =\big\langle \bar P^x(1)+\bar P^y(\rho)+\delta\bar P^v       \big\rangle,\\
\bar{\mathfrak s}_{1.4}^{\beta\delta}  =\big\langle \bar P^x(1)+\bar R^y(\beta)+\delta\bar P^v      \big\rangle,\quad
\bar{\mathfrak s}_{1.5}^{\beta\delta}  =\big\langle \bar R^x(1)+\bar R^y(\beta)+\delta\bar P^v      \big\rangle,\\
\bar{\mathfrak s}_{1.6}^{\delta}       =\big\langle \bar Z(t)+\delta\bar P^v                        \big\rangle,\quad
\bar{\mathfrak s}_{1.7}^{\delta}       =\big\langle \bar Z(1)+\delta\bar P^v                        \big\rangle,\quad
\bar{\mathfrak s}_{1.8}                =\big\langle \bar P^v                                        \big\rangle,
\end{gather*}
where $\delta,\delta'\in\{0,1\}$.
%, а $\rho$, $\beta$ run the set of smooth functions from~$t$ with $\rho\ne0$.
\end{lemma}

To classify the subalgebras of the pseudoalgebra~$\mathfrak g_{\ref{eq:SymPotNizhEq}'}$, we exploit the connection between the Lie symmetries of the equations~(\ref{eq:SymPotNizhEq}$'$) and~\eqref{eq:dN}.
Since the pseudoalgebra~$\mathfrak g_{\ref{eq:SymPotNizhEq}'}$ is a pseudosubalgebra of the pseudoalgebra~$\mathfrak g_{\ref{eq:dN}}$, classifying its one- and two-dimensional subalgebras is quite straightforward.
Analyzing the proofs of Lemmas~\ref{lem:dN1DInequivSubalgs} and~\ref{lem:dN2DInequivSubalgs} in~\cite[Lemmas~5 and~6]{vinn2024a} shows
that an exhaustive list of such inequivalent subalgebras is constituted precisely by the subalgebras of the pseudoalgebra~$\mathfrak g_{\ref{eq:dN}}$ collected in Lemmas~\ref{lem:dN1DInequivSubalgs} and~\ref{lem:dN2DInequivSubalgs}
whose natural projections to~$\langle D^{\rm s}\rangle$ are zero.
In addition, one should neglect the gauging of parameters in the basis elements of the subalgebras
by the adjoint action of the corresponding scaling transformation.
Analogously, based on Lemma~\ref{lem:dNLaxPair1DInequivSubalgs},
we can construct a complete list of inequivalent one-dimensional subalgebras of the pseudoalgebra~$\mathfrak g_{\ref{eq:LaxSymPotNizhEq}'}$.
To facilitate matching the classification lists of subalgebras, we preserve the numbering of the subalgebra families from Lemmas~\ref{lem:dN1DInequivSubalgs}, \ref{lem:dN2DInequivSubalgs} and~\ref{lem:dNLaxPair1DInequivSubalgs}.

\begin{lemma}\label{lem:SymPotNizhEq1DInequivSubalgs}
A complete list of $G_{\ref{eq:SymPotNizhEq}'}$-inequivalent one-dimensional subalgebras of the pseudoalgebra~$\mathfrak g_{\ref{eq:SymPotNizhEq}'}$
consists of the following families of subalgebras:
\begin{gather*}
\mathfrak s_{1.1}^0=\big\langle D^t(1)\big\rangle,\quad
\mathfrak s_{1.3}^\rho  =\big\langle P^x(1)+P^y(\rho)       \big\rangle,\quad
\mathfrak s_{1.4}^\beta =\big\langle P^x(1)+R^y(\beta)      \big\rangle,\\
\mathfrak s_{1.5}^\beta =\big\langle R^x(1)+R^y(\beta)      \big\rangle,\quad
\mathfrak s_{1.6}       =\big\langle Z(t)                   \big\rangle,\quad
\mathfrak s_{1.7}       =\big\langle Z(1)                   \big\rangle.
\end{gather*}
%where $\rho$ and $\beta$ run the set of smooth functions from~$t$ with $\rho\ne0$.
\end{lemma}

\begin{lemma}\label{lem:SymPotNizhEq2DInequivSubalgs}
A complete list of $G_{\ref{eq:SymPotNizhEq}'}$-inequivalent two-dimensional subalgebras of the pseudoalgebra~$\mathfrak g_{\ref{eq:SymPotNizhEq}'}$
is constituted by the nonabelian algebras
\begin{gather*}
\mathfrak s_{2.1}=\big\langle D^t(1),\,D^t(t)\big\rangle,\quad
\mathfrak s_{2.4}        =\big\langle D^t(1),\,D^t(t)+Z(1)\big\rangle,\\
\mathfrak s_{2.5}^{0\mu} =\big\langle D^t(1),\,P^x({\rm e}^{-t})+\mu P^y({\rm e}^{-t})\big\rangle,\quad
\mathfrak s_{2.6}^{0\delta}=\big\langle D^t(1),\,P^x({\rm e}^{-t})+\delta R^y({\rm e}^{-t})\big\rangle,\\
\mathfrak s_{2.7}^{0\nu}   =\big\langle D^t(1),\,R^x({\rm e}^{-t})+\nu R^y({\rm e}^{-t})\big\rangle,\quad
\mathfrak s_{2.8}^0        =\big\langle D^t(1),\,Z({\rm e}^{-t})\big\rangle,
\end{gather*}
and the abelian algebras
\begin{gather*}
\mathfrak s_{2.14}^{0\nu\delta'}=\big\langle D^t(1),\,
P^x(1)+\nu P^y(1)+\delta'R^y(1)\big\rangle,
\\
\mathfrak s_{2.15}^{0\nu}=\big\langle D^t(1),\,R^x(1)+\nu R^y(1)\big\rangle,\quad
\mathfrak s_{2.16}^0=\big\langle D^t(1),\,Z(1) \big\rangle,\\
\mathfrak s_{2.17}^{\rho\alpha\beta}=\big\langle P^x(1)+R^y(\beta),\,P^y(\rho)+R^x(\rho\beta)\big\rangle,\\
\mathfrak s_{2.18}^{\rho\beta\sigma}=\big\langle P^x(1)+P^y(\rho),\,-R^x(\rho\beta)+R^y(\beta)+Z(\sigma)\big\rangle_{(\beta,\sigma)\ne(0,0)},\\
\mathfrak s_{2.19}^{\beta^1\beta^2} =\big\langle P^x(1)+R^y(\beta^1),\,R^y(\beta^2)\big\rangle_{\beta^2\ne0},\quad
\mathfrak s_{2.20}^{\beta\sigma}    =\big\langle P^x(1)+R^y(\beta),\,Z(\sigma)\big\rangle_{\sigma\ne0},\\
\mathfrak s_{2.21}^{\alpha\beta^1\beta^2} =\big\langle R^x(1)+R^y(\beta^1),\,R^x(\alpha)+R^y(\beta^2)\big\rangle_{\beta^2\ne\alpha\beta^1},\\
\mathfrak s_{2.22}^{\alpha\beta\sigma}    =\big\langle R^x(1)+R^y(\beta),\,R^x(\alpha)+R^y(\alpha\beta)+Z(\sigma)\big\rangle_{\alpha_t\ne0},\\
\mathfrak s_{2.23}^{\beta\sigma}          =\big\langle R^x(1)+R^y(\beta),\,Z(\sigma)\big\rangle_{\sigma\ne0},\ \
\mathfrak s_{2.24}^\sigma =\big\langle Z(t),\,Z(\sigma)\big\rangle_{\sigma_{tt}\ne0},\ \
\mathfrak s_{2.25}^\sigma =\big\langle Z(1),\,Z(\sigma)\big\rangle_{\sigma_t\ne0},
\end{gather*}
where
%$\rho$, $\tilde\rho$, $\alpha$, $\beta$, $\beta^1$, $\beta^2$ and $\sigma$ run the set of smooth functions from~$t$ with $\rho\ne0$,,
$\mu\in[-1,1]\setminus\{0\}$ and $\nu\in[-1,1]$ $(\!{}\bmod G_{\ref{eq:SymPotNizhEq}'})$,
$\delta,\delta'\in\{0,1\}$ $(\!{}\bmod G_{\ref{eq:SymPotNizhEq}'})$,
and in addition the conditions specified after the corresponding subalgebras hold.
\end{lemma}

\begin{lemma}\label{lem:SymPotNizhEqLaxPair1DInequivSubalgs}
A complete list of $G_{\ref{eq:LaxSymPotNizhEq}'}$-inequivalent one-dimensional subalgebras of the pseudoalgebra~$\mathfrak g_{\ref{eq:LaxSymPotNizhEq}'}$ is exhausted by the following families of subalgebras:
\begin{gather*}
\bar{\mathfrak s}_{1.1}^{\delta\delta'}=\big\langle \bar D^t(1)+\delta'\bar D^{\rm\psi}\big\rangle,\quad
\bar{\mathfrak s}_{1.3}^{\rho\delta}   =\big\langle \bar P^x(1)+\bar P^y(\rho)+\delta\bar D^{\rm\psi}       \big\rangle,\\
\bar{\mathfrak s}_{1.4}^{\beta\delta}  =\big\langle \bar P^x(1)+\bar R^y(\beta)+\delta\bar D^{\rm\psi}      \big\rangle,\quad
\bar{\mathfrak s}_{1.5}^{\beta\delta}  =\big\langle \bar R^x(1)+\bar R^y(\beta)+\delta\bar D^{\rm\psi}      \big\rangle,\\
\bar{\mathfrak s}_{1.6}^{\delta}       =\big\langle \bar Z(t)+\delta\bar D^{\rm\psi}                        \big\rangle,\quad
\bar{\mathfrak s}_{1.7}^{\delta}       =\big\langle \bar Z(1)+\delta\bar D^{\rm\psi}                        \big\rangle,\quad
\bar{\mathfrak s}_{1.8}                =\big\langle \bar D^{\rm\psi}                                        \big\rangle,
\end{gather*}
where $\delta,\delta'\in\{0,1\}$.
\end{lemma}

\section{Conclusion}\label{sec:Conclusion}

The primary contribution of this paper is the comprehensive study of Nyzhnyk models across their full spectrum.
We have streamlined their nomenclature, identified relations between models previously treated as independent,
and carefully analyzed the connections between different kinds of Nyzhnyk models:
symmetric and asymmetric, dispersive and dispersionless, standard and modified ones,
real, complex, mixed, or specific models with complex conjugate variables, systems and single equations,
as well as their corresponding linear or nonlinear Lax representations.
This has allowed us to systematize the previously obtained results
on the symmetry properties of the (real symmetric potential) dispersionless Nyzhnyk equation,
its nonlinear Lax representation and the (real symmetric) dispersionless Nyzhnyk system
and to extend them fully or partially to other Nyzhnyk models.

We have computed the maximal Lie invariance pseudoalgebras for all the considered models
and identified patterns in the relations between these pseudoalgebras
according to the connections and transitions between the corresponding models.
All the identified patterns were expected,
and they have effectively served for additionally verifying and correcting the related results.
In other words, these pseudoalgebras are consistent with respect to different types of connections between the corresponding models.
In particular, the dispersionless limiting process leads to the extension of the initial maximal Lie invariance pseudoalgebra by one generating vector field associated with scaling transformations of the spatial and dependent variables.
(The same pattern can be observed, for example, in the dispersionless limiting process from the Kadomtsev--Petviashvili equation to the Khokhlov--Zabolotskaya equation.)
The transition from a dispersive or dispersionless model to its linear or nonlinear Lax representation reduces, in the context of Lie symmetries, to the prolongation of the pseudoalgebra elements to the pseudopotential as an additional dependent variable and the addition of one generating vector field associated with the scaling or translational gauge transformations of the pseudopotential, respectively.
Varying the base field or additionally constraining the independent or dependent variables as complex conjugates in the complex case merely requires a proper modification of the interpretation of the functional parameters in the generating elements of the corresponding pseudoalgebra via complexification or realification, while the form of these elements is formally preserved.
Thus, a consistent algebraic picture regarding the Lie symmetries of Nyzhnyk models has been obtained for the first time, reflecting the hierarchy of these models at the level of their symmetry structures. %\looseness=-1

Due to the similar structure of the maximal Lie invariance pseudoalgebras of standard symmetric Nyzhnyk models and the application of appropriate algebraic approaches, the results from~\cite{boyk2024a,vinn2024a} regarding some such dispersionless models, described in the introduction, were extended to other symmetric Nyzhnyk models after their analysis and deepening.
Namely, using a special modification of the megaideal-based version of the algebraic method from~\cite{malt2024a}, we have computed the point-symmetry pseudogroups of
the (symmetric potential dispersive) Nyzhnyk equation~(\ref{eq:SymPotNizhEq}$'$),
the (symmetric dispersive) Nyzhnyk system~(\ref{eq:SymNyzhnykSystem}$'$),
their Lax representations~(\ref{eq:LaxSymNyzhnykSystem}$'$) and~(\ref{eq:LaxSymPotNizhEq}$'$), as well as the nonlinear Lax representation~\eqref{eq:LaxdNSystem} of the (symmetric) dispersionless Nyzhnyk system~\eqref{eq:dNSystem}.
We have proved that the contact-symmetry pseudogroup of the (symmetric potential dispersive) Nyzhnyk equation~(\ref{eq:SymPotNizhEq}$'$), just as for its dispersionless counterpart, is the first prolongation of the point-symmetry pseudogroup of this equation.
In each obtained pseudogroup, we have singled out a complete set of independent discrete point symmetries, and
these sets exhibit more nontrivial behavior in the course of transitions between models than the structures formed by continuous symmetries.
The above series of computations demonstrated the particular efficiency of the algebraic method for finding the (pseudo)groups of point symmetries for a collection of models of different structure (e.g., simultaneously single equations and systems with different numbers of equations) whose maximal Lie invariance (pseudo)algebras, however, share a similar structure.
It is interesting that among all the symmetric Nyzhnyk models, only the (symmetric potential) dispersionless Nyzhnyk equation~\eqref{eq:dN} possesses the property that its point-symmetry pseudogroup coincides with the stabilizer of its maximal Lie invariance pseudoalgebra in the pseudogroup of local diffeomorphisms in the underlying space of independent and dependent variables.
One may conjecture that the asymmetric (potential) dispersionless Nyzhnyk equation~\eqref{eq:AsymdN} is equally unique with respect to this property among asymmetric Nyzhnyk models.
For the maximal Lie and contact invariance pseudoalgebras of the (symmetric potential) dispersionless Nyzhnyk equation~\eqref{eq:dN}, the following assertion was proven in~\cite{boyk2024a}: the (finite-dimensional) subalgebras of these pseudoalgebras that naturally arise when constructing the corresponding point-symmetry pseudogroups by the algebraic method determine the point diffeomorphisms stabilizing these pseudoalgebras.
In the present paper, we have extended this assertion to all other symmetric models.

The exhaustive classifications of one- and two-dimensional subalgebras of the maximal Lie invariance pseudoalgebra of the (symmetric) dispersive Nyzhnyk equation~(\ref{eq:SymPotNizhEq}$'$) and of one-dimen\-sional subalgebras of the maximal Lie invariance pseudoalgebra of its linear Lax representation~(\ref{eq:LaxSymPotNizhEq}$'$) have been performed.
These results serve as a natural complement to the known classifications for the dispersionless counterparts and establish the foundation for carrying out Lie reductions and constructing exact solutions, as was done in~\cite{vinn2024a} for the (real symmetric potential) dispersionless Nyzhnyk equation~\eqref{eq:dN}, following the optimized Lie reduction procedure described therein for the case of (1+2)-dimensional partial differential equations.\looseness=-1

A broader natural continuation of the symmetry analysis of Nyzhnyk models is, first and foremost, the further study of asymmetric models in a manner similar to the symmetric ones.
The asymmetric case differs substantially from the symmetric one both in the structure of the models themselves and in the structure of the associated maximal Lie invariance pseudoalgebras, and consequently, of the corresponding point-symmetry pseudogroups.
Therefore, the aforementioned results obtained for symmetric models cannot be directly extended to asymmetric models.
Such an extension requires separate consideration with all the necessary computations performed from scratch, and we have begun this study.
It will include the construction of point-symmetry pseudogroups of asymmetric models via the megaideal-based version of the algebraic method, the identification of defining subalgebras of their maximal Lie invariance pseudoalgebras, the determination of their defining geometric properties, an exhaustive classification of their Lie reductions, the study of their hidden Lie symmetries and other hidden symmetry-like objects, as well as the construction of exact invariant solutions.

In view of the integrability of the Nyzhnyk models, it is also important to further study generalized and nonlocal symmetries, cosymmetries, conservation laws, recursion operators and B\"acklund transformations both for the models themselves and (in the context of searching for hidden structures) for their submodels.
In addition to~\cite{vinn2026a}, in the literature there are other works on submodels of Nyzhnyk models within the framework of symmetry analysis of differential equations and integrability theory; however, they are restricted solely to stationary cases.
In particular, the recursion operator of the stationary Nyzhnyk system, that is, the submodel of the system~\eqref{eq:SymNyzhnykSystem} under the condition $u_t=v_t=w_t=0$, was found in~\cite{marv2003a} and used to generate nonlocal symmetries of this system.
In~\cite{fera1999c}, a B\"acklund transformation was constructed between the stationary Nyzhnyk system and its modified counterpart, which is a stationary submodel of the system~\eqref{eq:ModNyzhnykSystem2b}; see~\cite[Sections~9.7~and~9.8]{roge2002A}.
Therein, it was shown how these standard and modified stationary Nyzhnyk systems arise in the context of differential geometry.

In summary, this paper establishes the foundation and simultaneously outlines a clear plan of action for a more systematic study of Nyzhnyk models within the framework of the symmetry analysis of differential equations, where the structural similarity between models or the corresponding symmetry-like objects will play a significant role.

\section*{Acknowledgments}

%\noprint{

\begin{table}[!ht]
\begin{minipage}{20mm}
    \includegraphics[width=20mm]{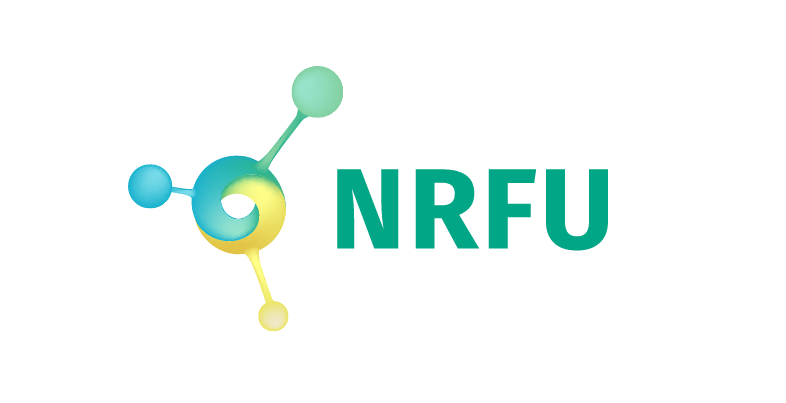}
\end{minipage}
\
\begin{minipage}{138mm}%}
V.M.B. expresses gratitude to the National Research Foundation of Ukraine, thanks to whose grant support the project 2025.07/0405 ``Algebraic methods for studying equations of mathematical physics'' is being implemented.
The contents of this work do not necessarily reflect the views of the National Research Foundation of Ukraine and are the sole responsibility of the Institute of Mathematics of NAS of Ukraine.\looseness=-1
\end{minipage}
\end{table}
\vspace{-4.2mm}
\noindent
This work was supported in part by a grant from the Simons Foundation (SFI-PD-Ukraine-00014586, O.O.V., V.M.B.) and in part by the Ministry of Education, Youth and Sports of the Czech Republic
(M\v SMT \v CR) under RVO funding for I\v C47813059 (R.O.P.).
The authors sincerely thank the anonymous reviewer for a number of helpful suggestions and comments,
which significantly improved the presentation of the results.
O.O.V. acknowledges the financial support
of the National Academy of Sciences of Ukraine within the framework of the project 0125U002856 for young scientists
and also thanks the Mathematical Institute of the Silesian University in Opava for the hospitality and support during scientific visits.
R.O.P.  expresses his gratitude for the hospitality shown by the University of Vienna during his long-term stay there.
%\noprint{
The authors express their deepest thanks to the Armed Forces of Ukraine and the civil Ukrainian people
for their bravery and courage in defense of peace and freedom in Europe and in the entire world from russism.
%}

\end{document}